\documentclass[a4paper,10pt]{article}
\usepackage{fullpage}

\usepackage{enumitem}
\usepackage{natbib}

\usepackage[toc,page,header]{appendix}
\usepackage{minitoc}

\usepackage[most]{tcolorbox}

\usepackage[margin=1in]{geometry}

\usepackage{tcolorbox}

\usepackage{bbm}

\usepackage{footnote}

\usepackage{enumitem}
\setlist[itemize]{leftmargin=*}

\usepackage[pdfencoding=auto, psdextra, pagebackref=true]{hyperref}
\hypersetup{
    colorlinks,
    linkcolor={blue!50!black},
    citecolor={blue!50!black},
    urlcolor={blue!80!black}
}

\renewcommand*{\backref}[1]{}

\renewcommand*{\backrefalt}[4]{%
  \ifcase #1%
  \or
    \space(Cited on p.~#2)%
  \else
    \space(Cited on pp.~#2)%
  \fi
}

\usepackage{url}
\usepackage{graphicx}
\usepackage{subcaption}

\usepackage{float}
\usepackage{algorithm}
\usepackage{algorithmic}
\usepackage{pdflscape}
\usepackage{multirow}
\usepackage{adjustbox}
\usepackage{dsfont}
\usepackage{tikz}
\usepackage{multicol}
\usepackage{comment}
\usepackage{array}

\usepackage{booktabs}
\usepackage{makecell}

\usepackage{amsthm,amsmath,amssymb,mathtools,bm}

\usepackage[capitalize,noabbrev]{cleveref}

\usepackage{xcolor}

\renewcommand{\eqref}[1]{(\ref{#1})}

\theoremstyle{plain}
\newtheorem{theorem}{Theorem}[section]
\newtheorem{lemma}{Lemma}[section]
\newtheorem{proposition}{Proposition}[section]
\newtheorem{corollary}{Corollary}[section]

\theoremstyle{definition}
\newtheorem{definition}{Definition}[section]
\newtheorem{assumption}{Assumption}[section]

\theoremstyle{remark}
\newtheorem{remark}{Remark}[section]

\usepackage{enumitem}

\DeclareMathOperator*{\argmin}{arg\,min}

\newcommand{\NE}{\operatorname{NE}}

\newcommand{\Reg}{\operatorname{Reg}}
\newcommand{\Gap}{\operatorname{Gap}}

\newcommand{\diam}{\operatorname{diam}}

\begin{document}

\title{Learning in Continuous Games\\ from  Pairwise Preference Feedback}

\author{Anas Barakat\footnote{Singapore University of Technology and Design.}}

\date{}

\maketitle

\begin{abstract}
We study learning in continuous games when players receive only pairwise
preference feedback, revealing which of two actions is preferred but neither
payoff values nor preference magnitudes. We first show that standard external
regret and coarse correlated equilibria (CCE) are not identifiable from this
ordinal information: the same sequence of play can incur zero and linear
regret in two ordinally equivalent games, while the distributions that remain
CCE across all cardinal representations consistent with the same preferences 
are exactly those supported on pure Nash equilibria. 
Motivated by this gap, we develop a first-order ordinal theory based on
normalized unilateral preference directions, introducing an ordinal
directional regret benchmark and corresponding equilibrium notions.
We show that block-normalized pseudogradient dynamics achieve sublinear ordinal regret and, under additional structure, Nash-convergence guarantees. We then use a single-comparison estimator to implement these dynamics from finite pairwise comparisons. With one comparison per player and round, the resulting algorithm achieves sublinear finite-resolution ordinal regret against arbitrary opponent behavior and, in ordinal potential games, almost-sure last-iterate convergence to the Nash set. Our results provide regret, dynamics, and equilibrium guarantees directly from preference feedback without reconstructing cardinal utilities.
\end{abstract}

\section{Introduction}

Preference feedback is especially useful for learning when numerical rewards are difficult to specify or observe, but alternatives can still be compared reliably. For instance, humans can compare large language model responses without assigning calibrated scores in reinforcement learning from human feedback \citep{christiano-et-al17deep-rl-human-pref,ouyang-et-al22llms-hf}; users reveal relative preferences over items in recommendation systems \citep{bobadilla-et-al13recommender-sys-survey,jannach-et-al16recommender-systems}; and preference-based robot learning uses comparisons between candidate behaviors or trajectories when numerical objectives are difficult to specify \citep{hejna-sadigh22fewshot}.
In continuous optimization, such pairwise comparisons can suffice to optimize an objective without observing its function values \citep{saha-et-al21dueling}.
 
In games, this distinction between ordinal preferences and numerical utilities becomes especially consequential.
A strategic game can be specified through the
players' preference relations over joint action profiles, with payoff
functions providing numerical representations of these preferences
\citep{osborne-rubinstein94}. 
In particular, best responses and pure Nash equilibria depend only on which actions each player prefers, holding the other players’ actions fixed.
However, much of existing theory on learning in games is formulated in cardinal terms: external regret aggregates numerical payoff differences and its vanishing is related to coarse correlated equilibria \citep{fudenberg-levine98,cesabianchi-lugosi06book,
nisan-et-al07agt-book}, while first-order and payoff-based learning rules use utility gradients or numerical payoff observations
(see, e.g., \citet{bravo-leslie-mertikopoulos18bandit-concave-games,
mertikopoulos-zhou19}). 
This creates a structural mismatch: strategic incentives are determined by ordinal preferences, while standard learning objectives and dynamics may depend on the particular numerical utilities used to represent them. This raises a basic question: what remains of learning-in-games theory when players observe only their preferences, rather than a particular numerical representation of them?

To isolate what can be learned from preferences alone, we study continuous
games under \emph{ordinal pairwise comparison feedback}. At a fixed opponents'
profile, each player may query which of two unilateral actions is preferred,
but receives no information about the magnitude of the corresponding utility
difference. This raises a question beyond payoff estimation: are the learning
objects themselves identified by the feedback?
Indeed, replacing player $i$'s cost $J_i$ by $\phi_i\circ J_i$ for any strictly
increasing $\phi_i$ preserves every unilateral comparison. A learner observing
only pairwise preferences therefore cannot distinguish between these numerical
representations. Yet central objects of standard learning theory need not be
preserved: such transformations can alter the pseudogradient field, the
external regret assigned to the same sequence of play, and the set of coarse
correlated equilibria. Thus two games that are indistinguishable from
preference feedback can induce different standard learning objects. This
leads to our central question:
\begin{quote}
\centering
\emph{Which learning objectives, dynamics, and equilibrium guarantees are
intrinsic to players' unilateral preferences, and can they be realized
efficiently using only pairwise comparisons?}
\end{quote}

Our starting point is that smooth unilateral preferences retain a natural
first-order geometry. Infinitesimal comparisons reveal the signs of
directional derivatives and therefore identify, whenever
$\nabla_iJ_i(x)$ is nonzero, its direction but not
its magnitude \citep{saha-et-al21dueling}. This suggests taking normalized unilateral preference
directions as the primitive first-order objects. They retain the directional
information preserved under smooth monotone reparameterizations while
discarding precisely the player-specific gradient magnitudes that comparison
feedback does not identify. This leads to representation-invariant regret
benchmarks, equilibrium notions, and learning dynamics, and provides learning 
targets that we subsequently estimate directly from finite comparisons.

A key distinction from stochastic ranking models is that our comparison oracle
is purely ordinal: conditional on which alternative is preferred, the feedback
contains no information about the magnitude of the utility difference.
Gap-dependent models such as Plackett--Luce may instead reveal cardinal
information through repeated observations
\citep{liu-et-al26iclr-ranking-feedback}.
Likewise, unlike finite context-ordinal games defined directly from rankings
\citep{gemp26nash-without-numbers}, our smooth latent-utility model exposes a
local first-order geometry. We use this geometry to define learning objectives,
dynamics, and equilibrium guarantees directly from unilateral preferences. 
See Section~\ref{subsec:related-work} for a discussion.

\subsection{Contributions}

We develop a first-order framework for learning in continuous games from
\emph{ordinal pairwise comparison feedback}: comparisons reveal only the
ordering of two unilateral actions and carry no information about the magnitude
of their cost difference. Our performance benchmarks, learning dynamics, and
equilibrium guarantees are designed directly from the first-order information that this preference feedback identifies.\\

\noindent\textbf{1. Identifiability limits of classical learning objectives.}
We show that ordinally equivalent games---i.e., games inducing exactly the
same unilateral comparison oracle---can assign the same sequence of play zero
external regret in one cardinal representation of the game and linear regret in another
(Proposition~\ref{prop:cardinal-regret-noninvariant}), and can have different
coarse correlated equilibrium (CCE) sets
(Proposition~\ref{prop:cce-not-ordinally-invar}). 
In contrast, ordinally equivalent games have identical best-response correspondences 
and hence the same pure Nash equilibria. More sharply, for every continuous game
$\mathcal G$ with compact action spaces,
$\bigcap_{\widetilde{\mathcal G}\sim_{\rm ord}\mathcal G}
\operatorname{CCE}(\widetilde{\mathcal G})
=\Delta(\operatorname{NE}(\mathcal G)),$ 
where $\sim_{\rm ord}$ denotes ordinal 
game equivalence and $\Delta(\operatorname{NE}(\mathcal G))$ denotes the distributions
supported on the set of Nash equilibria. Thus, the CCE distributions that remain valid under
every cardinal realization consistent with the same unilateral preferences are
exactly those supported on pure Nash equilibria
(Theorem~\ref{thm:cce-vs-ne}). Pairwise preference feedback therefore creates
not only an estimation problem, but an identifiability problem: classical
regret and CCE are not, in general, determined by the information available to
the players.\\

\noindent\textbf{2. A first-order ordinal regret-to-equilibrium theory.}
Motivated by this identifiability gap, we formulate learning directly on the
first-order preference geometry that remains invariant under regular smooth
monotone reparameterizations of the cost functions. We introduce
\emph{ordinal directional regret}, a first-order benchmark based on normalized
unilateral preference directions rather than cardinal payoff differences or
gradient magnitudes. This benchmark induces an \emph{ordinal coarse
variational equilibrium} residual and an exact no-regret-to-equilibrium
correspondence: sublinear ordinal regret for every player implies vanishing
ordinal equilibrium violation of the empirical distribution of play
(Section~\ref{subsec:ordinal-coarse-equilibrium}). We further introduce an
ordinal Nash residual whose zero set coincides with the pure Nash equilibria
under blockwise pseudoconvexity. Projected block-normalized pseudogradient
dynamics achieve $\mathcal O(\sqrt{T})$ ordinal directional regret for every
player against arbitrary opponent behavior
(Theorem~\ref{thm:discrete-regret}). 
Consequently, the empirical distribution of play achieves
an $\mathcal{O}(T^{-1/2})$ ordinal coarse variational equilibrium residual
(Corollary~\ref{cor:discrete-equilibrium}). 
Block-normalized pseudogradient dynamics retain the local preference
directions preserved under regular monotone reparameterizations while
discarding player-specific gradient magnitudes, thereby matching the geometry
of the learning dynamics to the information revealed by pairwise preferences.\\

\noindent\textbf{3. Direct learning from a single comparison and last-iterate
Nash convergence.}
For every comparison radius $\delta>0$, we identify an exact
finite-resolution vector field $D^\delta$ determined solely by the pairwise
comparison oracle. This vector field need not coincide with the infinitesimal
normalized preference field $D$, 
and its convergence to $D$ need not be
uniform near stationary points. We construct an unbiased single-comparison estimator of
$D^\delta$ and a comparison-normalized learning rule that does not reconstruct
a cardinal utility representation. With one comparison per player and round,
the resulting algorithm achieves $\mathcal O\!\left(\frac{R_i\sqrt{d_iT}}{\nu_i}\right)$ 
expected finite-resolution ordinal regret against arbitrary opponent behavior,
where $R_i$, $d_i$, and $\nu_i$ denote the action-set diameter, dimension, and
comparison reliability, respectively (Theorem~\ref{thm:comparison-regret}). 
In ordinal potential games, where players' comparisons coincide exactly with
comparisons of a latent potential, the \emph{last iterate} converges almost
surely to the Nash set; if the Nash set is finite, it converges almost surely
to a single Nash equilibrium (Theorem~\ref{thm:as-comparison-ne}). 

\subsection{Related Work}
\label{subsec:related-work}

We discuss the most closely related work here and defer a broader review to
Appendix~\ref{apdx:related-work}.\\

\noindent\textbf{Dueling bandits and comparison-based optimization.}
Preference-based feedback has been studied extensively in dueling bandits;
see \citet{bengs-et-al21survey-pref-bandits} for a survey.
In continuous optimization, pairwise comparisons can be used to optimize
objectives without observing function values
\citep{yue-joachims09dueling-bandits,jamieson-et-al12,
kumagai17continuous-dueling-bandit,saha-et-al21dueling}, with extensions to
more general preferences, adversarial settings, and richer ranking feedback
\citep{saha-et-al25,blum-et-al24dueling-adversary,
saha-et-al24multiway,tang-et-al24zo-ranking,ye-et-al25zo}.
Most of this literature concerns a single objective, often under convexity or
\textit{smooth} stochastic preference models. Our setting is instead a nonconvex
multi-player game with a purely ordinal sign oracle, which reveals preference
but no cost-gap magnitude. Recent work on online convex optimization with
dueling feedback recovers standard regret guarantees under smooth
gap-dependent comparison models \citep{lu-et-al26oco-dueling-feedback};
our oracle deliberately excludes such cardinal gap information.\\

\noindent\textbf{Learning from preference feedback in games.}
Several recent works study which aspects of games are determined by ordinal
rather than cardinal information. In finite games,
\citet{biggar-shames25pref-graphs} encode unilateral rankings through
preference graphs, while \citet{conitzer24complexity-ordinal-games} studies 
equilibria robust across compatible cardinal representations.
\citet{gemp26nash-without-numbers} instead study finite context-ordinal games
specified directly by rankings, using social-choice aggregation across
opponents' contexts to define mixed equilibrium and learning notions.
Our smooth setting assumes latent utility functions observed only ordinally,
which leaves accessible a local first-order preference geometry.
\citet{abbadi-laraki-mertikopoulos26icml} show that the combinatorial
preference structure of a finite game constrains, but does not generally
determine, the long-run behavior of no-regret dynamics.

Closer to our continuous game setting,
\citet{dindovs-mezzetti06geb-better-reply} study stochastic better-reply
learning in aggregative games, while \citet{cheung14} studies
pairwise-comparison evolutionary dynamics in continuous-strategy population
games.
\citet{bervoets-bravo-faure20} establish almost-sure Nash convergence for
payoff-based learning in several classes of continuous games, but their
updates use numerical payoff variations.
\citet{krupa-bemporad26} propose to learn generalized Nash equilibria (GNE) from pairwise preferences by fitting a
numerical surrogate game without providing convergence guarantees. 
In contrast, we work directly with the first-order field identified by
unilateral comparisons, without reconstructing cardinal utilities.
\citet{liu-et-al26iclr-ranking-feedback} ask when ranking feedback suffices to
recover standard external regret. Their Plackett--Luce model is gap-dependent,
so repeated rankings can contain cardinal information which is used to estimate numerical
utilities; our purely ordinal oracle does not. We therefore formulate regret, equilibrium, and learning
dynamics directly from the information preserved by unilateral preferences.\\

\noindent\textbf{Normalized first-order dynamics.}
In single-objective optimization, comparison and sign feedback can recover
directional information without observing function values
\citep{saha-et-al21dueling,cheng-et-al20sign-opt,cai-et-al22one-bit}.
While our comparison-based implementation in Section~\ref{sec:learning-from-comparison-feedback} 
builds on this, normalization here is not merely a device for gradient estimation: the
block-normalized field is the first-order object identified by unilateral
preferences, and induces the corresponding regret, equilibrium, and learning
dynamics in the game.
Related to our invariance viewpoint, \citet{gupta-et-al26coop-bargaining} use normalized utility gradient directions to obtain bargaining solutions invariant to monotone utility transformations; 
our setting is decentralized and noncooperative, uses
unilateral preference feedback, studies different block-normalized dynamics 
and provides non-asymptotic guarantees. 
\citet{poveda-krstic-basar22tac} study magnitude-dependent rescalings of
gradient play for Nash seeking; our block normalization is instead dictated by
ordinal invariance and admits a direct single-comparison implementation.

\begin{figure}[t]
    \includegraphics[width=0.94\linewidth]{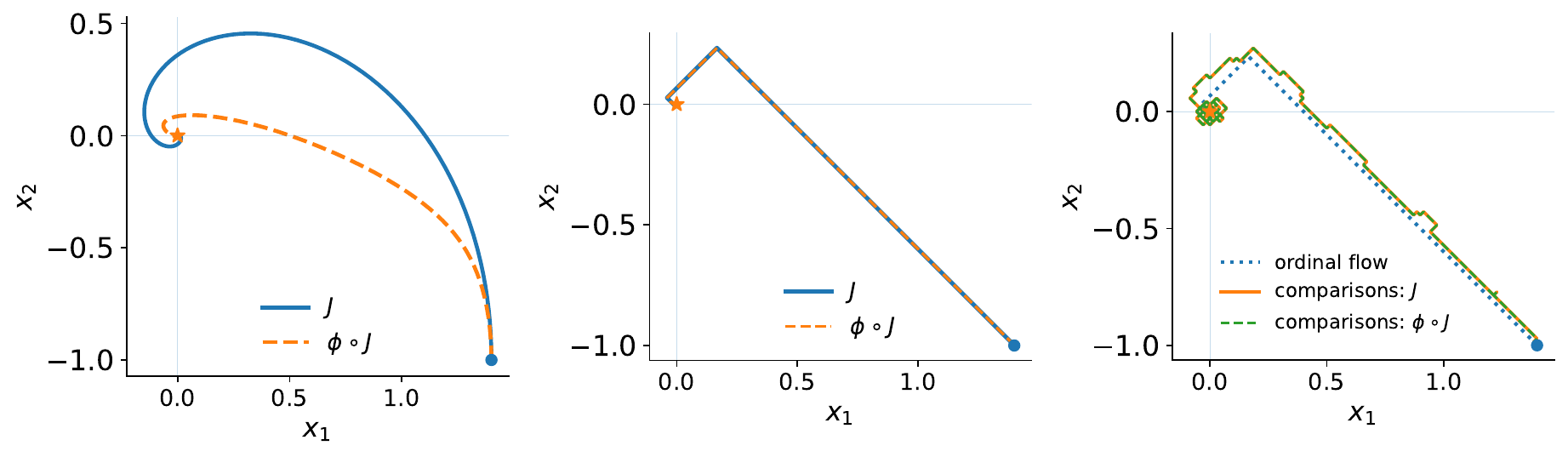}
    \caption{
    \textbf{Same preferences, different numerical representations} ($J$ and $\phi \circ J$). 
    The two games have identical unilateral preferences, best responses, and Nash
equilibrium, but different numerical utility representations.
\textbf{(Left)} Standard gradient play follows different trajectories.
\textbf{(Center)} Player-wise normalization makes the trajectories coincide.
\textbf{(Right)} Pairwise-comparison learning is likewise invariant under
common oracle randomness. The star denotes the common Nash. 
}
    \label{fig:ordinal-invariance}
\end{figure}

\section{Continuous Games with Pairwise Preference Feedback}
\label{sec:continuous-games-ordinal-feedback}

\noindent\textbf{Notation.} 
For an integer $N\geq1$, let $[N]:=\{1,\ldots,N\}$, and let
$[a]_+:=\max\{a,0\}$ denote the positive part of $a\in\mathbb R$.
For a nonempty closed convex set $\mathcal C\subseteq\mathbb R^d$,
$\Pi_{\mathcal C}$ denotes the Euclidean projection onto $\mathcal C$ and
$\operatorname{cl}(\mathcal C)$ its closure. 
The tangent and normal cones to the set $\mathcal C$ at a point $x \in\mathcal C$ are
$T_{\mathcal C}(x) :=
\operatorname{cl}
\left\{
\lambda(y-x):
y\in\mathcal C,\ \lambda\geq0
\right\}$ and 
$N_{\mathcal C}(x)
:=\left\{v:\langle v,y-x\rangle\leq0,\ \forall y\in\mathcal C\right\}.$
We write
$\operatorname{diam}(\mathcal C):=
\sup_{x,y\in\mathcal C}\|x-y\|_2\,,
\operatorname{dist}(x,\mathcal C)
:=\inf_{y\in\mathcal C}\|x-y\|_2$ 
where $\|\cdot\|_2$ is the Euclidean 2-norm. 
For a block vector
$v=(v_1,\ldots,v_N)$, with
$v_i\in\mathbb R^{d_i}$, we use
$\|v\|_{2,1}:=\sum_{i=1}^N\|v_i\|_2\,,
\|v\|_{2,\infty}:=\max_{i\in[N]}\|v_i\|_2.$
Finally,
$\mathbb S^{d-1}:=\{u\in\mathbb R^d:\|u\|_2=1\}$ and $\mathbb{B}_d$ denote the Euclidean unit
sphere and ball, $\Delta(\mathcal X)$ the set of probability distributions over
$\mathcal X$, and $\delta_x$ the Dirac measure at~$x$.\\

\noindent\textbf{Continuous games.}
We consider an $N$-player continuous game
$\mathcal G =\big((\mathcal X_i)_{i=1}^N,(J_i)_{i=1}^N\big),$ where the action sets $\mathcal{X}_i \subseteq\mathbb R^{d_i}$ are nonempty, closed, and convex sets and we define the joint action set $\mathcal X:=\prod_{i=1}^N\mathcal X_i\,.$ The cost functions $J_i: \mathcal{X} \to \mathbb{R}$ are supposed to be continuously differentiable on an open neighborhood of $\mathcal{X}$. Each player $i$ chooses an action
$x_i\in\mathcal X_i$ and incurs a cost
$J_i(x_i,x_{-i})$ where $x_{-i} \in \mathcal{X}_{-i} := \prod_{j \neq i}\mathcal{X}_j$ stands for the actions of all players but $i$ using the standard notation convention. When $\mathcal X_i$ is bounded, we write $R_i:=\operatorname{diam}(\mathcal X_i)$.
The best-response correspondence of player~$i$ is  
$\operatorname{BR}_i^{\mathcal G}(x_{-i}) :=
\argmin_{y_i\in\mathcal X_i}
J_i(y_i,x_{-i})$.
A profile $x^\star\in\mathcal X$ is a \emph{Nash equilibrium} if
$x_i^\star \in \operatorname{BR}_i^{\mathcal G}(x_{-i}^\star)\,, 
\forall i\in[N],$
and we denote the set of Nash equilibria by $\operatorname{NE}(\mathcal G)$.
A distribution $\mu\in\Delta(\mathcal X)$ is a
\emph{coarse correlated equilibrium} (CCE) if, for every $i\in[N]$ and every fixed deviation $y_i\in\mathcal X_i$,
$\mathbb E_{x\sim\mu}[J_i(x)] \leq
\mathbb E_{x\sim\mu}[J_i(y_i,x_{-i})].$
Equivalently,
$\sup_{y_i\in\mathcal X_i}
\mathbb E_{x\sim\mu}
\big[
J_i(x)-J_i(y_i,x_{-i})
\big]
\leq0\,, 
\forall i\in[N].$
We denote the corresponding set by $\operatorname{CCE}(\mathcal G)$. 
For any joint strategy $x \in \mathcal{X}$, we denote the standard pseudogradient field of the game by $F(x) = (\nabla_1 J_1(x), \cdots, \nabla_N J_N(x))\,,$ where for any $i \in [N], \nabla_i J_i(x)$ stands for the gradient of $J_i$ w.r.t. its $i$-th variable $x_i\,.$\\ 

\noindent\textbf{Pairwise comparison feedback.}
Player $i$ observes unilateral preferences via the comparison oracle:
\begin{equation}
\label{eq:comparison-oracle}
\mathsf C_i(a_i,b_i;x_{-i})
:=
\operatorname{sgn}
\left(
J_i(a_i,x_{-i})-J_i(b_i,x_{-i})
\right)
\in\{-1,0,+1\}\,,
\end{equation}
when querying the oracle with actions $a_i, b_i \in \mathcal{X}_i$ given fixed opponents' actions $x_{-i} \in \mathcal{X}_{-i}$.  
Thus $\mathsf C_i(a_i,b_i;x_{-i}) = -1$ means that $a_i$ is strictly preferred to $b_i$ under
the cost convention. 
We will also consider a noisy comparison oracle $Y_i = \xi_i\,\mathsf C_i(a_i,b_i;x_{-i}),$ 
where $\xi_i\in\{-1,+1\}$ satisfies
$\mathbb P(\xi_i=1)=\frac12+\nu_i,
\mathbb P(\xi_i=-1)=\frac12-\nu_i,
\nu_i\in(0,1/2].$
For a strict comparison, the reported sign is therefore correct with
probability $1/2+\nu_i$; exact comparisons correspond to $\nu_i=1/2$. 
The noise distribution depends only on the preference ordering and not on the numerical cost values. In the continuous optimization setting ($N = 1$), this model has been considered, e.g., in \cite{saha-et-al21dueling}.\\

\noindent\textbf{Unilateral ordinal equivalence.}
Two continuous games
$\mathcal G=((\mathcal X_i,J_i)_{i\in [N]})$ and
$\widetilde{\mathcal G}=((\mathcal X_i,\widetilde J_i)_{i\in [N]})$ with the same action
spaces are \emph{ordinally equivalent}, written
$\mathcal G\sim_{\mathrm{ord}}\widetilde{\mathcal G}$, if they induce the same comparisons: 
$\operatorname{sgn}
\left(
J_i(a_i,x_{-i})-J_i(b_i,x_{-i})
\right)
=
\operatorname{sgn}
\left(
\widetilde J_i(a_i,x_{-i})-\widetilde J_i(b_i,x_{-i})
\right)$ 
for any $i\in[N]$, $a_i,b_i\in\mathcal X_i$ and $x_{-i}\in\mathcal X_{-i}$. Equivalently, the two games generate exactly the
same unilateral comparison oracles. 
We denote the corresponding equivalence class by
$[\mathcal G]_{\mathrm{ord}}
:=
\{
\widetilde{\mathcal G}:
\widetilde{\mathcal G}\sim_{\mathrm{ord}}\mathcal G\}.$
We call $\widetilde{\mathcal G}$ a \emph{regular monotone
reparameterization} of $\mathcal G$ if
$\widetilde J_i=\phi_i \circ J_i$ (i.e., $\widetilde J_i(x)=\phi_i(J_i(x)), \forall x \in \mathcal{X}$), where $\phi_i: \mathbb{R} \to \mathbb{R}$ is continuously
differentiable and $\phi_i'(s)>0$ on a neighborhood of the range of~$J_i$.
Every such reparameterization is ordinally equivalent to $\mathcal G$.\\

\noindent\textbf{Ordinal potential games.}
A continuously differentiable function
$\Phi:\mathcal X\to\mathbb R$ is a \emph{smooth ordinal potential} for
$\mathcal G$ if 
$\operatorname{sgn}
\left(
J_i(y_i,x_{-i})-J_i(x_i,x_{-i})
\right)
=
\operatorname{sgn}
\left(
\Phi(y_i,x_{-i})-\Phi(x_i,x_{-i})
\right)$
for every player $i \in [N]$, every $x_{-i} \in \mathcal{X}_{-i}$, and every
$x_i,y_i\in\mathcal X_i$. 
A game admitting such a $\Phi$ is called a \emph{smooth ordinal potential
game}. Thus, for every fixed opponents' profile, the player's cost and the
potential induce exactly the same ordering over unilateral actions. 
The potential~$\Phi$ need not be observed. 

\section{Identifiability from Preference Feedback}
\label{sec:comparison-identification}

\subsection{Classical Regret and CCE Are Not Identifiable}
\label{subsec:external-regret-not-invariant}
For a sequence of joint actions $(x_t)_{t=1}^T$, the standard external regret of
player $i$ is defined as: 
\begin{equation}
\label{eq:external-regret}
\Reg_{i,T}^{J}
=
\sup_{y_i \in \mathcal{X}_i}\sum_{t=1}^T J_i(x_t)
- J_i(y_i,x_{-i,t}).
\end{equation}
This performance measure depends on cardinal utility
differences. Monotone reparameterizations can change not only its
scale, but whether a trajectory is no-regret at all, as we show next.

\begin{tcolorbox}[colframe=white!, top=2pt,left=2pt,right=2pt,bottom=2pt] 
\begin{proposition}[External regret is not ordinally identifiable]
\label{prop:cardinal-regret-noninvariant}
There exist ordinally equivalent smooth continuous games
$\mathcal G\sim_{\mathrm{ord}}\widetilde{\mathcal G}$ and a common sequence
of joint actions $(x_t)_{t\ge1}$ such that, for every even $T$,
$\Reg_{i,T}^{J}=0, \Reg_{i,T}^{\widetilde J}=\Omega(T).$
\end{proposition}
\end{tcolorbox}
Hence external regret is not determined by the ordinal equivalence class
identified by comparison-only feedback. The same is true of CCE.

\begin{tcolorbox}[colframe=white!, top=2pt,left=2pt,right=2pt,bottom=2pt] 
\begin{proposition}[CCE is not ordinally invariant]
\label{prop:cce-not-ordinally-invar}
There exist ordinally equivalent continuous games
$\mathcal G\sim_{\mathrm{ord}}\widetilde{\mathcal G}$ and a distribution
$\mu\in\Delta(\mathcal X)$ such that
$\mu\in\mathrm{CCE}(\mathcal G),
\mu\notin\mathrm{CCE}(\widetilde{\mathcal G}).$
\end{proposition}
\end{tcolorbox}

In contrast, ordinally equivalent games have identical best-response correspondences and therefore the same pure Nash equilibria. Thus pure Nash equilibrium is ordinal, whereas external regret and CCE generally are not. 
Because CCE averages utility differences across different opponents' profiles,
it is generally not preserved by nonlinear monotone transformations
\citep{hammond04expected-utility,hammond05utility-invariance,conitzer24complexity-ordinal-games,gemp26nash-without-numbers}.
We sharpen this observation by characterizing the CCE distributions that are
robust over the entire unilateral ordinal equivalence class. 

\begin{tcolorbox}[colframe=white!, top=2pt,left=2pt,right=2pt,bottom=2pt] 
\begin{theorem}[Ordinally robust CCE]
\label{thm:cce-vs-ne}
If $\mathcal{X}$ is compact, then $\bigcap_{\tilde{\mathcal{G}} \sim_{\mathrm{ord}} \mathcal{G}} \operatorname{CCE}(\tilde{\mathcal{G}}) = \Delta(\operatorname{NE}(\mathcal{G})).$
\end{theorem}
\end{tcolorbox}

Classical CCE permits positive conditional regret at some opponents'
profiles to be offset by negative regret elsewhere. Requiring the CCE
inequalities to hold for every cardinal realization consistent with the
same unilateral preferences eliminates this cross-context compensation and
forces best-response play almost surely. See Appendix~\ref{appx:proof-thm-cce-vs-ne} for a proof and a more detailed discussion. 
This motivates equilibrium notions defined directly from the ordinal information available to players. 

\subsection{Preference-Invariant First-Order Benchmarks}
\label{subsec:directional-reg-ordinal-coarse-var-eq}

Our starting point is to ask which first-order information is likewise intrinsic to the
comparison oracle. Under a smooth monotone
reparameterization $\widetilde J_i=\phi_i\circ J_i$,
$\nabla_i \widetilde{J}_i(x) = \phi_i'(J_i(x))\nabla_i J_i(x).$
Hence the magnitude $\|\nabla_iJ_i(x)\|$ depends on the cardinal representation,
whereas the ray $\{\lambda \nabla_i J_i(x):\lambda>0\}$ does not. This motivates us to consider the \emph{preference-direction field}, defined for every $i \in [N]$ and every joint strategy $x \in \mathcal{X}$ as follows: 
\begin{equation}
\label{eq:Di-pref-field}
D_i(x)
=
\begin{cases}
\dfrac{\nabla_i J_i(x)}{\|\nabla_i J_i(x)\|}, & \nabla_i J_i(x)\neq0,\\[1.5ex]
0, & \nabla_i J_i(x)=0,
\end{cases}
\qquad
D(x)= \begin{pmatrix} D_1(x) \\ \vdots \\ D_N(x)\end{pmatrix}. 
\end{equation}

\begin{tcolorbox}[colframe=white!, top=2pt,left=2pt,right=2pt,bottom=2pt] 
\begin{lemma}[Invariance of the preference-direction field]
\label{lem:D-invariance}
For every smooth player-wise monotone reparameterization
$\widetilde J_i=\phi_i\circ J_i$ with $\phi_i'>0$,
$\widetilde D_i(x)=D_i(x)$ for all $i \in [N],x \in \mathcal{X}$.
\end{lemma}
\end{tcolorbox}

Player-wise normalization is essential: a global normalization $F/\|F\|$
is not invariant because monotone reparameterizations rescale players'
gradient blocks differently.

When each $J_i(\cdot,x_{-i})$ is convex, standard external regret is bounded by its linearized counterpart 
$\sup_{y_i\in\mathcal X_i}\sum_{t=1}^T
\langle\nabla_iJ_i(x_t),x_{i,t}-y_i\rangle$.
Replacing $\nabla_iJ_i$ by the preference direction $D_i$ yields a
representation-invariant first-order benchmark, motivating the following
definition.
\begin{definition}[Ordinal directional regret]
\label{def:ordinal-regret}
The \emph{ordinal directional regret} of player $i \in [N]$ is
\begin{equation}
\label{eq:def-directional-regret}
\Reg^{D}_{i,T}
:=
\sup_{y_i\in\mathcal X_i}
\sum_{t=1}^T
\left\langle
D_i(x_t),x_{i,t}-y_i
\right\rangle .
\end{equation}
We say that player $i$ has no ordinal directional regret if
$\limsup_{T\to\infty}\frac{\Reg^D_{i,T}}{T}\le0.$
\end{definition}
Since $D$ is independent of the arbitrary magnitude of
the utility gradient, so is $\Reg^D_{i,T}$. 

\noindent\textbf{An ordinal Nash residual.}
The usual unilateral-deviation gap
$\sum_{i=1}^N \left[
J_i(x) - \inf_{y_i\in\mathcal X_i}J_i(y_i,x_{-i})
\right]$ 
has an ordinally invariant zero set but a representation-dependent numerical
value. For compact convex action sets, define instead
\begin{equation}
\label{eq:ordinal-nash-gap}
\Gap_D(x)
=
\sup_{y\in\mathcal X}
\langle D(x),x-y\rangle
=
\sum_{i=1}^N
\sup_{y_i\in\mathcal X_i}
\left\langle D_i(x),x_i-y_i\right\rangle .
\end{equation}
This is the standard Stampacchia variational inequality gap associated with the ordinal
pseudogradient~$D$. The gap is invariant under all player-wise smooth monotone utility reparameterizations with strictly positive derivative and satisfies $0\le\Gap_D(x) \le \sum_{i=1}^NR_i$ since $\|D_i(x)\|\leq1$.  
To relate its zero set to Nash equilibria, we impose the standard first-order sufficiency condition below.

\begin{assumption}[Blockwise pseudoconvexity]
\label{ass:block-pseudoconvex}
For every $i \in [N]$ and every $x_{-i} \in \mathcal{X}_{-i}$, the function
$z_i\longmapsto J_i(z_i,x_{-i})$ is pseudoconvex on $\mathcal X_i$, i.e., satisfies: 
For any $y_i \in \mathcal{X}_i, x = (x_i, x_{-i}) \in \mathcal{X}$,  
if $\left\langle \nabla_i J_i(x),y_i-x_i\right\rangle \geq 0$ 
then $J_i(y_i,x_{-i}) \geq J_i(x_i,x_{-i})$. 
\end{assumption} 

Pseudoconvexity is preserved by regular increasing reparameterizations and
ensures that unilateral first-order stationarity is globally optimal. 
Under this assumption, we obtain the following characterization of Nash equilibria. 

\begin{tcolorbox}[colframe=white!, top=2pt,left=2pt,right=2pt,bottom=2pt] 
\begin{proposition}[Ordinal Nash characterization]
\label{prop:ordinal-gap-ne}
Suppose the action sets are convex and
Assumption~\ref{ass:block-pseudoconvex} holds. Then
$\Gap_D(x)=0$ if and only if $x\in\NE(\mathcal G).$
\end{proposition}
\end{tcolorbox}

Hence $\Gap_D$ provides a representation-invariant first-order certificate of
pure Nash equilibrium. 

\subsection{From Ordinal Regret to Ordinal Coarse Equilibrium}
\label{subsec:ordinal-coarse-equilibrium}

Ordinal directional regret induces the following distributional analogue of
the classical no-regret-to-CCE correspondence. 

\begin{definition}[Ordinal coarse variational equilibrium]
For $\mu\in\Delta(\mathcal X)$, define  
\begin{equation}
\label{eq:def-ordinal-coarse-var-gap}
\mathcal E_D(\mu)
:=
\sum_{i=1}^N
\left[
\sup_{y_i\in\mathcal X_i}
\mathbb E_{x\sim\mu}
\left\langle D_i(x),x_i-y_i\right\rangle
\right]_+ .
\end{equation}
We call $\mu$ an \emph{ordinal coarse variational equilibrium} if
$\mathcal E_D(\mu)=0$.
\end{definition}

For a point mass,
$\mathcal E_D(\delta_x)=\Gap_D(x)$, so under
Assumption~\ref{ass:block-pseudoconvex},
$\mathcal E_D(\delta_x)=0$ iff $x\in\NE(\mathcal G)$ (see Proposition~\ref{prop:ordinal-gap-ne}). 
For the empirical distribution of play $\widehat\mu_T=T^{-1}\sum_{t=1}^T\delta_{x_t}$, it follows directly from the definitions of directional regret $\Reg^D_{i,T}$ in \eqref{eq:def-directional-regret} and $\mathcal E_D(\mu)$ in \eqref{eq:def-ordinal-coarse-var-gap} that: 
\begin{equation}
\label{eq:ordinal-coarse-var-gap-from-dir-reg}
\mathcal E_D(\widehat\mu_T)
=
\frac{1}{T}
\sum_{i=1}^N
[\Reg^D_{i,T}]_+ .
\end{equation}
Hence sublinear ordinal regret for each player implies
$\mathcal E_D(\widehat\mu_T)\to0$, i.e., 
the empirical distribution has asymptotically vanishing ordinal coarse equilibrium violation.

\section{Learning on First-Order Preference Geometry}
\label{sec:ordinal-learning-dyn}

We now study learning dynamics induced by the preference-direction field $D$.

\subsection{Continuous-Time Block-Normalized Dynamics}
\label{subsec:continuous-time-dyn}

For compact convex action sets, consider the projected \textit{block-normalized pseudogradient dynamics}: 
\begin{equation}
\label{eq:BNPGD}
\dot x_i(t)
=
\Pi_{T_{\mathcal X_i}(x_i(t))}
\!\left(-D_i(x(t))\right),
\qquad i\in[N].
\tag{BNPGD}
\end{equation}
At interior points this reduces to $\dot x_i=-D_i(x)$.
By Lemma~\ref{lem:D-invariance}, the dynamics are invariant under regular
player-wise monotone reparameterizations. Because $D$ may be discontinuous at
stationary blocks, our statements apply to absolutely continuous solutions of
\eqref{eq:BNPGD} satisfying the dynamics almost everywhere.
The continuous-time regret and ordinal-potential analyses are in
Appendix~\ref{apdx:cont-time-BNPGD}.

\subsection{Discrete-Time No-Ordinal-Regret Learning}
\label{subsec:discrete-time-ordinal-grad}

Inspired by \eqref{eq:BNPGD} descent, we introduce the corresponding projected discrete-time dynamics: 
\begin{equation}
\label{eq:bnpgd-discrete}
x_{i,t+1}
=
\Pi_{\mathcal X_i}
\left(
x_{i,t}-\eta_tD_i(x_t)
\right),
\qquad i\in[N]\,,
\end{equation}
starting from an initial point $x_0 \in \mathcal{X}.$ 
We assume exact access to $D_i(x_t)$ for now. Section~\ref{sec:learning-from-comparison-feedback}
shows how these directions can be learned from finite pairwise comparisons.

\begin{tcolorbox}[colframe=white!, top=2pt,left=2pt,right=2pt,bottom=2pt] 
\begin{theorem}[Directional regret of discrete block-normalized learning]
\label{thm:discrete-regret}
Suppose $\eta_t=\eta>0$. Then, the algorithm~\eqref{eq:bnpgd-discrete} incurs 
$\Reg^D_{i,T}
\leq \frac{R_i^2}{2\eta} +\frac{\eta T}{2}.$
With 
$\eta=R_i/\sqrt T$, 
$\Reg^D_{i,T} \leq R_i\sqrt T.$
\end{theorem}
\end{tcolorbox}

Combining Theorem~\ref{thm:discrete-regret} with
\eqref{eq:ordinal-coarse-var-gap-from-dir-reg} gives:

\begin{tcolorbox}[colframe=white!, top=2pt,left=2pt,right=2pt,bottom=2pt] 
\begin{corollary}[Convergence to coarse variational equilibria]
\label{cor:discrete-equilibrium}
Let $\widehat\mu_T=T^{-1}\sum_{t=1}^T\delta_{x_t}$. If every player runs
\eqref{eq:bnpgd-discrete} with
$\eta_i=R_i/\sqrt T$, then
$\mathcal E_D(\widehat\mu_T)
\leq
\frac{1}{\sqrt T}
\sum_{i=1}^NR_i.$
\end{corollary}
\end{tcolorbox}
This shows that the ordinal coarse equilibrium residual of empirical play decays as $\mathcal{O}(T^{-1/2})$. 

\section{Learning from Single Pairwise Comparisons}
\label{sec:learning-from-comparison-feedback} 

We now replace direct access to the preference-direction field $D$ by pairwise comparisons.\\

\noindent\textbf{Feasible comparison queries.}
Throughout this section, we assume that each $\mathcal X_i$ is compact,
convex, and has nonempty interior.
To ensure feasible symmetric queries, 
fix $c_i\in\operatorname{int}(\mathcal X_i)$
and $\rho_i>0$ such that
$c_i+\rho_i\mathbb B_{d_i}\subseteq\mathcal X_i,$ 
and let $\rho:=\min_{i\in[N]}\rho_i$.
For $0<\delta<\rho_i$, define for any $z_i\in\mathcal X_i,$
$\mathsf Q_{i,\delta}(z_i) :=
c_i+
\left(1-\frac{\delta}{\rho_i}\right)(z_i-c_i).$  
Then, 
$\mathsf Q_{i,\delta}(z_i)\pm\delta u 
\in\mathcal X_i$ and 
$\|\mathsf Q_{i,\delta}(z_i)-z_i\|
=
\frac{\delta}{\rho_i}\|z_i-c_i\|
\le
\frac{R_i}{\rho_i}\delta$ for any $u\in\mathbb S^{d_i-1}$, 
we write
$\mathcal X_i^\delta:=\mathsf Q_{i,\delta}(\mathcal X_i)$ and 
$\mathcal X^\delta:=\prod_{i=1}^N\mathcal X_i^\delta.$ 

\subsection{The Finite-Resolution Comparison Vector Field}
\label{subsec:directional-reg-bound}
Let $U_i\sim\operatorname{Unif}(\mathbb S^{d_i-1})$. For any radius
$\delta>0$ and any $x \in \mathcal{X}^{\delta}$, define: 
\begin{equation}
\label{eq:central-comparison}
S_i^\delta(x,U_i)
:=
\mathsf C_i
\left(
x_i+\delta U_i,\,
x_i-\delta U_i;\,
x_{-i}
\right) =
\operatorname{sgn}\!\left(
J_i(x_i+\delta U_i,x_{-i})
-
J_i(x_i-\delta U_i,x_{-i})
\right).
\end{equation}
Differentiability gives $
J_i(x_i+\delta U_i,x_{-i})
-
J_i(x_i-\delta U_i,x_{-i})
=
2\delta
\langle\nabla_iJ_i(x),U_i\rangle
+o(\delta),$
so infinitesimal comparisons reveal the sign of
$\langle D_i(x),U_i\rangle$, but not the magnitude of the gradient.
Define the \emph{$\delta$-comparison vector field} 
induced by the comparison oracle at finite resolution $\delta > 0$ as follows: 
\begin{equation}
\label{eq:delta-comparison-dir}
D_i^\delta(x)
:=
\frac{1}{\kappa_{d_i}}
\mathbb E_{U_i}
\left[
S_i^\delta(x,U_i)U_i
\right], \quad \forall x \in \mathcal{X}^{\delta}\,,
\end{equation}
where
$\kappa_d :=
\mathbb E_{U\sim\operatorname{Unif}(\mathbb S^{d-1})}|U_1|$ and $U_1$ is the 
first coordinate of the random variable $U$. We collect some of the 
properties of $D_i^{\delta}$ in the following result. 

\begin{tcolorbox}[colframe=white!, top=2pt,left=2pt,right=2pt,bottom=2pt] 
\begin{proposition}[Properties of the comparison field]
\label{prop:comparison-field}
For any $i \in [N]$ and every $\delta>0$,
(i) $D_i^\delta$ is invariant under every strictly increasing transformation
$J_i\mapsto\phi_i\circ J_i$;
(ii) $\|D_i^\delta(x)\|\leq1$ for every $x\in\mathcal{X}^{\delta}$; 
(iii) for any $x \in \mathcal{X}$ with $x_i \in \operatorname{int}(\mathcal{X}_i)$ and $\nabla_iJ_i(x)\neq0$, $D_i^\delta(x)$ is well-defined for all sufficiently small $\delta > 0$, and 
$\lim_{\delta\downarrow0}D_i^\delta(x)=D_i(x)$.
\end{proposition}
\end{tcolorbox}
Hence $D^\delta$ is comparison-invariant and converges pointwise to
$D$ away from stationary blocks.\\

\noindent\textbf{Comparison-normalized gradient play.} We now use the finite-resolution comparison vector field~$D^{\delta}$ and its expectation expression~\eqref{eq:delta-comparison-dir} to motivate the stochastic algorithm we consider, which uses only pairwise stochastic comparisons.
At round $t$, each player $i$ samples
$U_{i,t}\sim\operatorname{Unif}(\mathbb S^{d_i-1})$ and queries the symmetric
comparison~\eqref{eq:central-comparison} at radius $\delta_t > 0$.
Under the noisy comparison model of
Section~\ref{sec:continuous-games-ordinal-feedback}, the observed stochastic sign is
$Y_{i,t} =\xi_{i,t}S_i^{\delta_t}(x_t,U_{i,t}),$
with exact comparisons corresponding to $\nu_i=1/2$. We define the associated
single-comparison estimator $\widehat D_{i,t}
:= \frac{Y_{i,t}}{2\nu_i\kappa_{d_i}}U_{i,t}.$ 
This random-direction construction is standard in comparison-based
optimization \citep{saha-et-al21dueling,cai-et-al22one-bit}; here it is
applied blockwise to a coupled game and evaluated through ordinal regret. 

\begin{tcolorbox}[colframe=white!, top=2pt,left=2pt,right=2pt,bottom=2pt] 
\begin{lemma}[Single-comparison estimator]
\label{lem:one-bit-estimator}
Let $\mathcal F_t = \sigma(\xi_0, \dots, \xi_{t-1}, U_0, \dots, U_{t-1})$ be the $\sigma$-algebra induced by the random uniform directions and comparison noise up to round $t-1$ for $t \geq 1$. In particular,
$z_t$ and
$x_t=\mathsf Q_{\delta_t}(z_t)$ are $\mathcal F_t$-measurable, and
$x_t\in\mathcal X^{\delta_t}$ by construction.
Conditionally on $\mathcal F_t$ and the current action profile
$x_t$,
$\mathbb E[\widehat D_{i,t}\mid\mathcal F_t]
= D_i^{\delta_t}(x_t),$
and $\mathbb E[\|\widehat D_{i,t}\|^2\mid\mathcal F_t] \leq G_i^2,$ 
where $G_i= \frac{1}{2\nu_i\kappa_{d_i}}
= \Theta\!\left(\frac{\sqrt{d_i}}{\nu_i}\right).$ 
\end{lemma}
\end{tcolorbox}

To keep all comparison queries feasible, each player maintains a
reference iterate $z_{i,t}\in\mathcal X_i$ and plays the contracted
action $x_{i,t}
=
\mathsf Q_{i,\delta_t}(z_{i,t})
\in\mathcal X_i^{\delta_t}\,.$ 
Player $i$ then samples
$U_{i,t}\sim\operatorname{Unif}(\mathbb S^{d_i-1})$ and queries
$(x_{i,t}+\delta_tU_{i,t},x_{-i,t})$ versus 
$(x_{i,t}-\delta_tU_{i,t},x_{-i,t})$, 
which are feasible by construction. From the observed sign $Y_{i,t}$,
the player updates the reference iterate as follows: 
\begin{equation}
\label{eq:cngp}
z_{i,t+1}
=
\Pi_{\mathcal X_i}
\left(
z_{i,t}-\eta_{i,t}\widehat D_{i,t}
\right)\,, \quad \widehat D_{i,t}
=
\frac{Y_{i,t}}{2\nu_i\kappa_{d_i}}U_{i,t},
\tag{CNGP} 
\end{equation}
We call this procedure comparison-normalized gradient play (CNGP).
The profile $x_t$ is the action profile actually played and used as the
center of the comparison queries, while $z_t$ is the internal reference
iterate. Note that $\|x_{i,t}-z_{i,t}\| \le \frac{R_i}{\rho_i}\delta_t$, 
so the two sequences become asymptotically indistinguishable when
$\delta_t\to0$. Each round still uses one pairwise comparison per player,
and all players' queries can be performed in parallel. 
The full pseudocode is provided in Appendix~\ref{app:cngp-pseudocode}.

\subsection{Single-Comparison Ordinal Regret Guarantees}

Define the finite-resolution regret and expected pseudo-regret, respectively, as follows: 
\begin{equation}
\label{eq:finite-res-reg}
\Reg_{i,T}^{\delta}
:=
\sup_{y_i\in\mathcal X_i}
\sum_{t=1}^T
\left\langle
D_i^{\delta_t}(x_t),
x_{i,t}-y_i
\right\rangle,\, 
\overline{\Reg}_{i,T}^{\delta}
:=
\sup_{y_i\in\mathcal X_i}
\mathbb E
\left[
\sum_{t=1}^T
\left\langle
D_i^{\delta_t}(x_t),
x_{i,t}-y_i
\right\rangle
\right]\,.
\end{equation}
The benchmark is evaluated at the played actions $x_t$; the reference
iterates $z_t$ enter only the implementation of \eqref{eq:cngp}. 
Both quantities are invariant under arbitrary strictly increasing transformations of the players' utilities since $D_i^{\delta}$ is invariant (by Proposition~\ref{prop:comparison-field}). 

\begin{tcolorbox}[colframe=white!, top=2pt,left=2pt,right=2pt,bottom=2pt] 
\begin{theorem}[Directional regret from one comparison per round]
\label{thm:comparison-regret}
Let player $i \in [N]$ run \eqref{eq:cngp} with constant stepsize $\eta_i>0$ and radii
$0<\delta_t<\rho$. Then
$
\overline{\Reg}_{i,T}^{\delta}
\le
\frac{R_i^2}{2\eta_i}
+
\frac{\eta_iG_i^2T}{2}
+
\frac{R_i}{\rho_i}\sum_{t=1}^T\delta_t
$ for any $T \geq 1\,.$ If in addition  
$\eta_i=\frac{R_i}{G_i\sqrt T},
\delta_t=\frac{\delta}{\sqrt T},
0<\delta<\rho,$ 
then
\begin{equation}
\overline{\Reg}_{i,T}^{\delta}
\leq 
R_i(G_i+1)\sqrt T 
=
O\left(
\frac{R_i\sqrt{d_iT}}{\nu_i}
\right), \quad 
\mathbb E[\Reg_{i,T}^{\delta}]  
\leq 
3R_iG_i\sqrt T 
=
O\left(
\frac{R_i\sqrt{d_iT}}{\nu_i}
\right)\,.
\end{equation}
Furthermore, 
$\mathbb E\bigl[[\Reg_{i,T}^{\delta}]_+\bigr] \leq 4R_iG_i\sqrt T$.
\end{theorem}
\end{tcolorbox}

The guarantee requires no monotonicity, potential structure, or assumptions on opponents; the $\sqrt{d_i}$ factor is the cost of one random comparison per round (minibatching is discussed in Appendix~\ref{subsubsec:trading-comp-for-dim}). Importantly, $D^\delta(x)\to D(x)$ only pointwise away from stationary blocks, so vanishing resolution does not in general imply sublinear infinitesimal regret; see Proposition~\ref{prop:finite-vs-infinitesimal-regret}. This motivates evaluating finite-query learning using the exact oracle-identified field $D^\delta$. 

\subsection{Last-Iterate Nash Convergence in Ordinal Potential Games}
\label{subsec:as-cvg-Nash-ordinal-pot}

In this section, we suppose that the game has an additional ordinal potential structure 
to prove convergence to Nash equilibria. Under this structure, 
although the potential $\Phi$ is never observed, every unilateral comparison generated by a player
coincides with the corresponding comparison of~$\Phi$: 
for every $i\in[N]$, every $\delta\in(0,\rho)$,
$x\in\mathcal X^\delta$, and
$u\in\mathbb S^{d_i-1}$,
\[
\operatorname{sgn}
\left(
J_i(x_i+\delta u,x_{-i})
-
J_i(x_i-\delta u,x_{-i})
\right)
=
\operatorname{sgn}
\left(
\Phi(x_i+\delta u,x_{-i})
-
\Phi(x_i-\delta u,x_{-i})
\right).
\]
Consequently,
$D_{i,J_i}^\delta(x)=D_{i,\Phi}^\delta(x)$ for any $x\in\mathcal X^\delta$. 
Hence the comparison signal used by CNGP is exactly the finite-resolution comparison field of a common ordinal potential, 
even though the potential itself is never observed.

\begin{assumption}[Nondegeneracy]
\label{as:non-degeneracy-ord-pot}
For every $i\in[N]$ and $x\in\mathcal X$,
$\nabla_i J_i(x)=0
\Longleftrightarrow
\nabla_i\Phi(x)=0.$
\end{assumption}
This condition rules out ordinal representations that preserve all unilateral
comparisons while introducing or removing stationary points. Together with
the ordinal-potential property, it ensures that the players' and potential's
first-order preference directions are positively aligned whenever nonzero.

\begin{assumption}[Robbins-Monro stepsizes and vanishing resolution]
\label{as:rm-stepsizes} The stepsizes $(\eta_t)$ and the resolution $(\delta_t)$ satisfy: 
$\sum_{t=1}^\infty\eta_t=\infty, \sum_{t=1}^\infty\eta_t^2<\infty$ 
with $\eta_t>0, \delta_t \in (0,\rho), \forall t \in \mathbb{N}$ and $\delta_t\to0.$
\end{assumption}

\begin{tcolorbox}[colframe=white!, top=2pt,left=2pt,right=2pt,bottom=2pt] 
\begin{theorem}[Almost-sure last-iterate convergence to Nash equilibria]
\label{thm:as-comparison-ne}
Let $\mathcal G$ be an ordinal potential game with ordinal potential $\Phi$, and
suppose that each $\mathcal X_i$ is compact and convex with nonempty interior. 
Assume that the costs $J_i$ and the potential $\Phi$ are continuously differentiable on an open neighborhood of $\mathcal{X}$ on which $\Phi$ is $L$-smooth.
Let Assumptions~\ref{ass:block-pseudoconvex} and
\ref{as:non-degeneracy-ord-pot} hold, and suppose that $\operatorname{int}_{\mathbb R}
\bigl(\Phi(\operatorname{NE}(\mathcal G))\bigr)
=
\varnothing.$ 
Let $(x_t)$ be generated by \eqref{eq:cngp}, 
where $\widehat D_{i,t}$ is defined in
\eqref{eq:cngp} under the gap-independent noisy
comparison model with $\nu_i\in(0,1/2]$.
Suppose that the stepsizes and comparison radii satisfy
Assumption~\ref{as:rm-stepsizes}. Then, with probability one,
$\operatorname{dist}
\bigl(x_t,\operatorname{NE}(\mathcal G)\bigr)
\longrightarrow 0$ and $\Phi(x_t)$ converges.  
If $\operatorname{NE}(\mathcal G)$ is finite, then there exists an
$\operatorname{NE}(\mathcal G)$-valued random variable $X_\infty$ such that
$x_t\to X_\infty$ almost surely. 
In particular, if the Nash equilibrium is unique, the iterates converge to it
almost surely.
\end{theorem}
\end{tcolorbox}

A convenient admissible choice is
$\eta_t=\frac{\eta}{t^p}, 
\delta_t=\frac{\delta}{t^q},$
with $1/2<p\leq1$, $q>0$ 
and $\delta \in (0,\rho).$ 

\section{Conclusion and Future Work}

We developed an ordinal framework for learning in continuous games from
pairwise preference feedback, aligning the information available to players
with the regret benchmark, learning dynamics, and equilibrium guarantees.
This perspective opens several directions, including richer
preference-feedback models and extensions to dynamical games.

\subsection*{AI use statement}

In this work, ChatGPT was used to assist with literature search and identification of relevant references, drafting and editing portions of the manuscript, code generation for the illustrative figure, and technical proof development. In particular, ChatGPT was used interactively to explore proof strategies, refine technical arguments, and assist in the development of proofs for some of the main results.

The research questions, conceptual framing, main mathematical definitions, theorem statements, and decisions about which results to pursue were developed
by the author, with AI tools used for brainstorming. 
AI-generated proof suggestions were subsequently selected, independently checked, completed, revised, and integrated by the author, and references identified with AI assistance were verified against the original sources. The author takes full responsibility for the final content of the manuscript. 

\bibliographystyle{apalike}
\bibliography{references}


\newpage
\appendix 

\tableofcontents

\newpage
\section{Extended Related Work Discussion}
\label{apdx:related-work}

\noindent\textbf{Comparison-based optimization and ranking feedback.}
Preference-based online learning has a long history in dueling bandits; see
\citet{bengs-et-al21survey-pref-bandits} for a survey. In continuous
single-objective optimization, comparison queries can replace numerical
function evaluations and still provide useful directional information
\citep{jamieson-et-al12,kumagai17continuous-dueling-bandit,
saha-et-al21dueling,cai-et-al22one-bit}.
Subsequent work considers richer preference models, adversarial comparisons,
multiway feedback, and ranking oracles
\citep{blum-et-al24dueling-adversary,saha-et-al24multiway,
saha-et-al25,tang-et-al24zo-ranking,ye-et-al25zo,
chen-et-al25compo,lu-et-al26oco-dueling-feedback}.
\cite{maran-et-al24bandits-ranking} study finite-armed bandits with ranking feedback, where the learner observes the ranking of empirical rewards of previously pulled arms rather than their numerical values. They obtain sublinear regret in stochastic settings but show that standard sublinear regret is impossible under adversarial rewards. Their feedback model differs from ours in that rankings aggregate past arm rewards, whereas our oracle provides contemporaneous unilateral comparisons in a continuous game. 
Recent work also studies nonconvex optimization from comparisons:
\citet{wang-et-al26fosp-comp} give comparison-query guarantees for reaching
first-order stationary points of a smooth objective, while
\citet{scheinberg-xiong26} take a representation-free perspective in which
the preference relation itself is primitive and formulate optimality and
stationarity certificates through the geometry of preference level sets.
This distinction is relevant here because quantitative gradient-norm
criteria depend on the chosen cardinal representation, whereas our goal is
to define learning objects determined by comparisons themselves.
\citet{krupa-bemporad26} learn generalized Nash equilibria from pairwise preferences by fitting parametric surrogate objectives from preference data and repeatedly solving the learned GNEP. In contrast, we do not reconstruct a numerical game: we identify an ordinal first-order field directly from the comparison oracle and establish regret and almost-sure last-iterate convergence guarantees for the resulting comparison-based dynamics. 
\citet{liu-et-al26iclr-ranking-feedback} study standard external
regret from ranking observations. For instantaneous-utility rankings they show that sublinear external regret is impossible in general, and obtain positive results under additional temporal structure; their stochastic Plackett--Luce model is gap-dependent, so repeated rankings can be used to
estimate numerical utilities. In contrast, our comparison noise is
gap-independent: conditional on the underlying ordering, repeated
observations contain no information about the magnitude of the utility gap. 

\medskip
\noindent\textbf{Learning from human preferences.}
Pairwise and ranking feedback is central to reinforcement learning from human
feedback
\citep{christiano-et-al17deep-rl-human-pref,
ouyang-et-al22llms-hf,zhu-jordan-jiao23rlhf}.
Beyond the standard reward-modeling pipeline,
\citet{munos-et-al24nash-learning-from-hf} learn a pairwise preference model
and formulate policy optimization through a regularized Nash equilibrium;
other recent work studies best-policy identification from trajectory
preferences \citep{agnihotri-et-al26}, preference-based multi-agent
reinforcement learning \citep{zhu-et-al24}, and ordinal formulations of
reinforcement learning \citep{gemp-et-al26rl-social-choice}.
The information contained in the feedback itself is also an active concern:
\citet{feng-et-al25}, for example, augment ordinary trajectory comparisons
with queries designed to elicit preference strength. Our setting isolates
the opposite extreme---binary feedback containing only the ordering---and
asks what learning and equilibrium guarantees remain identifiable without
recovering preference magnitudes. 
Recent related work on preference learning \citep{pukdee-balcan-ravikumar26} asks what latent preference information is recoverable from pairwise-comparison data under models such as the Bradley–Terry model. Our identifiability question is instead game-theoretic: rather than recovering a numerical representation, we ask which regret, equilibrium, and learning objects are invariant across all representations compatible with the observed unilateral ordering.

\medskip
\noindent\textbf{Ordinal and invariant learning in games.}
Ordinal games and ordinal potential games have a long history in game theory
\citep{monderer-shapley96potential-games,cruz-et-al00ordinal,
durieu-et-al08ordinal-games,ewerhart20ordinal-potential}, as do questions
concerning utility representations and invariance
\citep{hammond05utility-invariance,tewolde-conitzer24ijcai}.
Recent work has renewed these questions from an algorithmic perspective.
Preference graphs encode the unilateral ordering structure of finite games
\citep{biggar-shames25pref-graphs}; \citet{conitzer24complexity-ordinal-games}
studies equilibrium concepts robust to compatible cardinal representations. More precisely, \citet{conitzer24complexity-ordinal-games} studies finite ordinal games and mediated/repeated-game equilibria that remain enforceable across compatible cardinal utility representations, allowing punishment strategies to sustain robust mixed outcomes. Our robustness question is different: we keep the one-shot CCE notion fixed and intersect its solution sets over all continuous cardinal representations consistent with the same unilateral comparisons; Theorem~\ref{thm:cce-vs-ne} shows that this intersection is exactly the set of distributions supported on pure Nash equilibria. 
\citet{gemp26nash-without-numbers} develop mixed equilibrium, regret,
and learning notions directly for finite context-ordinal games.
\citet{abbadi-laraki-mertikopoulos26icml} investigate which aspects of the
long-run behavior of no-regret dynamics are, and are not, determined by the
preference structure of a finite game. In smooth games,
\citet{ahunbay26fo-cceSTOC} develops first-order notions of coarse correlated
equilibrium based on regret against local vector-field deviations.
Concurrently, \cite{cai-et-al26proximal} develop similar notions. 
Our first-order benchmark arises from a different principle: ordinal identifiability. Player-wise normalized unilateral gradients are precisely the local directions preserved under smooth monotone reparameterizations, and we subsequently learn this preference geometry directly from comparisons.

\medskip
\noindent\textbf{Normalized and limited-feedback dynamics.}
Normalized gradient and subgradient methods are classical in optimization,
with deterministic, stochastic, and online variants analyzed under a range
of convex and nonconvex assumptions
\citep{cortes06finite-time-cv,hazan-levy-shalevshwartz15ngd,
gao-et-al18,murray-swenson-kar19ngd,cutkosky-mehta20momentum-nsgd,
zhao-et-al21sngd,zhao-et-al24SNGD,josz26normalized-grad-descent}.
Our use of normalization has a different purpose: player-wise normalization
removes precisely the gradient magnitudes that unilateral preference
comparisons do not identify. Accordingly, our stochastic algorithm does not
normalize a noisy gradient; a single pairwise comparison instead gives an unbiased
estimator of the exact finite-resolution ordinal field $D^\delta$.
Learning in continuous games with limited numerical feedback has also been
studied under bandit and payoff-based information
\citep{heliou-cohen-mertikopoulos17,
bravo-leslie-mertikopoulos18bandit-concave-games,
mertikopoulos-zhou19,bervoets-bravo-faure20},
while \citet{cheung14} studies deterministic pairwise-comparison dynamics in
continuous-strategy population games and
\citet{poveda-krstic-basar22tac} consider rescaled dynamics for Nash
seeking. In contrast, the dynamics studied here are dictated by ordinal
identifiability, and their finite-comparison implementation yields both
no-ordinal-regret and last-iterate Nash-convergence guarantees.\\

\noindent\textbf{Additional details on the comparison to
\cite{liu-et-al26iclr-ranking-feedback}.}
The Plackett--Luce (PL) model is not invariant under arbitrary strictly
increasing transformations of utilities: its ranking probabilities depend on
utility differences and therefore encode cardinal utility-gap information.
As the temperature $\tau\to0^+$, the PL model becomes increasingly
deterministic and approaches purely ordinal ranking feedback, since the
ordering remains observable while information about the magnitude of utility
gaps vanishes. This limiting regime is therefore closely related to the
information structure considered here.
For instantaneous-utility rankings,
\citet{liu-et-al26iclr-ranking-feedback} show that standard sublinear external
regret is impossible in general. Their positive result in this setting assumes
that the utility sequence has sublinear temporal variation, allowing recent
ranking observations to be pooled in a sliding window to estimate an
approximately constant cardinal utility vector. Their estimation bound
contains a factor $\tau(e^{1/\tau}+1)^2$, which diverges as
$\tau\to0^+$. Thus, precisely as the PL observations approach purely ordinal
rankings, recovering the cardinal information needed for standard external
regret becomes increasingly ill-conditioned.

\section{Proofs for Section~\ref{sec:comparison-identification}: Ordinal Feedback, Invariance and Identification}

\subsection{Proof of Lemma~\ref{lem:D-invariance}}

\begin{proof}
If $F_i(x)\neq0$, then
$
\frac{\widetilde F_i(x)}{\|\widetilde F_i(x)\|}
=
\frac{\phi_i'(J_i(x))F_i(x)}
     {\phi_i'(J_i(x))\|F_i(x)\|}
=
D_i(x),$
and $\widetilde F_i(x)=0$ if and only if $F_i(x)=0$.
\end{proof}

\subsection{Proof of Proposition~\ref{prop:cardinal-regret-noninvariant}}

\begin{proof}
Consider two players with $\mathcal X_i=\mathcal X_j=[0,1]$, and let player $i$ have the smooth cost
\[
J_i(x_i,x_j)=2-x_j+x_i(1-2x_j).
\]
The cost of the second player is immaterial. Suppose player $i$ always selects $x_{i,t}=1$, while player $j$ alternates between $x_{j,t}=0$ and $x_{j,t}=1$. Over every pair of rounds, player $i$ incurs total cost $3+0=3$. Any fixed comparator $y\in[0,1]$ incurs
\[
J_i(y,0)+J_i(y,1)=(2+y)+(1-y)=3.
\]
Hence $\Reg_{i,T}^{J}=0$ for every even $T$.

Now consider the ordinally equivalent representation $\widetilde J_i=\exp(J_i)$. The learner incurs $e^3+1$ over every pair of rounds, while a fixed comparator $y$ incurs $e^{2+y}+e^{1-y}$. Since $y\mapsto e^{2+y}+e^{1-y}$ is strictly increasing on $[0,1]$, its minimum is attained at $y=0$. Therefore
\[
\Reg_{i,T}^{\widetilde J}
=
\frac{T}{2}\left(e^3+1-e^2-e\right)
=
\Omega(T),
\]
while $\widetilde J_i=\exp(J_i)$ induces exactly the same preference comparisons as $J_i$.
\end{proof}

\subsection{Proof of Proposition~\ref{prop:cce-not-ordinally-invar}}

\begin{proof}
Consider a two-player game with $\mathcal X_1=\mathcal X_2=[0,1],$
and costs
\[
J_1(x_1,x_2)
=
2-x_2+x_1(1-2x_2),
\qquad
J_2(x_1,x_2)=0.
\]
Let $\mu =\frac12\delta_{(1,0)} + \frac12\delta_{(1,1)}.$ 
We first show that $\mu$ is a coarse correlated equilibrium of
$\mathcal G$. Under $\mu$, player~$1$ incurs expected cost
\[
\mathbb E_{x\sim\mu}[J_1(x)]
=
\frac12 J_1(1,0)+\frac12 J_1(1,1)
=
\frac12(3+0)
=
\frac32.
\]
For any fixed deviation $y_1\in[0,1]$,
\begin{equation*}
\mathbb E_{x\sim\mu}
\big[J_1(y_1,x_2)\big]
=
\frac12 J_1(y_1,0)
+
\frac12 J_1(y_1,1)
=
\frac12(2+y_1)
+
\frac12(1-y_1)
=
\frac32.
\end{equation*}
Thus player~$1$ cannot improve its expected cost by any fixed deviation.
Player~$2$ is indifferent between all actions since $J_2\equiv 0$.
Therefore $\mu\in\operatorname{CCE}(\mathcal G).$

Now define the ordinally equivalent game $\widetilde{\mathcal G}$ by
$\widetilde J_1=e^{J_1}, \widetilde J_2=J_2.$ 
Since $s\mapsto e^s$ is strictly increasing, the transformation preserves
all unilateral preference comparisons and best responses.
Under the same distribution $\mu$, player~$1$ now incurs expected cost
\[
\mathbb E_{x\sim\mu}[\widetilde J_1(x)]
=
\frac12 e^{J_1(1,0)}
+
\frac12 e^{J_1(1,1)}
=
\frac12(e^3+1).
\]
On the other hand, the fixed deviation $y_1=0$ yields
\[
\mathbb E_{x\sim\mu}
\big[\widetilde J_1(0,x_2)\big]
=
\frac12 e^{J_1(0,0)}
+
\frac12 e^{J_1(0,1)}
=
\frac12(e^2+e).
\]
Since
$e^3+1-e^2-e
= (e-1)^2(e+1)>0,$
we have
\[
\mathbb E_{x\sim\mu}[\widetilde J_1(x)]
>
\mathbb E_{x\sim\mu}
\big[\widetilde J_1(0,x_2)\big].
\]
Hence player~$1$ has a profitable fixed deviation, so
$\mu\notin\operatorname{CCE}(\widetilde{\mathcal G}).$
This proves that the CCE correspondence is not preserved under arbitrary
strictly increasing transformations.
\end{proof}

\subsection{Proof of Theorem~\ref{thm:cce-vs-ne}}
\label{appx:proof-thm-cce-vs-ne}

The inclusion $\Delta(\operatorname{NE}(\mathcal G))\subseteq \bigcap_{\widetilde{\mathcal G}\sim_{\mathrm{ord}}\mathcal G}
\operatorname{CCE}(\widetilde{\mathcal G})$ is immediate from the invariance of best responses under unilateral ordinal transformations. For the reverse inclusion, recall first that the usual CCE condition only controls \emph{average} regret, so positive regret at some opponents' profiles may be offset by negative regret at others. Robustness over the ordinal equivalence class rules out precisely this cross-context compensation: if, for some fixed deviation \(y_i\), the conditional regret were positive on a set of opponents' profiles of positive measure, one could choose \(\alpha_i\) large on that set and small elsewhere, without changing any unilateral comparison, thereby making the overall expected regret positive. Hence ordinal robustness forces every fixed-deviation conditional regret to be nonpositive almost everywhere. Taking the infimum over deviations then implies that the player's conditional expected cost equals the pointwise minimum cost, which is possible only if the player is best-responding almost surely. Applying this argument to every player yields support on \(\operatorname{NE}(\mathcal G)\). We provide a complete proof below in this section.\\   

\noindent\textbf{Proof idea.} We start by highlighting the intuition behind the proof of the more challenging and interesting inclusion $\left( \bigcap_{\widetilde{\mathcal G}\in[\mathcal G]_{\rm ord}}
\operatorname{CCE}(\widetilde{\mathcal G})
\subseteq
\Delta(\operatorname{NE}(\mathcal G)) \right)$. 

For a fixed deviation $y_i$, ordinary CCE only requires

$$
\int q_{i,y_i}(z)\,\lambda(dz)\leq0,
\qquad
q_{i,y_i}(z)
=
\mathbb E[
J_i(X_i,z)\mid X_{-i}=z]
-
J_i(y_i,z).
$$

Hence positive conditional regret at some opponents' profiles may be offset
by negative conditional regret at others. Unilateral comparison feedback,
however, does not identify the relative cardinal scale of regret across
different $z$: replacing $J_i(x_i,z)$ by
$\alpha_i(z)J_i(x_i,z)$ for any positive $\alpha_i(z)$ leaves all
unilateral comparisons unchanged. Requiring the CCE inequality to hold
under every such reweighting eliminates cross-context cancellation and
forces

$$
q_{i,y_i}(z)\leq0
\qquad
\text{a.e.}
$$

for every deviation $y_i$. Taking the infimum over $y_i$ then gives

$$
\mathbb E[
J_i(X_i,z)\mid X_{-i}=z]
=
\min_{y_i}J_i(y_i,z),
$$

which is possible only if $X_i$ is a best response almost surely.
Thus robustness to all cardinal representations consistent with the
comparison oracle turns the global expected-regret inequalities defining CCE
into pointwise best-response conditions, yielding Nash support.

In what follows, we provide a complete and detailed proof of Theorem~\ref{thm:cce-vs-ne}. 

\begin{proof}
We prove the two inclusions separately.\\

\noindent\textbf{(i) Proof of $\Delta(\operatorname{NE}(\mathcal G))
\subseteq \bigcap_{\widetilde{\mathcal G}\in[\mathcal G]_{\rm ord}}
\operatorname{CCE}(\widetilde{\mathcal G})$.}
Let $\mu\in\Delta(\operatorname{NE}(\mathcal G)),$ 
and consider an arbitrary
$\widetilde{\mathcal G}\in[\mathcal G]_{\rm ord}$.
By ordinal equivalence, for every $i \in [N]$, every $x_{-i} \in \mathcal{X}_{-i}$, and every
$a_i,b_i\in\mathcal X_i$,
\begin{equation*}
\operatorname{sgn}
\bigl(
J_i(a_i,x_{-i})-J_i(b_i,x_{-i})
\bigr)
=
\operatorname{sgn}
\bigl(
\widetilde J_i(a_i,x_{-i})
-
\widetilde J_i(b_i,x_{-i})
\bigr).    
\end{equation*}
Hence the two cost functions induce exactly the same minimizers on every
unilateral slice:
$\operatorname{BR}^{\mathcal G}_i(x_{-i})=
\operatorname{BR}^{\widetilde{\mathcal G}}_i(x_{-i}).$ Therefore, 
$\operatorname{NE}(\mathcal G) = \operatorname{NE}(\widetilde{\mathcal G}).$

Since $\mu$ is supported on $\operatorname{NE}(\mathcal G)$, for every
player $i \in [N]$ and every fixed deviation $y_i\in\mathcal X_i$,
$\widetilde J_i(x_i,x_{-i})\leq \widetilde J_i(y_i,x_{-i})$ 
for $\mu$-a.e. $x$. 
Integrating gives $\mathbb E_{x\sim\mu} \bigl[\widetilde J_i(x)\bigr] \leq \mathbb E_{x\sim\mu} \bigl[\widetilde J_i(y_i,x_{-i})\bigr].$ 
Thus $\mu\in\operatorname{CCE}(\widetilde{\mathcal G}).$
Since $\widetilde{\mathcal G}\in[\mathcal G]_{\rm ord}$ was arbitrary,
$\mu\in\bigcap_{\widetilde{\mathcal G}\in[\mathcal G]_{\rm ord}}
\operatorname{CCE}(\widetilde{\mathcal G}).$\\

\noindent\textbf{(ii) Proof of
$\displaystyle
\bigcap_{\widetilde{\mathcal G}\in[\mathcal G]_{\rm ord}}
\operatorname{CCE}(\widetilde{\mathcal G})
\subseteq
\Delta(\operatorname{NE}(\mathcal G))$.}
Let $\mu\in\bigcap_{\widetilde{\mathcal G}\in[\mathcal G]_{\rm ord}}
\operatorname{CCE}(\widetilde{\mathcal G}).$ 

We show that $\mu$ is supported on $\operatorname{NE}(\mathcal G)$.

\emph{Proof strategy.}
Fix a player $i \in [N]$. Conditional on an opponents' profile $z=x_{-i} \in \mathcal{X}_{-i}$, let

$$
A_i(z)
:=
\mathbb E_\mu\!\left[J_i(X_i,z)\mid X_{-i}=z\right]
$$

denote player $i$'s conditional expected cost. For a fixed deviation
$y_i\in\mathcal X_i$, define its conditional regret for any $z \in \mathcal{X}_{-i}$ as follows: 
$$
q_{i,y_i}(z)
:=
A_i(z)-J_i(y_i,z).
$$

The key step is to show that ordinal robustness forces
$$
q_{i,y_i}(z)\leq0
\qquad
\text{for a.e. }z
$$

for every fixed $y_i$. Otherwise, an ordinally invisible positive rescaling
of costs could amplify the contexts in which $q_{i,y_i}>0$ and violate the
CCE inequality. We then remove the dependence on $y_i$ of the exceptional null set 
using a countable dense set of deviations. This yields
$$
A_i(z)
\leq
\min_{y_i\in\mathcal X_i}J_i(y_i,z)
\qquad\text{a.e.}
$$
The reverse inequality holds automatically because $A_i(z)$ is an average
of the values $J_i(\cdot,z)$. Hence the conditional average equals the
minimum, which is possible only if the conditional distribution of $X_i$
is supported on best responses. Repeating this argument for every player
implies $\mu(\operatorname{NE}(\mathcal G))=1$.

\medskip
\noindent\emph{Step 1: Ordinal robustness implies nonpositive conditional regret.}
Fix $i\in[N]$, write $z:=x_{-i}$. 
Denote by $\pi_{-i}:\mathcal X\to\mathcal X_{-i},\pi_{-i}(x)=x_{-i},$
the canonical projection onto the opponents' action space. Let
$\lambda:=\mu_{-i}:=(\pi_{-i})_{\#}\mu$ be the marginal distribution of $X_{-i}$ induced by $\mu$; equivalently, for every Borel set $B\subseteq\mathcal X_{-i}$,
$\mu_{-i}(B) = \mu\bigl(\mathcal X_i\times B\bigr).$ 
Since $\mathcal X_i$ and $\mathcal X_{-i}$ are compact metric spaces, there exists a regular conditional distribution
$\kappa_i(\,\cdot\mid z)$ of $X_i$ given $X_{-i}=z$ such that the disintegration
$\mu(dx_i,dz) =\kappa_i(dx_i\mid z)\,\mu_{-i}(dz) = \kappa_i(dx_i\mid z) \lambda(dz)$ holds.

Define for any $z \in \mathcal{X}_{-i}$ the average cost of player $i \in [N]$ as follows: 
$$
A_i(z)
:=
\int_{\mathcal X_i}
J_i(x_i,z)\,
\kappa_i(dx_i|z).
$$
Since $J_i$ is continuous on the compact set $\mathcal X$, it is bounded,
and therefore $A_i$ is bounded and measurable. Fix an arbitrary deviation $y_i\in\mathcal X_i$ and define 
the conditional regret
$$
q_{i,y_i}(z)
:=
A_i(z)-J_i(y_i,z).
$$

We claim that
\begin{equation}
\label{eq:conditional-regret-nonpositive}
q_{i,y_i}(z)\leq0
\qquad
\text{for $\lambda$-a.e. }z.
\end{equation}

To see this, let $h$ be any nonnegative smooth function defined on an open
neighborhood of $\mathcal X_{-i}$, and let $t\geq0$. Consider the game
$\mathcal G^{t,h}$ obtained by replacing player $i$'s cost by
$$
J_i^{t,h}(x_i,z)
:=
\bigl(1+t h(z)\bigr)J_i(x_i,z)\,, \quad \forall x_i \in \mathcal{X}_i, \forall z \in \mathcal{X}_{-i}\,,
$$
and leaving every other player's cost unchanged. Since
$1+t h(z)>0,$ this transformation preserves every unilateral comparison: for any $a_i, b_i \in \mathcal{X}_i, z \in \mathcal{X}_{-i},$
$$
\operatorname{sgn}\!\left(
J_i^{t,h}(a_i,z)-J_i^{t,h}(b_i,z)
\right)
=
\operatorname{sgn}\!\left(
J_i(a_i,z)-J_i(b_i,z)
\right).
$$
Therefore, $\mathcal G^{t,h}\sim_{\rm ord}\mathcal G.$
Since $\mu\in\bigcap_{\widetilde{\mathcal G}\in[\mathcal G]_{\rm ord}}
\operatorname{CCE}(\widetilde{\mathcal G}),$  it follows that 
$\mu\in\operatorname{CCE}(\mathcal G^{t,h}).$ Then by definition of $\operatorname{CCE}(\mathcal G^{t,h})$, we have 
for the fixed deviation $y_i \in \mathcal{X}_i$, 
$$
\mathbb E_\mu\!\left[
\bigl(1+t h(X_{-i})\bigr)
\bigl(
J_i(X_i,X_{-i})-J_i(y_i,X_{-i})
\bigr)
\right]
\leq0.
$$

Disintegrating with respect to $X_{-i}$ (recall that $\mu(dx_i,dz) = \kappa_i(dx_i |z) \lambda(dz)$) yields: 
\begin{equation}
\label{eq:weighted-regret}
\int_{\mathcal X_{-i}}
\bigl(1+t h(z)\bigr)
q_{i,y_i}(z)
\lambda(dz)
\leq0\,, 
\qquad
\forall t\geq0.
\end{equation}
Therefore, 
$$
\int_{\mathcal X_{-i}} h(z) q_{i,y_i}(z)\,\lambda(dz)
\leq - \frac{1}{t} \int_{\mathcal X_{-i}} q_{i,y_i}(z)\,\lambda(dz) \,, 
\qquad
\forall t > 0.
$$
Taking the limit $t \to + \infty$ implies that necessarily: 
\begin{equation}
\label{eq:test-function}
\int_{\mathcal X_{-i}}
h(z)q_{i,y_i}(z) 
\lambda(dz)
\leq0
\end{equation}
for every nonnegative smooth $h$. 

We now localize \eqref{eq:test-function}, i.e., we use the freedom to choose the test function $h$ so that it concentrates on an arbitrary region of $\mathcal{X}_{-i}$, to pass from the global integral inequality~\ref{eq:test-function} to the local/almost-everywhere statement~\eqref{eq:conditional-regret-nonpositive} by using test functions with localized support. Suppose, toward a contradiction,
that
$$
\lambda\bigl(\{z \in \mathcal{X}_{-i}:q_{i,y_i}(z)>0\}\bigr)>0.
$$
Then, for some $\varepsilon>0$,
$E_\varepsilon:=\{z:q_{i,y_i}(z)\geq\varepsilon\}$ 
has positive $\lambda$-measure. By inner regularity of $\lambda$, choose a
compact set
$K\subseteq E_\varepsilon,\, \lambda(K)>0.$

Since $q_{i,y_i}$ is bounded by boundedness of $J_i$ over the compact set $\mathcal{X}$, let $M:=\|q_{i,y_i}\|_\infty < \infty.$ By outer regularity, choose an open neighborhood $U\supseteq K$ such that
$$
\lambda(U\setminus K)
<
\frac{\varepsilon\lambda(K)}{2M},
$$
where the claim is immediate if $M=0$. Choose a smooth bump function $h$
satisfying: 
$$
0\leq h\leq1,
\qquad
h=1\ \text{on }K,
\qquad
\operatorname{supp}(h)\cap\mathcal X_{-i}\subseteq U.
$$
Then
$$
\begin{aligned}
\int h(z)q_{i,y_i}(z)\,\lambda(dz)
&\geq
\int_K q_{i,y_i}(z)\,\lambda(dz)
-
\int_{U\setminus K}|q_{i,y_i}(z)|\,\lambda(dz)
\\
&\geq
\varepsilon\lambda(K)-M\lambda(U\setminus K)
\\
&>
\frac{\varepsilon\lambda(K)}{2}
>0,
\end{aligned}
$$

contradicting \eqref{eq:test-function}. Hence
$q_{i,y_i}(z)\leq0, 
\lambda\text{-a.e.}$
for every fixed $y_i\in\mathcal X_i$. Equivalently,
\begin{equation}
\label{eq:conditional-fixed-y}
A_i(z)
\leq
J_i(y_i,z)\,,
\qquad
\lambda\text{-a.e.}
\end{equation}

\medskip
\noindent\emph{Step 2: Obtain a common full-measure set for all deviations.}
The exceptional null set in \eqref{eq:conditional-fixed-y} may depend on
$y_i$. To remove this dependence, let
$D_i\subseteq\mathcal X_i$
be a countable dense subset. For each $y_i\in D_i$, let $Z_{i,y_i}$ be a
set of full $\lambda$-measure on which
$A_i(z)\leq J_i(y_i,z).$ 
Because $D_i$ is countable,
$Z_i:=\bigcap_{y_i\in D_i} Z_{i,y_i}$
still satisfies
$\lambda(Z_i)=1.$
Thus, for every $z\in Z_i$,
$$
A_i(z)\leq J_i(y_i,z)
\qquad
\forall y_i\in D_i.
$$
By continuity of $J_i(\cdot,z)$ and density of $D_i$,
$$
A_i(z)
\leq
\inf_{y_i\in D_i}J_i(y_i,z)
=
\min_{y_i\in\mathcal X_i}J_i(y_i,z)\,,
\qquad
\forall z\in Z_i.
$$
We have therefore shown that
\begin{equation}
\label{eq:A-less-than-min}
A_i(z)\leq \min_{y_i\in\mathcal X_i}J_i(y_i,z)
\qquad
\lambda\text{-a.e.}
\end{equation}

\medskip
\noindent\emph{Step 3: Nonpositive conditional regret forces best-response play.}
Note first that for any $z \in \mathcal{X}_{-i}$, 
$$
J_i(x_i,z)\geq \min_{y_i\in\mathcal X_i}J_i(y_i,z)\,, 
\qquad
\forall x_i\in\mathcal X_i.
$$

Averaging with respect to $\kappa_i(\cdot|z)$ gives 
$$
A_i(z)
=
\int_{\mathcal X_i}J_i(x_i,z)\,
\kappa_i(dx_i\mid z)
\geq
\min_{y_i\in\mathcal X_i}J_i(y_i,z).
$$

Together with \eqref{eq:A-less-than-min}, this yields
\begin{equation}
\label{eq:average-equals-min}
A_i(z)=\min_{y_i\in\mathcal X_i}J_i(y_i,z)
\qquad
\lambda\text{-a.e.}
\end{equation}

Now define for any $x_i \in \mathcal{X}_i, x_{-i} \in \mathcal{X}_{-i}$, 
$r_i(x_i,z) := J_i(x_i,z)-\min_{y_i\in\mathcal X_i}J_i(y_i,z)\geq0.$

For $\lambda$-a.e. $z$, \eqref{eq:average-equals-min} implies
$$
\int_{\mathcal X_i}
r_i(x_i,z)\,
\kappa_i(dx_i\mid z)
=
A_i(z)-\min_{y_i\in\mathcal X_i}J_i(y_i,z)
=
0.
$$

Since $r_i(\cdot,z)$ is nonnegative, it must vanish
$\kappa_i(\cdot\mid z)$-almost surely. Therefore

$$
\kappa_i
\bigl(
\operatorname{BR}^{\mathcal G}_i(z)
\mid z
\bigr)
=
1
\qquad
\lambda\text{-a.e. }z.
$$
Integrating over $z$ gives
\begin{equation}
\label{eq:player-best-response-as}
\mu\left(
\left\{
x \in \mathcal{X}:
x_i\in\operatorname{BR}_i^{\mathcal G}(x_{-i})
\right\}
\right) = 1\,.
\end{equation}

\noindent\emph{Step 4: Intersect the best-response conditions across players.}
Since \eqref{eq:player-best-response-as} holds for every $i\in[N]$, and
the number of players is finite,

$$
\mu\!\left(
\bigcap_{i=1}^N
\left\{
x \in \mathcal{X}:
x_i\in\operatorname{BR}_i^{\mathcal G}(x_{-i})
\right\}
\right)
=
1.
$$

The set inside the probability is exactly
$\operatorname{NE}(\mathcal G)$. Hence
$\mu(\operatorname{NE}(\mathcal G))=1,$
or equivalently,
$\mu\in\Delta(\operatorname{NE}(\mathcal G)),$
which concludes the proof of the desired inclusion. 
\end{proof}

\subsection{Proof of Proposition~\ref{prop:ordinal-gap-ne}}

\begin{proof}
Since every summand in \eqref{eq:ordinal-nash-gap} is nonnegative, $\Gap_D(x)=0$ if and only if, 
\begin{equation}
\label{eq:variational-charac}
\left\langle D_i(x),y_i-x_i\right\rangle \geq 0\, \quad \forall i \in [N], \forall\, y_i \in \mathcal{X}_i.
\end{equation}
If $\nabla_iJ_i(x)\neq0$, multiplying the previous inequality~\eqref{eq:variational-charac} by the positive scalar $\|\nabla_iJ_i(x)\|$ yields $\langle\nabla_iJ_i(x),y_i-x_i\rangle\geq0$ for all $y_i$. If $\nabla_iJ_i(x)=0$, the same condition holds trivially. 

By pseudoconvexity, this first-order condition is sufficient for $x_i$ to globally minimize $J_i(\cdot,x_{-i})$ and hence $x$ is a Nash equilibrium. 

Conversely, if $x$ is a NE, since every global minimizer of a differentiable function over a convex set satisfies the corresponding first-order variational inequality, it follows that \eqref{eq:variational-charac} holds and hence by the first step in this proof we obtain that $\Gap_D(x)=0$. 
\end{proof}

\section{Continuous-Time Analysis of \eqref{eq:BNPGD}}
\label{apdx:cont-time-BNPGD}

In the unconstrained case (and at interior points of a full-dimensional action set), \eqref{eq:BNPGD} reduces  
to $\dot x_i = -\frac{\nabla_iJ_i(x)}{\|\nabla_iJ_i(x)\|}$ whenever $\nabla_i J_i(x) \neq 0$, and $\dot{x}_i = 0$ otherwise. 
Because normalization is singular at stationary blocks, the field $D_i$ is
continuous on $\{x \in \mathcal{X}:\nabla_iJ_i(x)\neq0\}$ but may be discontinuous where
$\nabla_iJ_i(x)=0$ (see Appendix~\ref{apdx:ex-discontinuity} for a simple example), despite the convention $D_i=0$ at such points.
Accordingly, our continuous-time statements apply to any absolutely continuous
trajectory satisfying \eqref{eq:BNPGD} almost everywhere. The derivative
identities are likewise understood almost everywhere.

\begin{tcolorbox}[colframe=white!, top=2pt,left=2pt,right=2pt,bottom=2pt] 
\begin{proposition}[Invariance of the dynamics]
\label{prop:dynamics-invariant}
Suppose $\widetilde J_i=\phi_i\circ J_i$ for every $i$, where
$\phi_i'>0$. Then the block-normalized projected vector fields generated
by $\mathcal G = ((\mathcal{X})_{i \in [N]}, (J_i)_{i \in [N]})$ and $\widetilde{\mathcal G} = ((\mathcal{X})_{i \in [N]}, (\tilde{J}_i)_{i \in [N]})$ coincide.
\end{proposition}
\end{tcolorbox}

In contrast, ordinary gradient play is transformed blockwise as
$-\nabla_iJ_i(x)\longmapsto
-\phi_i'(J_i(x))\nabla_iJ_i(x),$ 
so ordinally equivalent games may induce different relative adaptation speeds
and different joint trajectories due to the player-wise factor $\phi_i'(J_i(x))$. 
This is specific to the multi-player
setting: for a single objective, normalization only reparametrizes the speed
along a gradient-flow trajectory, whereas player-wise normalization generally
changes the geometry of the coupled dynamics.

In this section, we show that \eqref{eq:BNPGD} enjoys a $\mathcal{O}(1/T)$ continuous-time average ordinal regret (see Theorem~\ref{thm:ct-regret}), and in ordinal potential games the potential is a Lyapunov function yielding convergence to Nash equilibria under suitable conditions (see Theorem~\ref{thm:cv-BNPG-ordinal-pot}). 

\subsection{Discontinuity example of the preference-direction vector field}
\label{apdx:ex-discontinuity}

In one dimension and for a single player ($N=1$), let $J(x)=\frac{1}{2}x^2$ for any $x \in \mathbb{R}.$
Then, 
\[
D(x)
=
\begin{cases}
1, & x>0,\\
0, & x=0,\\
-1, & x<0.
\end{cases}
\]
Hence $D$ is well defined everywhere under the convention $D(0)=0$, but it is
discontinuous at the stationary point $x=0$.

\subsection{Continuous-time regret and equilibrium}

For continuous play, define
\begin{equation}
\label{eq:cont-dir-reg}
\Reg_i^D(T)
:=
\sup_{y_i\in\mathcal X_i}
\int_0^T
\langle D_i(x(t)),x_i(t)-y_i\rangle\,dt.
\end{equation}

\begin{tcolorbox}[colframe=white!, top=2pt,left=2pt,right=2pt,bottom=2pt] 
\begin{theorem}[Constant continuous-time directional regret]
\label{thm:ct-regret}
Every trajectory of \eqref{eq:BNPGD} satisfies
$\Reg_i^D(T)\leq\frac{\diam(\mathcal X_i)^2}{2}$ for every $T\geq 0.$
Consequently,
$\frac{\Reg_i^D(T)}{T}=O(T^{-1}).$
\end{theorem}
\end{tcolorbox}
Thus block-normalized play is no-regret with respect to the ordinal benchmark
against arbitrary opponent behavior, without monotonicity or potential
assumptions.

Let $\mu_T =\frac1T\int_0^T\delta_{x(t)}\,dt$ denote the occupation measure. Combining
Theorem~\ref{thm:ct-regret} with
\eqref{eq:ordinal-coarse-var-gap-from-dir-reg} gives:

\begin{tcolorbox}[colframe=white!, top=2pt,left=2pt,right=2pt,bottom=2pt] 
\begin{corollary}[Ordinal coarse equilibrium at rate $O(1/T)$]
\label{cor:ct-ocve}
The occupation measure of block-normalized play satisfies
$\mathcal E_D(\mu_T)
\leq
\frac{1}{2T}\sum_{i=1}^N R_i^2.$
Hence time-average play converges to the set of ordinal coarse variational
equilibria at rate $O(T^{-1})$.
\end{corollary}
\end{tcolorbox}
This is the ordinal counterpart of the classical no-regret-to-CCE
correspondence.

\subsection{Nash convergence in ordinal potential games}
No-regret alone controls time-average behavior. We next identify a class of
games in which the same dynamics also admits a Lyapunov function and converges
in the last iterate.

\begin{tcolorbox}[colframe=white!, top=2pt,left=2pt,right=2pt,bottom=2pt] 
\begin{lemma}[First-order alignment of costs and the ordinal potential]
\label{lem:ordinal-potential-gradient-alignment}
Suppose that $\mathcal G$ is an ordinal potential game with costs
$J_i$ and ordinal potential $\Phi$ that are $C^1$ on an open
neighborhood of $\mathcal X$, and suppose that each $\mathcal X_i$ has nonempty interior. 
Under Assumption~\ref{as:non-degeneracy-ord-pot}, for every
$i\in[N]$ and every $x\in\mathcal X$, either
$\nabla_iJ_i(x) = \nabla_i\Phi(x) = 0$, 
or there exists a scalar $\lambda_i(x)>0$ s.t. 
\begin{equation}
\label{eq:positive-gradient-collinearity}
\nabla_iJ_i(x)
=
\lambda_i(x)\nabla_i\Phi(x).
\end{equation}
In particular, whenever the gradients are nonzero,
$
D_i(x)= \frac{\nabla_iJ_i(x)}
{\|\nabla_iJ_i(x)\|}
=
\frac{\nabla_i\Phi(x)}
{\|\nabla_i\Phi(x)\|}.
$
\end{lemma}
\end{tcolorbox}

\begin{remark}[Role of Assumption~\ref{as:non-degeneracy-ord-pot}]
The nondegeneracy condition is needed because ordinal equivalence alone
does not preserve first-order stationarity. For example, the functions
$s\mapsto s$ and $s\mapsto s^3$ induce the same ordering on
$\mathbb R$, although their derivatives behave differently at the
origin. Thus, without nondegeneracy, an ordinal representation could
create or remove stationary blocks. Assumption~\ref{as:non-degeneracy-ord-pot}
rules out precisely this pathology and allows first-order stationarity
of the potential to be identified with first-order stationarity of the
players' costs.
\end{remark}

Thus block normalization removes the player-specific factor $\lambda_i(x)$:
within an ordinal potential game, \eqref{eq:BNPGD} follows the normalized
block gradients of a common potential.

\begin{tcolorbox}[colframe=white!, top=2pt,left=2pt,right=2pt,bottom=2pt] 
\begin{theorem}[Convergence of \eqref{eq:BNPGD} in ordinal potential games]
\label{thm:cv-BNPG-ordinal-pot}
Let $\mathcal{G}$ be an ordinal potential game with potential $\Phi$ such that
Assumption~\ref{as:non-degeneracy-ord-pot} holds. 
Suppose in addition that $\mathcal{X}_i$ has nonempty interior for any $i \in [N]$.  
Then the following statements
hold:
\begin{enumerate}[label=(\roman*), leftmargin=*,itemsep=0pt,topsep=0pt]
\item \label{pot-dissipation} (potential dissipation)
$\frac{d}{dt}\Phi(x(t))
= - \sum_{i=1}^N \|\nabla_i\Phi(x(t))\|
\,
\|\Pi_{T_{\mathcal X_i}(x_i)}\!\left(-D_i(x)\right)\|^2
\le 0$ for almost every $t$ along every absolutely
continuous solution of \eqref{eq:BNPGD}.
\item \label{pot-dissipation-interior} (block-$\ell_{2,1}$ dissipation)
$\frac{d}{dt}\Phi(x(t)) =
-\sum_{i=1}^N \|\nabla_i\Phi(x(t))\|
= -\|\nabla\Phi(x(t))\|_{2,1}$
for almost every $t$, if the trajectory remains in $\operatorname{int}(\mathcal X)$.
\item \label{fos} (first-order stationarity)
For any $T>0$,
$\min_{0\le t\le T}
\|\nabla\Phi(x(t))\|_{2,1}
\le
\frac{\Phi(x(0))-\Phi_{\inf}}{T},$
where
$\Phi_{\inf}:=\inf_{x\in\mathcal X}\Phi(x)>-\infty$,
if again the trajectory remains interior.
\item \label{convergence} (convergence)
Assuming a solution of \eqref{eq:BNPGD} remains in a compact set
$K\subset\operatorname{int}(\mathcal X)$ and $\nabla\Phi$ is Lipschitz on $K$,
then for every $i\in[N]$,
$\|\nabla_i\Phi(x(t))\|\longrightarrow0.$
If in addition Assumption~\ref{ass:block-pseudoconvex} holds, then the rest
points of \eqref{eq:BNPGD} coincide exactly with the set of Nash equilibria
and
$\operatorname{dist}\bigl(x(t),\operatorname{NE}\bigr)\longrightarrow 0.$
\end{enumerate}
\end{theorem}
\end{tcolorbox}
The interior dissipation identity admits a useful geometric interpretation.
Since $\|\cdot\|_{2,1}$ is dual to
$\|v\|_{2,\infty}=\max_i\|v_i\|_2$,
$\min_{\|v_i\|\le1,\ \forall i}
\langle\nabla\Phi(x),v\rangle=
-\|\nabla\Phi(x)\|_{2,1}.$
Hence block-normalized play is the steepest-descent direction of the ordinal
potential when each player's instantaneous speed is bounded by one.

\section{Proofs for Section~\ref{sec:ordinal-learning-dyn}: Ordinal Learning Dynamics}

\subsection{Proofs for Section~\ref{subsec:continuous-time-dyn}: Continuous-time ordinal dynamics}

\subsubsection{Proof of Theorem~\ref{thm:ct-regret}}

\begin{proof}
Fix $y_i\in\mathcal X_i$ and define $V_i(t)=\frac12\|x_i(t)-y_i\|^2$. The characterization of the projection in the  \eqref{eq:BNPGD} dynamics yields $-D_i(x)=\dot x_i+n_i$ for some $n_i\in N_{\mathcal X_i}(x_i)$, where $N_{\mathcal X_i}(x_i)$ denotes the normal cone to the set $\mathcal{X}_i$. Since $\langle n_i,y_i-x_i\rangle\leq0$ by definition of the normal cone,
\[
\dot V_i(t)=\langle \dot x_i, x_i-y_i\rangle
= -\langle D_i(x),x_i-y_i\rangle + \langle n_i,y_i-x_i\rangle
\leq-\langle D_i(x),x_i-y_i\rangle.
\]
Integrating the above inequality from $0$ to $T$ and rearranging gives: 
\[
\int_0^T\langle D_i(x(t)),x_i(t)-y_i\rangle\,dt
\leq
\frac12\|x_i(0)-y_i\|^2
-\frac12\|x_i(T)-y_i\|^2
\leq\frac{\diam(\mathcal X_i)^2}{2}.
\]
Taking the supremum over $y_i$ proves the claim.
\end{proof}

\subsubsection{Proof of Lemma~\ref{lem:ordinal-potential-gradient-alignment}}

\begin{proof}
Let $i\in[N]$ and $x\in\mathcal X$.
By Assumption~\ref{as:non-degeneracy-ord-pot},
the two block gradients vanish simultaneously. Thus it remains only
to consider the case $\nabla_iJ_i(x)\neq0, \nabla_i\Phi(x)\neq0\,.$

Suppose first that $x_i\in\operatorname{int}(\mathcal X_i)$.
Define the two $C^1$ functions
\[
f_i(y_i):=J_i(y_i,x_{-i}),
\qquad
\phi_i(y_i):=\Phi(y_i,x_{-i}).
\]
The ordinal-potential property implies that, for all
$y_i,z_i\in\mathcal X_i$,
\[
\operatorname{sgn}
\bigl(
f_i(y_i)-f_i(z_i)
\bigr)
=
\operatorname{sgn}
\bigl(
\phi_i(y_i)-\phi_i(z_i)
\bigr).
\]
We claim that the two gradients $\nabla f_i(x_i)$ and $\nabla \phi_i(x_i)$ are positively collinear.\\ 

\noindent\textbf{Proof strategy via orthogonal complements.} Set
$a:=\nabla f_i(x_i), b:=\nabla\phi_i(x_i).$ 
By assumption, $a\neq0$ and $b\neq0$.
To show that $a$ and $b$ are collinear, it suffices to show that $\operatorname{span}\{a\}
=\operatorname{span}\{b\}.$ In the following, we equivalently show by taking orthogonal complements that 
$a^\perp := \{d:\langle a,d\rangle=0\}$ and $b^\perp := \{d:\langle b,d\rangle=0\}$ coincide, i.e., 
$a^\perp=b^\perp.$ Finally, since $a$ and $b$ are both nonzero,
$a^\perp$ and $b^\perp$ are both $(d_i-1)$-dimensional subspaces, it suffices to show that $a^\perp \subseteq b^\perp$ 
to conclude the equality $a^\perp = b^\perp\,.$\\

\noindent\textbf{Proof of the inclusion $a^\perp \subseteq b^\perp$.} 
Fix any direction $d\in\mathbb R^{d_i}$ satisfying
$\langle a,d\rangle=0$. We prove that $\langle b, d\rangle = 0$ as well. 
For $\varepsilon>0$, define
$d_\varepsilon^+
:= d+\varepsilon a,
d_\varepsilon^-
:=
d-\varepsilon a.$ 
Since $x_i\in\operatorname{int}(\mathcal X_i)$, for every fixed
$\varepsilon>0$ and all sufficiently small $h>0$,
$x_i+hd_\varepsilon^\pm\in\mathcal X_i$.

By differentiability of $f_i$ (which follows from differentiability of $J_i)$,
\begin{align*}
f_i(x_i+hd_\varepsilon^+)-f_i(x_i)
=
h\langle a,d_\varepsilon^+\rangle+o(h)
=
h\varepsilon\|a\|^2+o(h)>0
\end{align*}
for all sufficiently small $h>0$.
By the ordinal-potential property,
\[
\phi_i(x_i+hd_\varepsilon^+)-\phi_i(x_i)>0
\]
for the same sufficiently small $h$.
Dividing by $h$ and letting $h\downarrow0$ gives
$\langle b,d+\varepsilon a\rangle\ge0.$ 
Similarly,
\begin{align*}
f_i(x_i+hd_\varepsilon^-)-f_i(x_i)
=
h\langle a,d_\varepsilon^-\rangle+o(h)
= -h\varepsilon\|a\|^2+o(h)<0
\end{align*}
for sufficiently small $h>0$.
Ordinal equivalence therefore gives
\[
\phi_i(x_i+hd_\varepsilon^-)-\phi_i(x_i)<0,
\]
and hence
$\langle b,d-\varepsilon a\rangle\le0.$
Letting $\varepsilon\downarrow0$ in the preceding two inequalities
yields
$\langle b,d\rangle\ge0$ and
$\langle b,d\rangle\le0,$ 
so $\langle b,d\rangle=0.$ 
This proves that $a^\perp \subseteq b^\perp$.\\ 

\noindent\textbf{Proof of positive collinearity.} 
We have proved that $a$ and $b$ are collinear: there exists
$\mu_i(x)\neq0$ such that $b=\mu_i(x)a.$
It remains to show that $\mu_i(x)>0$.
Consider the direction $a$. Since
$\langle a,a\rangle=\|a\|^2>0,$ 
differentiability gives
\[
f_i(x_i+ha)-f_i(x_i)>0
\]
for every sufficiently small $h>0$.
The ordinal-potential property therefore implies
\[
\phi_i(x_i+ha)-\phi_i(x_i)>0.
\]
Dividing by $h$ and letting $h\downarrow0$ yields
$\langle b,a\rangle\ge0.$ 
Since $b=\mu_i(x)a$,
$\langle b,a\rangle
=\mu_i(x)\|a\|^2.$ 
Thus $\mu_i(x)\ge0$. Because $b\neq0$, we have
$\mu_i(x)\neq0$, and consequently
$\mu_i(x)>0.$ 
Equivalently, there exists $\lambda_i(x)>0$ such that
\[
\nabla f_i(x_i)
=
\lambda_i(x)\nabla\phi_i(x_i).
\]

\noindent\textbf{Proof for the case $x_i\in\partial\mathcal X_i$.} 
It remains to treat $x_i\in\partial\mathcal X_i$. We will use the result for the case $x_i\in\operatorname{int}(\mathcal X_i)$ proved above. 
Fix $c_i\in\operatorname{int}(\mathcal X_i)$ which always exists since $\mathcal{X}_i$ has a nonempty interior by assumption, and let
$\varepsilon_n\downarrow0$ with $\varepsilon_n\in(0,1)$. Define
$x_i^n
:= (1-\varepsilon_n)x_i+\varepsilon_n c_i.$ 
By convexity and since $c_i\in\operatorname{int}(\mathcal X_i)$, 
we have $x_i^n\in\operatorname{int}(\mathcal X_i)$ 
for every $n$, and $x_i^n\to x_i$.

Set $x^n:=(x_i^n,x_{-i}).$ 
Since $\nabla_iJ_i(x)\neq0, \nabla_i\Phi(x)\neq0,$ 
continuity implies that both gradients remain nonzero at $x^n$
for all sufficiently large $n$. Applying the interior argument gives
\[
\frac{\nabla_iJ_i(x^n)}
{\|\nabla_iJ_i(x^n)\|}
=
\frac{\nabla_i\Phi(x^n)}
{\|\nabla_i\Phi(x^n)\|}.
\]
Passing to the limit and using continuity yields
\[
\frac{\nabla_iJ_i(x)}
{\|\nabla_iJ_i(x)\|}
=
\frac{\nabla_i\Phi(x)}
{\|\nabla_i\Phi(x)\|}.
\]
Since both gradients are nonzero,
\[
\nabla_iJ_i(x)
=
\frac{\|\nabla_iJ_i(x)\|}
{\|\nabla_i\Phi(x)\|}
\,\nabla_i\Phi(x).
\]
Thus \eqref{eq:positive-gradient-collinearity} holds with
$\lambda_i(x)
:=
\frac{\|\nabla_iJ_i(x)\|}
{\|\nabla_i\Phi(x)\|}
>0.$ 
This completes the proof.
\end{proof}

\subsubsection{Proof of Theorem~\ref{thm:cv-BNPG-ordinal-pot}}

\begin{proof}

\noindent\textbf{Proof of \ref{pot-dissipation}.} Let $x(t)$ be an absolutely
continuous solution of \eqref{eq:BNPGD}. By the differentiability of $\Phi$ and using the chain rule, we have for any $t \geq 0$,  
\begin{equation}
\label{eq:lyap1}
\frac{d}{dt}\Phi(x(t)) = \langle \dot{x}(t), \nabla \Phi(x(t)) \rangle = \sum_{i=1}^N \langle \dot{x}_i(t), \nabla_i \Phi(x(t)) \rangle = \sum_{i=1}^N \|\nabla_i \Phi(x(t))\| \cdot \langle \dot{x}_i(t), D_i(x(t)) \rangle\,, 
\end{equation}
where the last equality follows from using $D_i(x(t)) = \frac{\nabla_i \Phi(x(t))}{\|\nabla_i \Phi(x(t))\|}$ as proved in Lemma~\ref{lem:ordinal-potential-gradient-alignment}. 
When $\nabla_i\Phi(x)=0$, Assumption~\ref{as:non-degeneracy-ord-pot} gives $\nabla_iJ_i(x)=0$, hence $D_i(x)=0$ and that block contributes zero to the chain-rule identity. 

Now by Moreau's decomposition theorem applied, we can write for every $i \in [N]$ and any $t \geq 0, -D_i(x(t)) = \dot{x}_i(t) + n_i(t)$ where $\dot{x}_i(t) \in T_{\mathcal{X}_i}(x_i(t))$ and $n_i(t) \in N_{\mathcal{X}_i}(x_i(t))$. In particular, we have $\langle \dot{x}_i(t), n_i(t) \rangle = 0.$ \eqref{eq:BNPGD} gives $\dot x_i=\Pi_{T_{\mathcal X_i}(x_i)}\!\left(-D_i(x)\right)$ and using this in \eqref{eq:lyap1} gives: 
\begin{equation}
\label{eq:lyap2}
\frac{d}{dt}\Phi(x(t)) = \sum_{i=1}^N \|\nabla_i \Phi(x(t))\| \cdot \langle \dot{x}_i(t), - \dot{x}_i(t) - n_i(t)  \rangle = -  \sum_{i=1}^N \|\nabla_i \Phi(x(t))\| \cdot \|\dot{x}_i(t)\|^2\,, 
\end{equation}
which concludes the proof of the first item using \eqref{eq:BNPGD}, i.e., $\dot{x}_i(t) = \Pi_{T_{\mathcal X_i}(x_i)}\!\left(-D_i(x)\right)$.\\

\noindent\textbf{Proof of \ref{pot-dissipation-interior}.} If the trajectory $x(t)$ remains in $\operatorname{int}(\mathcal X)$, it follows that $\dot{x}_i(t) = \Pi_{T_{\mathcal X_i}(x_i)}\!\left(-D_i(x)\right) = - D_i(x(t))$ as the projection is not active. Since $\|D_i(x(t))\| = 1$ when $\nabla_i \Phi(x(t)) \neq 0$, we obtain that for any $t$ s.t. $\nabla_i \Phi(x(t)) \neq 0$, 
\begin{equation}
\label{eq:lyap2}
\frac{d}{dt}\Phi(x(t)) = -  \sum_{i=1}^N \|\nabla_i \Phi(x(t))\| = - \|\nabla \Phi(x(t))\|_{2,1}\,.  
\end{equation}

\noindent\textbf{Proof of \ref{fos}.} Let $T > 0$. Integrating the previous identity between $0$ and $T$, and rearranging the resulting identity yields: 
\begin{equation}
\label{eq:integ-pot}
\int_0^T \|\nabla \Phi(x(t))\|_{2,1} dt = \Phi(x(0)) - \Phi(x(T)) \leq \Phi(x(0)) - \Phi_{\inf}\,.
\end{equation}
As a consequence, 
\begin{equation}
\min_{0\le t\le T} \|\nabla\Phi(x(t))\|_{2,1} \leq \frac{1}{T} \int_0^T \|\nabla \Phi(x(t))\|_{2,1} dt \leq \frac{\Phi(x(0)) - \Phi_{\inf}}{T}\,.
\end{equation}

\noindent\textbf{Proof of \ref{convergence}.} 
A point $x^{\star}$ is a rest point of \eqref{eq:BNPGD} if and only if: 
\begin{equation}
\Pi_{T_{\mathcal X_i}(x_i^{\star})}\!\left(-D_i(x^{\star})\right) = 0\,, \quad \forall i \in [N]\,.
\end{equation}
Equivalently, $- D_i(x^\star) \in N_{\mathcal{X}_i}(x_i^\star)$, i.e., 
\begin{equation}
\langle D_i(x^{\star}), y_i - x_i^{\star} \rangle \geq 0\,, \quad \forall y_i \in \mathcal{X}_i, \forall i \in [N]\,.
\end{equation}
By \eqref{eq:ordinal-nash-gap}, this is equivalent to $\Gap_D(x^{\star}) = 0$, which by Proposition~\ref{prop:ordinal-gap-ne} is equivalent to $x^{\star} \in \NE(\mathcal{G})$ under Assumption~\ref{ass:block-pseudoconvex}.

Now it follows from \eqref{eq:integ-pot} by taking $T \to +\infty$ that:  
\begin{equation}
\int_0^{+\infty} \|\nabla \Phi(x(t))\|_{2,1} dt < +\infty.
\end{equation}

We now prove that if $x(t)$ remains in a compact set $K\subset\operatorname{int}(\mathcal X)$ and $\nabla\Phi$ is Lipschitz on $K$, then $\|\nabla\Phi(x(t))\|_{2,1}\longrightarrow0.$ The proof of this statement follows from proving the following two claims: 
\begin{enumerate}[label=(\roman*), leftmargin=*,itemsep=0pt,topsep=0pt]
\item For all $i \in [N], \int_0^\infty \|\nabla_i\Phi(x(t))\|\,dt<\infty\,.$
\item For any $i \in [N], g_i: t \mapsto \|\nabla_i\Phi(x(t))\|$ is uniformly continuous. 
\end{enumerate}
We start by proving each one of the statements before concluding. 

Along the trajectory, we have, for almost every $t\geq 0$,
\begin{equation}
\frac{d}{dt}\Phi(x(t))
=
\sum_{i=1}^N
\left\langle \nabla_i\Phi(x(t)),\dot{x}_i(t)\right\rangle
=
-\sum_{i=1}^N \|\nabla_i\Phi(x(t))\|.
\end{equation}
Hence, for every $T>0$,
\begin{equation}
\int_0^T
\sum_{i=1}^N\|\nabla_i\Phi(x(t))\|\,dt
=
\Phi(x(0))-\Phi(x(T)).
\end{equation}
Since $\Phi$ is bounded from below, it follows that for every $i \in [N],$
\begin{equation}
\label{eq:integrable-gi}
\int_0^\infty \|\nabla_i\Phi(x(t))\|\,dt<\infty\,, 
\end{equation}
which proves the first claim above. 

As for the second claim, observe first that $\|\dot{x}_i(t)\|\leq 1$ by definition of the normalized flow. 
It follows then that $\|\dot{x}(t)\|\leq \sqrt{N}$ for the Euclidean norm and as a consequence, for all $s,t\geq 0$, we have $\|x(t)-x(s)\| \leq \sqrt{N}\,|t-s|.$ 
Since $\nabla\Phi$ is $L$-Lipschitz on the compact set $K$ in which the trajectory lies, we obtain for all $s,t\geq 0$: 
\[
\begin{aligned}
|g_i(t)-g_i(s)| \leq \|\nabla_i\Phi(x(t))-\nabla_i\Phi(x(s))\|
\leq
\|\nabla\Phi(x(t))-\nabla\Phi(x(s))\| 
\leq
L\|x(t)-x(s)\|\\
\leq
L\sqrt{N}\,|t-s|\,, 
\end{aligned}
\]
which proves that $g_i$ is Lipschitz, and in particular uniformly continuous. This concludes the proof of the second claim above. 

We now show that $g_i(t)\to 0$ using a standard argument. We provide a detailed proof for completeness. Suppose otherwise for contradiction. Then there exist
$\varepsilon>0$ and a sequence $t_n\to\infty$ such that
$g_i(t_n)\geq \varepsilon$ for all $n$. 
By uniform continuity, there exists $\delta>0$, independent of $n$, such that
\[
|t-t_n|\leq\delta
\quad\Longrightarrow\quad
g_i(t)\geq\frac{\varepsilon}{2}.
\]
Passing to a subsequence if necessary, we may assume that the intervals
$[t_n-\delta,t_n+\delta]$ are pairwise disjoint. It follows that
\[
\begin{aligned}
\int_0^\infty g_i(t)\,dt
\geq
\sum_n
\int_{t_n-\delta}^{t_n+\delta}g_i(t)\,dt 
\geq
\sum_n \varepsilon\delta
=
\infty,
\end{aligned}
\]
contradicting the integrability result proved above \eqref{eq:integrable-gi}. Therefore
$\lim_{t\to\infty}\|\nabla_i\Phi(x(t))\|=0.$

We conclude by proving that $\operatorname{dist}\bigl(x(t),\operatorname{NE}\bigr) \longrightarrow0.$
Suppose, by contradiction, that this limit statement does not hold. Then there
exist $\varepsilon>0$ and a sequence $t_k\to\infty$ such that
\begin{equation}
\operatorname{dist}(x(t_k),\operatorname{NE})\geq \varepsilon\,, 
\qquad \forall k \geq 0.
\end{equation}
Since $K$ is compact, after passing to a subsequence we may assume that $x(t_k)\to \bar x$ for some $\bar x\in K.$

By continuity of $\nabla_i\Phi$ and using the limit statement we already established, we obtain that for every player $i \in [N]$,
\[
\nabla_i\Phi(\bar x)
=
\lim_{k\to\infty}\nabla_i\Phi(x(t_k))
=
0.
\] 
Using then Assumption~\ref{as:non-degeneracy-ord-pot} gives $\nabla_i J_i(\bar{x}) = 0$ for every $i \in [N]$. 
Hence $D(\bar{x}) = 0$, which then implies that 
$\Gap_D(\bar{x}) = 0$ using \eqref{eq:ordinal-nash-gap}. We conclude then from Proposition~\ref{prop:ordinal-gap-ne} which uses Assumption~\ref{ass:block-pseudoconvex} that $\bar{x} \in \NE(\mathcal{G})$. 

Then $\operatorname{dist}(x(t_k),\operatorname{NE}) \leq \|x(t_k)-\bar x\| \longrightarrow 0$ 
because $\bar x$ is a Nash equilibrium and $x(t_k)\to \bar x$. This
contradicts $\operatorname{dist}(x(t_k),\operatorname{NE})\geq\varepsilon$.
Therefore $\operatorname{dist}(x(t),\operatorname{NE})\longrightarrow 0,$ as claimed. 
\end{proof}

\subsection{Proofs for Section~\ref{subsec:discrete-time-ordinal-grad}: Discrete-time ordinal gradient play}

\subsubsection{Proof of Theorem~\ref{thm:discrete-regret}}
\label{apdx:proof-thm-discrete-reg}

\begin{proof}
Nonexpansiveness of projection gives, for every $y_i\in\mathcal X_i$,
\[
\|x_{i,t+1}-y_i\|^2
\leq
\|x_{i,t}-y_i\|^2
-2\eta\langle D_i(x_t),x_{i,t}-y_i\rangle
+\eta^2\|D_i(x_t)\|^2.
\]
Since $\|D_i(x_t)\|\leq1$, summing over $t$ and telescoping gives the result.
\end{proof}

\begin{remark}
On unbounded action sets (e.g., $\mathcal X_i=\mathbb R^{d_i}$), the uniform comparator supremum in $\Reg^D_{i,T}$ is generally infinite. Indeed,
it is infinite whenever $\sum_t D_i(x_t)\neq0$. This differs from standard external regret, where nonlinear utilities may make the cumulative comparator objective coercive and hence yield a finite infimum. Such quantitative growth away from an optimum is not encoded by the normalized preference field. On unbounded domains, we may use a comparator-dependent regret (i.e. for a fixed comparator $y_i$ rather than supremum in $\Reg^D_{i,T}$), or restrict the comparator class to a bounded neighborhood.
\end{remark}

\section{Proofs for Section~\ref{sec:learning-from-comparison-feedback}: Learning from Pairwise Comparisons}

\subsection{Pseudocode for Comparison-Normalized Gradient Play}
\label{app:cngp-pseudocode}

\begin{algorithm}[h]
\caption{Comparison-Normalized Gradient Play (CNGP)}
\label{alg:cngp}
\begin{algorithmic}[1]
\STATE Initialize $z_{i,1}\in\mathcal X_i$ for every player $i\in[N]$.
\FOR{$t=1,\ldots,T$}
    \FOR{each player $i$, in parallel}
        \STATE Set
        $x_{i,t} = \mathsf Q_{i,\delta_t}(z_{i,t}).$
        \STATE Sample $U_{i,t}\sim\operatorname{Unif}
        (\mathbb S^{d_i-1})$.
        \STATE Compare
        $(x_{i,t}+\delta_tU_{i,t},x_{-i,t})$
        and 
        $(x_{i,t}-\delta_tU_{i,t},x_{-i,t})$,
        and observe $Y_{i,t}$.
        \STATE Set
        $\widehat D_{i,t} =
        \frac{Y_{i,t}}{2\nu_i\kappa_{d_i}}U_{i,t}.$
        \STATE Update
        $z_{i,t+1}
        =
        \Pi_{\mathcal X_i}
        \left(
        z_{i,t}-\eta_{i,t}\widehat D_{i,t}
        \right).$
    \ENDFOR
\ENDFOR
\end{algorithmic}
\end{algorithm}

\subsection{Proofs for Section~\ref{subsec:directional-reg-bound}: Directional regret bound from comparisons}

\subsubsection{Finite-resolution regret does not control infinitesimal ordinal regret}
\label{apdx:finite-res-to-infinitesimal}

The following result shows that even with exact comparisons and vanishing comparison radii, sublinear
finite-resolution ordinal regret does not in general imply sublinear infinitesimal ordinal directional regret.

\begin{tcolorbox}[colframe=white!, top=2pt,left=2pt,right=2pt,bottom=2pt] 
\begin{proposition}[Finite-resolution regret does not control infinitesimal
ordinal regret]
\label{prop:finite-vs-infinitesimal-regret}
There exist a compact convex action set $\mathcal X\subset\mathbb R$, a
$C^\infty$ strongly convex cost function $J:\mathcal X\to\mathbb R$, a
sequence of comparison radii $\delta_t$ converging to zero, and a sequence of actions
$(x_t)_{t\geq1}$ such that all symmetric comparison queries
$x_t\pm\delta_t$ are feasible and, for every $K\geq1$,
\[
\Reg^{\delta}_{2K}
\leq \frac{\pi^2}{512},
\qquad\text{while}\qquad
\Reg^{D}_{2K}\geq 2K.
\]
\end{proposition}
\end{tcolorbox}

\paragraph{Proof sketch.} We provide an overview of the proof to highlight the intuition. 
The construction exploits the fact that $D^\delta(x)\to D(x)$ only
pointwise away from stationary points. Consider
$J(x)=x^2-\frac14x^3$ on $\mathcal X=[-1,1]$. For every small $x>0$,
$J'(x)>0$, so the infinitesimal direction is $D(x)=1$. On the other hand,
the symmetric finite difference is
\[
J(x+\delta)-J(x-\delta)
=
2\delta
\left(
2x-\frac34x^2-\frac14\delta^2
\right).
\]
Thus the first-order contribution is of size $\delta x$, while the
finite-resolution correction is of size $\delta^3$. Choosing
$x=\Theta(\delta^2)$ makes these terms comparable, so the sign of the finite
comparison can differ from the infinitesimal gradient direction even as
$\delta\downarrow0$.

For each resolution $\varepsilon_k$, we choose two positive actions on opposite
sides of this sign-change threshold:
\[
x_{2k-1}=\frac{\varepsilon_k^2}{16},
\qquad
x_{2k}=\frac{\varepsilon_k^2}{4},
\qquad
\delta_{2k-1}=\delta_{2k}=\varepsilon_k.
\]
Then $D(x_{2k-1})=D(x_{2k})=1$, whereas the corresponding
finite-resolution directions satisfy
\[
D^{\delta_{2k-1}}(x_{2k-1})=-1,
\qquad
D^{\delta_{2k}}(x_{2k})=1.
\]
Hence the finite-resolution directions cancel pairwise, so the comparator term
in $\Reg^\delta$ cancels as well; the remaining contribution is of order
$\sum_k \varepsilon_k^2$. Taking $\varepsilon_k\asymp 1/k$ makes this sum
finite. By contrast, the infinitesimal directions are equal to $+1$ at every
round, so the fixed comparator $y=-1$ yields a unit-order advantage per round
and therefore linear $\Reg^D$. The example shows that vanishing comparison
radii alone do not ensure that the finite-resolution field tracks the
infinitesimal one along trajectories approaching the stationary set.

\begin{proof}
Consider the one-dimensional action set
$\mathcal X=[-1,1]$ 
and the cost function $J: \mathcal{X} \to \mathbb{R}$ defined for any $x \in \mathcal{X}$ by: 
\begin{equation}
J(x)=x^2-\frac14x^3.
\end{equation}
Since the first and second derivatives of $J$ are respectively given by 
$J'(x)=2x-\frac34x^2,
J''(x)=2-\frac32x,$
we have $J''(x)\geq 1/2$ for every $x\in[-1,1]$. Hence $J$ is
$1/2$-strongly convex on $\mathcal X$. Moreover, for every $x\in(0,1]$,
$J'(x) = x\left(2-\frac34x\right)>0,$ 
so the infinitesimal preference direction satisfies
$D(x)=1$ for any $x \in (0,1]\,.$ 

In dimension one, $\mathbb S^0=\{-1,+1\}$ and $\kappa_1=1$. Therefore,
for every feasible $x\pm\delta$,
\[
D^\delta(x)
=
\operatorname{sgn}\!\left(J(x+\delta)-J(x-\delta)\right).
\]
Indeed, if
$s=\operatorname{sgn}(J(x+\delta)-J(x-\delta))$, then the two possible
values of $S^\delta(x,U)U$, corresponding to $U=1$ and $U=-1$, are both
equal to $s$.

A direct calculation yields
\begin{equation}
J(x+\delta)-J(x-\delta)
=
4x\delta
-\frac14\left(6x^2\delta+2\delta^3\right) 
=
2\delta
\left(
2x-\frac34x^2-\frac14\delta^2
\right).
\label{eq:finite-inf-separation-central-diff}
\end{equation}

For every $k\geq1$, let
$
\varepsilon_k:=\frac{1}{4k},
\delta_{2k-1}=\delta_{2k}:=\varepsilon_k,$ 
and define
\[
x_{2k-1}:=\frac{\varepsilon_k^2}{16},
\qquad
x_{2k}:=\frac{\varepsilon_k^2}{4}.
\]
Clearly, $\delta_t \to 0$, and all query points
$x_t\pm\delta_t$ belong to $\mathcal X$. Since $x_t>0$ for every $t$,
\[
D(x_t)=1.
\]

For the first action in each pair, using
\eqref{eq:finite-inf-separation-central-diff},
\begin{equation}
2x_{2k-1}
-\frac34x_{2k-1}^2
-\frac14\varepsilon_k^2
=
\frac{\varepsilon_k^2}{8}
-\frac{3\varepsilon_k^4}{1024}
-\frac{\varepsilon_k^2}{4}\\
=
-\frac{\varepsilon_k^2}{8}
-\frac{3\varepsilon_k^4}{1024}
<0,
\end{equation}
and hence
\[
D^{\delta_{2k-1}}(x_{2k-1})=-1.
\]
For the second action,
\begin{equation}
2x_{2k}
-\frac34x_{2k}^2
-\frac14\varepsilon_k^2
=
\frac{\varepsilon_k^2}{2}
-\frac{3\varepsilon_k^4}{64}
-\frac{\varepsilon_k^2}{4}
=
\frac{\varepsilon_k^2}{4}
-\frac{3\varepsilon_k^4}{64}
>0,
\end{equation}
where the last inequality follows from $\varepsilon_k\leq1/4$. Thus
\[
D^{\delta_{2k}}(x_{2k})=1.
\]
The finite-resolution directions therefore cancel pairwise:
\[
D^{\delta_{2k-1}}(x_{2k-1})
+
D^{\delta_{2k}}(x_{2k})
=0.
\]

For $T=2K$, the finite-resolution regret is consequently
\begin{align*}
\Reg^\delta_{2K}
&=
\sup_{y\in[-1,1]}
\sum_{t=1}^{2K}
D^{\delta_t}(x_t)(x_t-y)\\
&=
\sum_{k=1}^K
\left(
-x_{2k-1}+x_{2k}
\right)\\
&=
\frac{3}{16}
\sum_{k=1}^K\varepsilon_k^2\\
&=
\frac{3}{256}
\sum_{k=1}^K\frac{1}{k^2}\\
&\leq
\frac{3}{256}\frac{\pi^2}{6}
=
\frac{\pi^2}{512}.
\end{align*}

On the other hand, since $D(x_t)=1$ for every $t$,
\begin{align*}
\Reg^D_{2K}
=
\sup_{y\in[-1,1]}
\sum_{t=1}^{2K}(x_t-y)
=
\sum_{t=1}^{2K}x_t+2K,
\end{align*}
where the supremum is attained at $y=-1$. Therefore
\[
\Reg^D_{2K}\geq2K.
\]
This proves the claim.
\end{proof}

\subsubsection{Proof of Proposition~\ref{prop:comparison-field}}

\begin{proof}

\noindent\textbf{Proof of (i).} Observe that strictly increasing transformations preserve the sign, i.e., for any strictly increasing function $\phi_i: \mathbb{R} \to \mathbb{R}$, any $x \in \mathcal{X}, U_i$ s.t. $x_i \pm \delta U_i \in \mathcal{X}_i$,  
\begin{equation*}
    \operatorname{sgn}(J_i(x_i + \delta U_i, x_{-i}) - J_i(x_i - \delta U_i, x_{-i})) = 
    \operatorname{sgn}(\phi_i(J_i(x_i + \delta U_i, x_{-i})) - \phi_i(J_i(x_i - \delta U_i, x_{-i})))\,.
\end{equation*}
This implies that $S_i^{\delta}(x,U_i)$ is invariant under transformations of the cost functions $J_i$ by strictly increasing functions. As a consequence,  the \emph{$\delta$-comparison field} $D_i^\delta(x)
:=\frac{1}{\kappa_{d_i}}\mathbb E_{U_i}\left[S_i^\delta(x,U_i)U_i\right]$ defined in \eqref{eq:delta-comparison-dir} is also invariant.\\

\noindent\textbf{Proof of (ii).} 
The desired bound follows from: 
\begin{equation}
\|D_i^\delta(x)\|
=
\sup_{\|v\|=1}
\frac{1}{\kappa_{d_i}}
\mathbb E\!\left[
S_i^\delta(x,U_i)\langle U_i,v\rangle
\right]
\leq
\frac{1}{\kappa_{d_i}}
\sup_{\|v\|=1}
\mathbb E\!\left[
|\langle U_i,v\rangle|
\right]
=1\,,
\end{equation}
where the first equality uses the fact that $\|z\| = \sup_{\|v\| =1} \langle z,v \rangle$ (which follows from the Cauchy-Schwarz inequality) and the definition of $D_i^{\delta}(x)$ in \eqref{eq:delta-comparison-dir}, and the last inequality uses Lemma 3.1 in~\citet{fu-wang-yan23minimax-spherical} (which relies on rotational invariance).\\ 

\noindent\textbf{Proof of (iii).} Suppose $\nabla_i J_i(x) \neq 0.$ Then differentiability of $J_i(\cdot, x_{-i})$ gives, for every $u$,
\[
\frac{
J_i(x_i+\delta u,x_{-i})
-
J_i(x_i-\delta u,x_{-i})
}{2\delta}
\underset{\delta \to 0}{\longrightarrow}
\langle \nabla_iJ_i(x),u\rangle.
\]
For all $u$ outside the measure-zero great sphere $\{u:\langle \nabla_i J_i(x),u\rangle=0\}$,
\[
S_i^\delta(x,u)
\underset{\delta \to 0}{\longrightarrow}
\operatorname{sgn}(\langle \nabla_iJ_i(x),u\rangle).
\]
Dominated convergence therefore yields
\[
D_i^\delta(x) = \frac{1}{\kappa_{d_i}}
\mathbb E_{U_i}
\left[
S_i^\delta(x,U_i)U_i
\right]
\underset{\delta \to 0}{\longrightarrow}
\frac{1}{\kappa_{d_i}}
\mathbb E
\left[
\operatorname{sgn}(\langle \nabla_iJ_i(x),U_i\rangle) U_i
\right].
\]
By rotational symmetry, the expectation on the right must be collinear with $\nabla_i J_i(x)$. Taking its inner product with $\nabla_i J_i(x)/\|\nabla_i J_i(x)\|$ gives
$\mathbb E
\left|
\left\langle
U_i,\frac{\nabla_i J_i(x)}{\|\nabla_i J_i(x)\|}
\right\rangle
\right| =
\kappa_{d_i},$ and hence 
\begin{equation*}
\frac{1}{\kappa_{d_i}}
\mathbb E\left[\operatorname{sgn}(\langle \nabla_iJ_i(x),U_i\rangle) U_i\right] 
= \frac{\nabla_i J_i(x)}{\|\nabla_i J_i(x)\|} = D_i(x)\,.
\end{equation*}
\end{proof}

\subsubsection{Proof of Lemma~\ref{lem:one-bit-estimator}}

\begin{proof}
Conditional on $U_{i,t}$,
\begin{align}
\mathbb E[Y_{i,t}\mid U_{i,t},\mathcal F_t]
&=\mathbb E[\xi_{i,t} S_i^{\delta_t}(x_t,U_{i,t})\mid U_{i,t},\mathcal F_t] \nonumber\\
&= \left(\frac 12 + \nu_i\right) S_i^{\delta_t}(x_t,U_{i,t}) - \left(1 -\frac 12 - \nu_i \right)S_i^{\delta_t}(x_t,U_{i,t})\nonumber\\
&= 2\nu_i S_i^{\delta_t}(x_t,U_{i,t}).\nonumber
\end{align}
Hence
\[
\mathbb E[
\widehat D_{i,t}
\mid\mathcal F_t]
= \mathbb E\left[ \frac{Y_{i,t}}{2 \nu_i \kappa_{d_i}} U_{i,t}
\mid\mathcal F_t \right]
=
\frac{1}{\kappa_{d_i}}
\mathbb E_{U_{i,t}}
[
S_i^{\delta_t}(x_t,U_{i,t})U_{i,t}]
=
D_i^{\delta_t}(x_t).
\]
Moreover, $\|U_{i,t}\|=1$ and $|Y_{i,t}|\leq1$, so
\[
\|\widehat D_{i,t}\|^2
= \left\|\frac{Y_{i,t}}{2 \nu_i \kappa_{d_i}} U_{i,t} \right\|^2
\leq
\frac{1}{4\nu_i^2\kappa_{d_i}^2} = G_i^2\,.
\]
Taking conditional expectation concludes the proof. 
\end{proof}

\begin{remark}
For exact comparisons,
$G_i =\kappa_{d_i}^{-1} =
\sqrt{\frac{\pi d_i}{2}}\,(1+o(1)).$
\end{remark}

\subsubsection{Proof of Theorem~\ref{thm:comparison-regret}}

\begin{proof}
Let $y_i\in\mathcal X_i$. Projection nonexpansiveness gives
\begin{equation}
\|z_{i,t+1}-y_i\|^2
\leq
\|z_{i,t}-y_i\|^2
-2\eta_i
\langle\widehat D_{i,t},z_{i,t}-y_i\rangle
+
\eta_i^2\|\widehat D_{i,t}\|^2.
\end{equation}
Rearranging and summing,
\begin{equation}
\label{eq:pathwise-bound}
\sum_{t=1}^T
\langle\widehat D_{i,t},z_{i,t}-y_i\rangle
\leq
\frac{\|z_{i,1}-y_i\|^2}{2\eta_i}
+
\frac{\eta_i}{2}
\sum_{t=1}^T
\|\widehat D_{i,t}\|^2
\leq
\frac{R_i^2}{2\eta_i}
+
\frac{\eta_i}{2}
\sum_{t=1}^T
\|\widehat D_{i,t}\|^2.
\end{equation}
Now, introducing the sequence $(x_{i,t})$, we obtain: 
\begin{equation}
\sum_{t=1}^T
\langle\widehat D_{i,t},x_{i,t}-y_i\rangle 
=  \sum_{t=1}^T
\langle\widehat D_{i,t},z_{i,t}-y_i\rangle 
+ \sum_{t=1}^T
\langle\widehat D_{i,t},x_{i,t}-z_{i,t}\rangle\,.
\end{equation}
Taking expectations, using $\mathbb E[
\widehat D_{i,t}\mid\mathcal F_t]
=
D_i^{\delta_t}(x_t)$, $\mathbb E[
\|\widehat D_{i,t}\|^2
\mid\mathcal F_t]
\leq
G_i^2$ (Lemma~\ref{lem:one-bit-estimator}), \eqref{eq:pathwise-bound}, 
and noting that $z_{i,t}-y_i$ and $x_{i,t}-z_{i,t}$ are $\mathcal{F}_t$-measurable yields: 
\[
\mathbb E \left[
\sum_{t=1}^T
\langle
D_i^{\delta_t}(x_t),
x_{i,t}-y_i
\rangle \right]
\leq
\frac{R_i^2}{2\eta_i}
+
\frac{\eta_iG_i^2T}{2} 
+ \mathbb E \left[\sum_{t=1}^T \langle D_i^{\delta_t}(x_t), x_{i,t} - z_{i,t} \rangle\right]\,.
\]
Since this is true for every fixed $y_i \in \mathcal{X}_i$, taking the supremum gives: 
\begin{equation}
\overline{\Reg}_{i,T}^{\delta}
\leq
\frac{R_i^2}{2\eta_i}
+
\frac{\eta_i G_i^2T}{2} 
+  \mathbb E \left[\sum_{t=1}^T \langle D_i^{\delta_t}(x_t), x_{i,t} - z_{i,t} \rangle \right]\,.
\end{equation}
To obtain the first desired inequality, it remains to use Cauchy-Schwarz inequality for the last inner product, with the fact that $\|D_i^{\delta}(x_t)\| \leq 1$ (Proposition~\ref{prop:comparison-field}) and $\|x_{i,t}-z_{i,t}\| \le \frac{R_i}{\rho_i}\delta_t$. We obtain: 
\begin{equation}
\overline{\Reg}_{i,T}^{\delta}
\leq
\frac{R_i^2}{2\eta_i}
+
\frac{\eta_i G_i^2T}{2} 
+ \frac{R_i}{\rho_i} \sum_{t=1}^T \delta_t\,.
\end{equation}
Setting $\eta_i=\frac{R_i}{G_i\sqrt T}$ and $\delta_t = \frac{\delta}{\sqrt{T}}$ (with $\delta \in (0, \rho)$) 
gives:
\begin{equation}
\overline{\Reg}_{i,T}^{\delta}
\leq
R_iG_i\sqrt T 
+ R_i \frac{\rho}{\rho_i} \sqrt{T}
\leq R_i(G_i +1) \sqrt{T}
= O\!\left(
\frac{R_i\sqrt{d_iT}}{\nu_i}
\right)\,.
\end{equation}

To control the realized regret, define for all $i \in [N]$, the martingale difference sequence $(M_{i,t})$:  
\[
M_{i,t}
:= 
D_i^{\delta_t}(x_t)-\widehat D_{i,t}\,, \quad \forall t \geq 0\,.
\]
For an arbitrary reference point $y_i^0\in\mathcal X_i$, every $y_i\in\mathcal X_i$ and any $T > 0$, 
\[
\sum_{t=1}^T
\langle
M_{i,t},
x_{i,t}-y_i
\rangle
=
\sum_{t=1}^T
\langle
M_{i,t},
x_{i,t}-y_i^0
\rangle
+
\left\langle
\sum_{t=1}^TM_{i,t},
y_i^0-y_i
\right\rangle.
\]
Hence, for any $i \in [N]$ and any $T > 0$,
\begin{align}
\Reg_{i,T}^{\delta} &=  \sup_{y_i \in \mathcal{X}_i}
\sum_{t=1}^T
\langle
D_{i}^{\delta_t}(x_t),
x_{i,t}-y_i
\rangle \nonumber\\
&= \sup_{y_i \in \mathcal{X}_i}
\sum_{t=1}^T
\langle
\widehat{D}_{i,t},
x_{i,t}-y_i
\rangle + \sum_{t=1}^T
\langle
M_{i,t},
x_{i,t}-y_i
\rangle \nonumber\\
&= \sup_{y_i \in \mathcal{X}_i}
\sum_{t=1}^T
\langle
\widehat{D}_{i,t},
x_{i,t}-y_i
\rangle 
+ \sum_{t=1}^T
\langle
M_{i,t},
x_{i,t}-y_i^0
\rangle + 
\sum_{t=1}^T
\langle
M_{i,t},
y_i^0-y_i
\rangle   \nonumber\\
&\leq \sup_{y_i \in \mathcal{X}_i}
\sum_{t=1}^T
\langle
\widehat D_{i,t},
x_{i,t}-y_i
\rangle
+
A_T
+
R_i \cdot\left\|
\sum_{t=1}^TM_{i,t}
\right\|,  \label{eq:master-realized-regret}
\end{align}
where $A_T
:=  \sum_{t=1}^T
\langle M_{i,t},x_{i,t}-y_i^0 \rangle.$

By the pathwise projection inequality proved above in \eqref{eq:pathwise-bound}, 
\[ \sup_{y_i\in\mathcal X_i} \sum_{t=1}^T \langle \widehat D_{i,t}, x_{i,t}-y_i \rangle \leq \frac{R_i^2}{2\eta_i} + \frac{\eta_i}{2} \sum_{t=1}^T \|\widehat D_{i,t}\|^2 
+ \sum_{t=1}^T \langle \widehat D_{i,t}, x_{i,t}-z_{i,t} \rangle
\] 
Moreover, by construction,
$\|x_{i,t}-z_{i,t}\|
\le
\frac{R_i}{\rho_i}\delta_t,$
and, from the definition of the single-comparison estimator,
$\|\widehat D_{i,t}\|
=
\frac{|Y_{i,t}|}{2\nu_i\kappa_{d_i}}
\le
G_i$ almost surely.
Hence, using Cauchy-Schwarz inequality, we have: 
\begin{equation}
\sum_{t=1}^T
\langle
\widehat D_{i,t},
x_{i,t}-z_{i,t}
\rangle
\le 
\sum_{t=1}^T \|\widehat D_{i,t}\| 
\cdot \|x_{i,t}-z_{i,t}\|
\le
\frac{R_iG_i}{\rho_i}
\sum_{t=1}^T\delta_t.
\end{equation}
Moreover, since $x_{i,t}$ is $\mathcal F_t$-measurable and $\mathbb E[M_{i,t}\mid\mathcal F_t]=0$, $\mathbb E[A_T]=0.$

We also have, using
$M_{i,t} = D_i^{\delta_t}(x_t)-\widehat D_{i,t},
D_i^{\delta_t}(x_t) = \mathbb E[\widehat D_{i,t}\mid\mathcal F_t]$,
we obtain
\[
\begin{aligned}
\mathbb E[
\|M_{i,t}\|^2
\mid\mathcal F_t]
&=
\mathbb E\!\left[
\left\|
\widehat D_{i,t}
-
\mathbb E[\widehat D_{i,t}\mid\mathcal F_t]
\right\|^2
\Bigm|\mathcal F_t
\right] \\
&=
\mathbb E[
\|\widehat D_{i,t}\|^2
\mid\mathcal F_t]
-
\left\|
\mathbb E[
\widehat D_{i,t}\mid\mathcal F_t]
\right\|^2 \\
&\leq
\mathbb E[
\|\widehat D_{i,t}\|^2
\mid\mathcal F_t]
\leq
G_i^2\,, 
\end{aligned}
\]
where the last inequality follows from using Lemma~\ref{lem:one-bit-estimator}. 

By martingale orthogonality, 
\[ \mathbb E \left\| \sum_{t=1}^T M_{i,t} \right\|^2 = \sum_{t=1}^T \mathbb E\|M_{i,t}\|^2 \leq G_i^2T. \] 

Therefore, by Cauchy-Schwarz inequality, 
\[ \mathbb E \left\| \sum_{t=1}^T M_{i,t} \right\| \leq G_i\sqrt T. \] 

Taking expectations in \eqref{eq:master-realized-regret} therefore yields 
\[ \mathbb E[\Reg_{i,T}^{\delta}] \leq \frac{R_i^2}{2\eta_i} + \frac{\eta_iG_i^2T}{2} + R_iG_i\sqrt T +
\frac{R_iG_i}{\rho_i}
\sum_{t=1}^T\delta_t. \] 
With $\eta_i=\frac{R_i}{G_i\sqrt T},$ the first two terms sum to $R_iG_i\sqrt T$, and with $\delta_t = \frac{\delta}{\sqrt{T}}$ the last term is bounded by $R_i G_i \sqrt{T}$ as well. Hence, 
\[\mathbb E[\Reg_{i,T}^{\delta}] \leq 3R_iG_i\sqrt T = O\!\left( \frac{R_i\sqrt{d_iT}}{\nu_i} \right). 
\] 

Finally, we bound the positive part. Since the right-hand side of \eqref{eq:master-realized-regret}, apart from $A_T$, consists of nonnegative terms, 
\[ [\Reg_{i,T}^{\delta}]_+ \leq \frac{R_i^2}{2\eta_i} + \frac{\eta_i}{2} \sum_{t=1}^T \|\widehat D_{i,t}\|^2 
+\frac{R_iG_i}{\rho_i}\sum_{t=1}^T\delta_t
+|A_T| + R_i \left\| \sum_{t=1}^T M_{i,t} \right\|. \] 
It remains to control $A_T$. Note first that its increments form a scalar martingale difference sequence. 
To see this, define 
$Z_t :=\left\langle M_{i,t},x_{i,t}-y_i^0 \right\rangle,$ so that $A_T=\sum_{t=1}^T Z_t.$ 
Since $x_{i,t}$ is $\mathcal F_t$-measurable and
$\mathbb E[M_{i,t}\mid\mathcal F_t]=0,$ 
we have $\mathbb E[Z_t\mid\mathcal F_t] = \left\langle \mathbb E[M_{i,t}\mid\mathcal F_t],
x_{i,t}-y_i^0 \right\rangle =0.$
Thus $(Z_t)$ is a scalar martingale difference sequence.
Expanding the square gives
\[
A_T^2
=
\sum_{t=1}^T Z_t^2
+
2\sum_{1\leq s<t\leq T} Z_sZ_t.
\]
For $s<t$, $Z_s$ is $\mathcal F_t$-measurable, so by the tower property,
$\mathbb E[Z_sZ_t] = 
\mathbb E\!\left[
\mathbb E[Z_sZ_t\mid\mathcal F_t]
\right] = \mathbb E\!\left[
Z_s\,\mathbb E[Z_t\mid\mathcal F_t]
\right]
=0.$
Therefore the cross terms vanish, and
\[
\mathbb E[A_T^2]
=
\sum_{t=1}^T
\mathbb E[Z_t^2].
\]

By Cauchy--Schwarz inequality,
\[
Z_t^2
=
\left\langle
M_{i,t},
x_{i,t}-y_i^0
\right\rangle^2
\leq
\|M_{i,t}\|^2
\|x_{i,t}-y_i^0\|^2.
\]
Since $x_{i,t},y_i^0\in\mathcal X_i$,
$\|x_{i,t}-y_i^0\|
\leq
R_i,$
and hence
\[
\mathbb E[A_T^2]
\leq
R_i^2
\sum_{t=1}^T
\mathbb E\|M_{i,t}\|^2 \leq R_i^2G_i^2T.
\]
Therefore, by Cauchy--Schwarz inequality,
\[
\mathbb E[|A_T|]
\leq
\sqrt{\mathbb E[A_T^2]}
\leq
R_iG_i\sqrt T.
\]

Thus, it follows that: 
\[ \mathbb E\bigl[[\Reg_{i,T}^{\delta}]_+\bigr] \leq \frac{R_i^2}{2\eta_i} + \frac{\eta_iG_i^2T}{2} + 2R_iG_i\sqrt T 
+\frac{R_iG_i}{\rho_i}\sum_{t=1}^T\delta_t
. \] 
Choosing again $\eta_i=\frac{R_i}{G_i\sqrt T}$ and $\delta_t = \frac{\delta}{\sqrt{T}}$ with $\delta \in (0, \rho)$ gives 
\[\mathbb E\bigl[[\Reg_{i,T}^{\delta}]_+\bigr] \leq 4R_iG_i\sqrt T = O\!\left( \frac{R_i\sqrt{d_iT}}{\nu_i} \right).\] 
This proves the claimed bound.
\end{proof}

\subsubsection{Trading comparisons for dimension dependence using minibatch}
\label{subsubsec:trading-comp-for-dim}

The $\sqrt{d_i}$ factor in regret comes from using only one random directional comparison per iteration. It can be reduced by batching comparisons.

Suppose player $i$ draws $m_i$ independent directions and averages the corresponding estimators:
\[
\widehat D_{i,t}^{(m_i)}
=
\frac1{m_i}
\sum_{j=1}^{m_i}
\widehat D_{i,t}^{(j)}.
\]
Since $\|D_i^\delta(x)\|\leq1$,
\[
\begin{aligned}
\mathbb E
\left[
\|\widehat D_{i,t}^{(m_i)}\|^2
\mid\mathcal F_t
\right]
&=
\|D_i^\delta(x_t)\|^2
+
\frac1{m_i}
\mathbb E
\left[
\|\widehat D_{i,t}-D_i^\delta(x_t)\|^2
\mid\mathcal F_t
\right] \leq
1+
\frac{G_i^2}{m_i}.
\end{aligned}
\]
The proof of Theorem~\ref{thm:comparison-regret} therefore gives 
\[
\Reg_{i,T}^{\delta}
=
O\!\left(
R_i
\sqrt{
T\left(
1+
\frac{d_i}{m_i\nu_i^2}
\right)}
\right)
\]
in expectation, up to universal constants. Thus $m_i=1$ gives $\Reg_{i,T}^{\delta} = O\!\left(
\frac{R_i\sqrt{d_iT}}{\nu_i}
\right)$ whereas $m_i=\Theta(d_i/\nu_i^2)$ gives $\Reg_{i,T}^{\delta} = O(R_i\sqrt T)$
in terms of iterations, at the cost of $\Theta(d_i/\nu_i^2)$ comparisons per iteration. 

\begin{remark} 
To compare the finite-resolution guarantee with the infinitesimal benchmark,
let $b_{i,t}:= \|D_i^{\delta_t}(x_t)-D_i(x_t)\|.$
Then $\Reg_{i,T}^{D} \leq \Reg_{i,T}^{\delta}
+R_i\sum_{t=1}^T b_{i,t}.$
Uniform control of $b_{i,t}$ is impossible in general near vanishing
gradients; the normalized infinitesimal field is nonsmooth in general and may change sharply even when the
underlying utilities are smooth. This is why $D^\delta$ is the natural
oracle-matched benchmark at finite comparison resolution. 
\end{remark}

\subsection{Proofs for Section~\ref{subsec:as-cvg-Nash-ordinal-pot}}

\subsubsection{Proof of Theorem~\ref{thm:as-comparison-ne}: Almost sure convergence to Nash in ordinal potential games}

The proof relies on the projected differential inclusion stochastic approximation framework of
\citet{majewski-miasojedow-moulines18}.

We use the following $C^1$ specialization of
\citet[Theorem~4.2]{majewski-miasojedow-moulines18}. The limiting
dynamics remain a projected differential inclusion; the specialization
only replaces the general set-valued Lie derivative by the ordinary
directional derivative of the smooth Lyapunov function used here.

\begin{tcolorbox}[colframe=white!, top=2pt,left=2pt,right=2pt,bottom=2pt] 
\begin{theorem}[$C^1$ specialization of Theorem 4.2 in \citet{majewski-miasojedow-moulines18}]
\label{thm:majewski-et-al4.2}
Let $\mathcal K\subset\mathbb R^d$ be nonempty, compact, and convex,
let $\mathcal U\supseteq\mathcal K$ be open, and let
$F:\mathcal U\rightrightarrows\mathbb R^d$ be a set-valued mapping. 
Consider a recursion of the form
\begin{equation}
\label{eq:generic-ppad}
z_{t+1}
=
\Pi_{\mathcal K}
\left[
z_t+\eta_t
\bigl(f_t+e_{t+1}+r_t\bigr)
\right],
\qquad
f_t\in F(y_t)\,,
\end{equation}
where the set-valued mapping $F$, the stepsize sequence $(\eta_t)$, the noise sequence $(e_t)$ and the bias sequence $(r_t)$ satisfy the following assumptions:
\begin{enumerate}[label=(A\arabic*)]
\item \label{as:set-val-mapping}
$F$ is upper hemicontinuous, locally bounded, and has nonempty
compact convex values. 
\item \label{as:stepsize}
$\eta_t>0$, $\eta_t\to0$, and
$\sum_{t=1}^{\infty}\eta_t=\infty.$
\item \label{as:bias-and-stoch-error}
$r_t\to0$ and
$\sum_{t=1}^{\infty}\eta_t e_{t+1}$ 
converges.
\item \label{as:state-y-to-zt}
$y_t\in\mathcal U$ and
$\|y_t-z_t\|\to0.$
\item \label{as:bounded-error}
$\sup_t\|e_{t+1}\|<\infty$.
\end{enumerate} 

Suppose that $V\in C^1(\mathcal U)$ and that there exists an
upper semicontinuous function $U:\mathcal U\to\mathbb R$ such that
\begin{equation}
\label{eq:majewski-specialized-Lyapunov}
\left\langle
\nabla V(x),
\Pi_{T_{\mathcal K}(x)}u
\right\rangle
\le
U(x)
\le0,
\qquad
\forall x\in\mathcal K,\quad
\forall u\in F(x).
\end{equation}
Let $\mathcal S:=\{x\in\mathcal K:U(x)=0\}.$ 
Then the image under $V$ of the limit set of $(z_t)$ is a compact
interval contained in $V(\mathcal S)$. If
$\operatorname{int}_{\mathbb R}V(\mathcal S)=\varnothing,$
then
\[
\operatorname{dist}(z_t,\mathcal S)\longrightarrow0.
\]
\end{theorem}
\end{tcolorbox}

Under the empty-interior condition, the preceding theorem also implies
that $V(z_t)$ converges. Indeed, the image under $V$ of the limit set
of $(z_t)$ is an interval contained in $V(\mathcal S)$. Since
$V(\mathcal S)$ has empty interior, this interval must be a singleton,
say $\{V_\infty\}$. Compactness of $\mathcal K$ and continuity of $V$
then imply that $V(z_t)\to V_\infty$.\\

\noindent\textbf{Proof outline.}
To apply Theorem~\ref{thm:majewski-et-al4.2}, we proceed in four steps:
\begin{enumerate}
    
\item \textbf{Completion of the limiting normalized field.}
We complete the normalized potential field at zero-gradient blocks and
show that the resulting set-valued field is upper hemicontinuous,
locally bounded, and has nonempty compact convex values
(definition in~\eqref{def:completed-potential-field} and
Lemma~\ref{lem:completed-field-regularity}).

\item \textbf{Graph consistency and stochastic-approximation representation.}
We show that the finite-resolution comparison fields approach the
completed limiting field in the sense of graphs
(Lemma~\ref{lem:graph-consistency-finite-resolution}). This allows us
to represent our \eqref{eq:cngp} algorithm as a projected perturbed approximate discretization
of the projected differential inclusion
\[
\dot{x}(t)
\in
\Pi_{T_{\mathcal X}(x(t))}F(x(t))\,,
\]
see Lemma~\ref{lem:cngp-ppad}. 

\item \textbf{Lyapunov analysis and potential stationarity.}
We show that the ordinal potential $\Phi$ is a Lyapunov function for
the projected limiting inclusion and identify the zero set of its
dissipation with the constrained stationary points of the potential
(Lemma~\ref{lem:lyapunov-ineq}).

\item \textbf{From potential stationarity to Nash equilibrium.}
We establish first-order alignment between the players' costs and the
ordinal potential under Assumption~\ref{as:non-degeneracy-ord-pot} and, using blockwise pseudoconvexity, show that the
potential stationary set coincides with the Nash equilibrium set
(Lemmas~\ref{lem:ordinal-potential-gradient-alignment}
and~\ref{lem:potential-stationarity-ne}). Theorem~\ref{thm:majewski-et-al4.2}
then yields the desired almost-sure convergence result.
\end{enumerate}

Let $\mathcal U$ be the open neighborhood of $\mathcal X$ on which $\Phi$ is $C^1$. 
For every $i\in[N]$, let
$g_i(x):=\nabla_i\Phi(x),$
and define
$\mathcal D_i:\mathcal U\rightrightarrows\mathbb R^{d_i}$ by
\begin{equation}
\label{def:completed-potential-field}
\mathcal D_i(x)
:=
\begin{cases}
\left\{
\dfrac{g_i(x)}{\|g_i(x)\|}
\right\},
&
g_i(x)\neq0,
\\[2ex]
\overline{\mathbb B}_{d_i},
&
g_i(x)=0.
\end{cases}
\end{equation}
Set $
\mathcal D(x)
:= \prod_{i=1}^N\mathcal D_i(x)\,,
F(x):=-\mathcal D(x).$\\

At zero-gradient blocks, the unit ball provides a closed convex
completion of the normalized field, while away from such blocks the
completion reduces to the unique normalized gradient direction. The
graph-consistency result below (Lemma~\ref{lem:graph-consistency-finite-resolution}) shows that this completion captures all
asymptotic cluster behavior of the finite-resolution comparison fields
as the comparison radius vanishes.

We start by checking Assumption~\ref{as:set-val-mapping} regarding the regularity of the set-valued mapping $F$.

\begin{tcolorbox}[colframe=white!, top=2pt,left=2pt,right=2pt,bottom=2pt] 
\begin{lemma}[Regularity of the completed field]
\label{lem:completed-field-regularity}
The map
$F=-\mathcal D:\mathcal U\rightrightarrows\mathbb R^d$ defined above in \eqref{def:completed-potential-field}  
is upper hemicontinuous, locally bounded, and has nonempty compact
convex values. Moreover,
$\sup_{v\in\mathcal D(x)}\|v\|
\le\sqrt N.$
\end{lemma}
\end{tcolorbox}

\begin{proof}
Nonemptiness, compactness, convexity, and uniform boundedness are
immediate from the definition.

It remains to establish upper hemicontinuity. Since the values are
compact and the map is locally bounded, it suffices in finite
dimension to prove that its graph is closed.

Let $x_n\to x$ and $v_n\to v$ with
$v_n\in\mathcal D(x_n)$. Let $i \in [N]$. 
If $g_i(x)\neq0$, continuity of $g_i$ gives $g_i(x_n)\neq0$ for sufficiently large $n$, and hence  
\begin{equation*}
    v_{i,n}
=
\frac{g_i(x_n)}{\|g_i(x_n)\|}
\longrightarrow
\frac{g_i(x)}{\|g_i(x)\|}.
\end{equation*}
Thus $v_i\in\mathcal D_i(x)$.
If $g_i(x)=0$, then $\|v_{i,n}\|\le1$, so
$\|v_i\|\le1$ and $v_i\in\overline{\mathbb B}_{d_i}
=\mathcal D_i(x)$.
Thus $v\in\mathcal D(x)$ and the graph is closed.
\end{proof}

\begin{tcolorbox}[colframe=white!, top=2pt,left=2pt,right=2pt,bottom=2pt] 
\begin{lemma}[Graph consistency]
\label{lem:graph-consistency-finite-resolution}
As $\delta\downarrow0$,
\[
\sup_{x\in\mathcal X^\delta}
\operatorname{dist}
\left(
(x,D^\delta_\Phi(x)),
\operatorname{gph}(\mathcal D)
\right)
\longrightarrow0\,,
\]
where
$\operatorname{gph}(\mathcal D)
:= \{(x,v)\in \mathcal X\times\mathbb R^d:\ v\in\mathcal D(x)\}$ 
is the graph of the set-valued mapping $\mathcal D$.
\end{lemma}
\end{tcolorbox}

\begin{proof}
Suppose, towards a contradiction, that the graph-consistency statement
fails. Then there exist $\varepsilon>0$, a sequence
$\delta_n\downarrow0$, and $x_n\in\mathcal X^{\delta_n}$ such that
\begin{equation}
\label{eq:contradiction-graph-cv}
\operatorname{dist}
\left(
(x_n,D_\Phi^{\delta_n}(x_n)),
\operatorname{gph}(\mathcal D)
\right)
\ge\varepsilon,
\qquad
\forall n.
\end{equation}

Since $\mathcal X$ is compact, after passing to a subsequence we may
assume that
$x_n\to\bar x\in\mathcal X.$ 
Moreover, since
$\|D_\Phi^{\delta_n}(x_n)\|\le\sqrt N,$ 
the sequence $(D_\Phi^{\delta_n}(x_n))$ is bounded. After extracting
a further subsequence, we may therefore assume that
$D_\Phi^{\delta_n}(x_n)\to v$
for some
$v=(v_1,\ldots,v_N)\in\mathbb R^d.$

We claim that
$v\in\mathcal D(\bar x).$
To prove this, let $i\in[N]$. 

If $\nabla_i\Phi(\bar x)=0$, then
$\|D_i^{\delta_n}(x_n)\|\le1$
for all $n$. 
Passing to the limit gives
$\|v_i\|\le 1.$
Hence
$v_i\in\overline{\mathbb B}_{d_i}
=\mathcal D_i(\bar x).$

Suppose now that $\nabla_i\Phi(\bar x)\neq0$.
By continuity of $\nabla_i\Phi$,
$\nabla_i\Phi(x_n)\neq0$ for all sufficiently large $n$. Define
\[
q_n
:=
\frac{\nabla_i\Phi(x_n)}
{\|\nabla_i\Phi(x_n)\|}.
\]
By Lemma~\ref{lem:uniform-potential-alignment-as},
\[
\left\langle
q_n,D_i^{\delta_n}(x_n)
\right\rangle
\ge
1-
\frac{L\delta_n}
{\kappa_{d_i}\|\nabla_i\Phi(x_n)\|}.
\]
Since $\|D_i^{\delta_n}(x_n)\|\le1$, we have
\begin{align*}
\left\|
D_i^{\delta_n}(x_n)-q_n
\right\|^2
&=
\|D_i^{\delta_n}(x_n)\|^2+1
-
2\left\langle
D_i^{\delta_n}(x_n),q_n
\right\rangle
\\
&\le
2\left(
1-
\left\langle
D_i^{\delta_n}(x_n),q_n
\right\rangle
\right)
\\
&\le
\frac{2L\delta_n}
{\kappa_{d_i}\|\nabla_i\Phi(x_n)\|}
\longrightarrow0 \,,
\end{align*}
where the above limit follows from noting that 
$\delta_n\to0$ while
$\|\nabla_i\Phi(x_n)\|
\longrightarrow
\|\nabla_i\Phi(\bar x)\|>0.$ 
Therefore,
$D_i^{\delta_n}(x_n)-q_n\to0.$ 
Since 
$q_n
\longrightarrow
\frac{\nabla_i\Phi(\bar x)}
{\|\nabla_i\Phi(\bar x)\|},$ 
and
$D_i^{\delta_n}(x_n)\to v_i,$ 
we conclude that
\[
v_i
=
\frac{\nabla_i\Phi(\bar x)}
{\|\nabla_i\Phi(\bar x)\|}
\in
\mathcal D_i(\bar x).
\]

Since $i\in[N]$ was arbitrary, we have shown that $v_i\in\mathcal D_i(\bar x)$ for all $i\in [N]$.
Hence $v\in \prod_{i=1}^N\mathcal D_i(\bar x)
=\mathcal D(\bar x),$ 
so
$(\bar x,v)\in\operatorname{gph}(\mathcal D).$ 

Consequently,
\[
\operatorname{dist}\left(
(x_n,D_\Phi^{\delta_n}(x_n)),
\operatorname{gph}(\mathcal D)
\right)
\le
\left\|
(x_n,D_\Phi^{\delta_n}(x_n))-(\bar x,v)
\right\|
\longrightarrow0,
\]
because
$x_n\to\bar x,
D_\Phi^{\delta_n}(x_n)\to v.$ 
This contradicts
\eqref{eq:contradiction-graph-cv}. Therefore, as $\delta\downarrow0$, 
\[
\sup_{x\in\mathcal X^\delta}
\operatorname{dist}
\left(
(x,D_\Phi^\delta(x)),
\operatorname{gph}(\mathcal D)
\right)
\longrightarrow0\,.
\]
\end{proof}

\begin{tcolorbox}[colframe=white!, top=2pt,left=2pt,right=2pt,bottom=2pt] 
\begin{lemma}[CNGP tracks the completed normalized inclusion]
\label{lem:cngp-ppad}
Under the assumptions of
Theorem~\ref{thm:as-comparison-ne}, there exist sequences
$y_t\in\mathcal X$, $v_t\in\mathcal D(y_t)$, $r_t\to0$, and a
bounded martingale-difference sequence $e_{t+1}$ such that, almost
surely,
\[
z_{t+1}
=
\Pi_{\mathcal X}
\left[
z_t+\eta_t(-v_t+e_{t+1}+r_t)
\right],
\qquad
\|y_t-z_t\|\to0,
\]
and
$\sum_{t=1}^{\infty}\eta_t e_{t+1}$
converges. 
Consequently, $(z_t)$ is an $\mathcal X$-projected perturbed
approximate discretization of $F=-\mathcal D$ (in the sense of Theorem~\ref{thm:majewski-et-al4.2}). 
\end{lemma}
\end{tcolorbox}

\begin{proof}
Recall first that \eqref{eq:cngp} can be written as follows: 
\begin{equation}
\label{eq:recall-algo-cngp}
z_{t+1} = \Pi_{\mathcal X}(z_t - \eta_t \hat{D}_t)\,.
\end{equation}
Define $e_{t+1} := - M_{t+1} = - (\widehat D_t
-
D_\Phi^{\delta_t}(x_t))$. We can then write: 
\begin{equation}
\label{eq:decomp1}
- \hat{D}_t = - (\hat{D}_t - D_{\Phi}^{\delta_t}(x_t)) - D_{\Phi}^{\delta_t}(x_t) = - D_{\Phi}^{\delta_t}(x_t) + e_{t+1}\,.
\end{equation}

Our next step is to decompose the term $- D_{\Phi}^{\delta_t}(x_t)$ using the graph consistency result (Lemma~\ref{lem:graph-consistency-finite-resolution}) relating it to our completed set-valued mapping $-\mathcal{D}$ defined in \eqref{def:completed-potential-field}.  

Define
\[
\varepsilon_t
:=
\sup_{x\in\mathcal X^{\delta_t}}
\operatorname{dist}
\left(
(x,D_\Phi^{\delta_t}(x)),
\operatorname{gph}(\mathcal D)
\right).
\]
By Lemma~\ref{lem:graph-consistency-finite-resolution} and
$\delta_t\to0$, we have $\varepsilon_t\to0$. Since
$x_t\in\mathcal X^{\delta_t}$,
\begin{equation}
\label{eq:dist}
\operatorname{dist}
\left(
(x_t,D_\Phi^{\delta_t}(x_t)),
\operatorname{gph}(\mathcal D)
\right)
\le \varepsilon_t.
\end{equation}
Since $\mathcal X$ is compact and
$\sup_{v\in\mathcal D(x)}\|v\|\le\sqrt N$, the closed graph
$\operatorname{gph}(\mathcal D)$ is a compact subset of
$\mathcal X\times\mathbb R^d$. 
Hence, the distance in \eqref{eq:dist}
is attained and there exist $y_t\in\mathcal X$ and
$v_t\in\mathcal D(y_t)$ such that
\[
\left\|
(x_t,D_\Phi^{\delta_t}(x_t))-(y_t,v_t)
\right\|
=
\operatorname{dist}\left(
(x_t,D_\Phi^{\delta_t}(x_t)),
\operatorname{gph}(\mathcal D)
\right)
\le\varepsilon_t.
\]
In particular,
\begin{equation}
\label{eq:limits-r_t-y_t-x_t}
\|y_t-x_t\|\le\varepsilon_t \longrightarrow 0,
\qquad
\|v_t-D_\Phi^{\delta_t}(x_t)\|\le\varepsilon_t \longrightarrow 0.
\end{equation}

We define now $r_t := v_t - D_\Phi^{\delta_t}(x_t)$ and we use \eqref{eq:decomp1} to obtain: 
\begin{equation}
\label{eq:decomp-hatDt}
- \hat{D}_t = - D_{\Phi}^{\delta_t}(x_t) + e_{t+1} 
= - v_t + e_{t+1} + r_t\,, \quad v_t \in \mathcal{D}(y_t)\,.
\end{equation}
Therefore, using \eqref{eq:recall-algo-cngp} and the above decomposition~\eqref{eq:decomp-hatDt}, we obtain: 
\[
z_{t+1}
=
\Pi_{\mathcal X}
\left[
z_t+
\eta_t
\bigl(
f_t + e_{t+1} + r_t
\bigr)
\right]\,,
\]
where $f_t := - v_t \in F(y_t) = - \mathcal{D}(y_t), e_{t+1} = - M_{t+1} = - (\widehat D_t
-
D_\Phi^{\delta_t}(x_t)), r_t = v_t - D_{\Phi}^{\delta_t}(x_t)$. 

We now check that the assumptions \ref{as:set-val-mapping} to \ref{as:bounded-error} in Theorem~\ref{thm:majewski-et-al4.2} are satisfied. 

Assumption~\ref{as:set-val-mapping} is satisfied by Lemma~\ref{lem:completed-field-regularity}. 

Assumption~\ref{as:stepsize} follows from
Assumption~\ref{as:rm-stepsizes}; indeed
$\sum_t\eta_t^2<\infty$ implies $\eta_t\to0$.

Assumption~\ref{as:bias-and-stoch-error}: 
Since the comparison field of the game and that of the potential
coincide exactly at finite resolution,
Lemma~\ref{lem:one-bit-estimator} gives
$\mathbb E[M_{t+1}\mid\mathcal F_t]=0.$

Moreover, the single-comparison estimator is uniformly bounded:
$
\|\widehat D_t\|
\le
G:=
\left(
\sum_{i=1}^N G_i^2
\right)^{1/2}
$ almost surely. 
Consequently,
$\|M_{t+1}\|\le \|\hat{D}_t\| + \|D_{\Phi}^{\delta_t}(x_t)\| \leq G + \sqrt{N}$ almost surely. Therefore, 
\begin{equation*}
\sum_{t=1}^{+\infty} \mathbb{E}[\|\eta_tM_{t+1}\|^2|\mathcal{F}_t] \leq (G + \sqrt{N})^2 \sum_{t=1}^{+\infty} \eta_t^2 < \infty. 
\end{equation*}
Since $\sum_{t=1}^{\infty}\eta_t^2<\infty$ by Assumption~\ref{as:rm-stepsizes}, 
the vector-valued martingale series 
$\sum_{t=1}^T\eta_t M_{t+1}$ 
converges almost surely. Noting that $e_{t+1} = - M_{t+1}$, it follows that
$\sum_{t=1}^{\infty}\eta_t e_{t+1}$ 
converges almost surely, as required by
Assumption~\ref{as:bias-and-stoch-error}. 
We henceforth fix a sample path in the probability-one event on
which the preceding martingale series converges and the uniform
bounds hold. 

Recall from \eqref{eq:limits-r_t-y_t-x_t} that $r_t\to0$.  
Furthermore, recall that by construction,
$x_{i,t}=
\mathsf Q_{i,\delta_t}(z_{i,t}),$ 
and hence
$\|x_{i,t}-z_{i,t}\|
\le
\frac{R_i}{\rho_i}\delta_t.$ 
Therefore, with
$C_{\mathcal X}
:=
\left(
\sum_{i=1}^N
\frac{R_i^2}{\rho_i^2}
\right)^{1/2}$, 
we have
\begin{equation}
\label{eq:x-z-asymptotic2}
\|x_t-z_t\|
\le
C_{\mathcal X}\delta_t
\longrightarrow0.
\end{equation}
Hence, we obtain using \eqref{eq:limits-r_t-y_t-x_t} and \eqref{eq:x-z-asymptotic2}, 
\begin{equation}
\label{eq:state-y-to-zt-below}
\|y_t-z_t\|
\le
\|y_t-x_t\|+\|x_t-z_t\| \le \varepsilon_t + C_{\mathcal{X}} \delta_t
\longrightarrow0\,,
\end{equation}
where the limit follows from using $\varepsilon_t \to 0, \delta_t \to 0$. 
Thus Assumptions~\ref{as:bias-and-stoch-error} and
\ref{as:state-y-to-zt} are satisfied.

Assumption~\ref{as:bounded-error} also holds as we have shown above that 
$\sup_t\|M_{t+1}\|<\infty$ almost surely. 
\end{proof}

\noindent\textbf{Lyapunov analysis of the projected limiting inclusion.} 
We now verify the Lyapunov condition required by
\citet[Theorem~4.2]{majewski-miasojedow-moulines18}.
Recall that $F(x)=-\mathcal D(x)$ where $\mathcal{D}$ is defined in \eqref{def:completed-potential-field}. 

The projected limiting field on $\mathcal X$ is
\[
F_{\mathcal X}(x)
:=
\Pi_{T_{\mathcal X}(x)}F(x).
\]
Since
$\mathcal X=\prod_{i=1}^N\mathcal X_i,$ 
we have
$T_{\mathcal X}(x)
=\prod_{i=1}^N T_{\mathcal X_i}(x_i)$, 
and Euclidean projection onto the tangent cone acts blockwise.

We next show that the ordinal potential $\Phi$ acts as a Lyapunov function for the projected differential inclusion induced by the \eqref{eq:cngp} stochastic approximation algorithm, i.e. we verify the second part in Theorem~\ref{thm:majewski-et-al4.2}, namely \eqref{eq:majewski-specialized-Lyapunov} and the identification of the zero set of the upper-bounding function $U$. 

\begin{tcolorbox}[colframe=white!, top=2pt,left=2pt,right=2pt,bottom=2pt] 
\begin{lemma}[Lyapunov inequality for the projected normalized field]
\label{lem:lyapunov-ineq}
For every $i\in[N]$ define
\begin{equation}
\label{eq:psi_i}
\Psi_i(x)
:=
\begin{cases}
\displaystyle
\|g_i(x)\| \cdot 
\left\|
\Pi_{T_{\mathcal X_i}(x_i)}
\left(
-\frac{g_i(x)}{\|g_i(x)\|}
\right)
\right\|^2,
&
\text{if} \,\,g_i(x)\neq0,
\\[3ex]
0,
&
\text{if} \,\, g_i(x)=0,
\end{cases}
\end{equation}
and let $U(x)
:=
-\sum_{i=1}^N\Psi_i(x).$ 
Then $U:\mathcal X\to\mathbb R$ is upper semicontinuous on $\mathcal X$, $U\le0$, and for
every $x\in\mathcal X$ and every $u\in F(x)$,
\begin{equation}
\label{eq:lyap-ineq}
\left\langle
\nabla\Phi(x),
\Pi_{T_{\mathcal X}(x)}u
\right\rangle
=
U(x).
\end{equation}
Moreover,
\begin{equation}
\label{eq:identification-zero-set}
\{x\in\mathcal X:U(x)=0\}
=
\left\{
x\in\mathcal X:
-\nabla\Phi(x)\in N_{\mathcal X}(x)
\right\}.
\end{equation}
\end{lemma}
\end{tcolorbox}

\begin{proof}

\noindent\textbf{Proof of \eqref{eq:lyap-ineq}.}
Let $x\in\mathcal X$ and $w\in F_{\mathcal X}(x).$ 
By definition of the projected field, there exists
$u\in F(x)=-\mathcal D(x)$ such that
$w=\Pi_{T_{\mathcal X}(x)}(u).$ 
Since both the action set and tangent cone are Cartesian products,
this means blockwise that
$w_i =\Pi_{T_{\mathcal X_i}(x_i)}(u_i),
i\in[N].$ 

Observe first that 
\begin{equation*}
\langle\nabla\Phi(x),w\rangle
= \sum_{i=1}^N \langle g_i(x), w_i\rangle\,.
\end{equation*}
To prove \eqref{eq:lyap-ineq}, we will show in what follows that for every $i \in [N]$, 
\begin{equation}
 \langle g_i(x), w_i\rangle = - \Psi_i(x)\,.
\end{equation}

Consider first a block $i \in [N]$ such that $g_i(x)\neq0$.
Then $F_i(x)$ is a singleton and necessarily
$u_i=-\frac{g_i(x)}{\|g_i(x)\|}.$ 

For a cone $C\subseteq\mathbb{R}^{d_i}$, let
$C^\circ := \{v\in\mathbb{R}^{d_i} :
\langle v,w\rangle \leq 0,\ \forall w\in C\}$ 
denote its polar cone. For a closed convex cone $C$, Moreau's decomposition gives
\[
u_i=\Pi_C(u_i)+\Pi_{C^\circ}(u_i),
\qquad
\langle \Pi_C(u_i),\Pi_{C^\circ}(u_i)\rangle=0.
\]
Taking
$C=T_{\mathcal X_i}(x_i), C^\circ=N_{\mathcal X_i}(x_i)$, 
and recalling that $w_i=\Pi_C(u_i)$, we obtain
\begin{equation*}
\langle u_i,w_i\rangle = \langle \Pi_C(u_i)+\Pi_{C^\circ}(u_i),\Pi_C(u_i)\rangle  
= \|\Pi_C(u_i)\|^2 = \|w_i\|^2\,.
\end{equation*}
Since $g_i(x)=-\|g_i(x)\|u_i$, 
it follows that
\begin{align}
\langle g_i(x),w_i\rangle
&=
-\|g_i(x)\|
\langle u_i,w_i\rangle
\nonumber\\
&=
-\|g_i(x)\|\cdot\|w_i\|^2
\nonumber\\
&=
-\|g_i(x)\| \cdot 
\left\|
\Pi_{T_{\mathcal X_i}(x_i)}
\left(
-\frac{g_i(x)}{\|g_i(x)\|}
\right)
\right\|^2
\nonumber\\
&=
-\Psi_i(x).
\label{eq:block-lie-dissipation}
\end{align}

If instead $g_i(x)=0$, the block $u_i$ may be any element of
$-\overline{\mathbb B}_{d_i}$, but independently of this choice,
\[
\langle g_i(x),w_i\rangle=0=-\Psi_i(x).
\]
Summing over the blocks gives
\[
\langle\nabla\Phi(x),w\rangle
=
-\sum_{i=1}^N\Psi_i(x)
=
U(x).
\]
The right-hand side is independent of the particular selection
$u\in F(x)$ at zero-gradient blocks. Consequently, 
$\left\langle
\nabla\Phi(x),
\Pi_{T_{\mathcal X}(x)}u
\right\rangle
=
U(x).$ Since $\Psi_i(x)\ge0$ for every $i$, we also have $U(x)\le0$.\\

\noindent\textbf{Proof of upper semicontinuity of $U$.}
Since $U=-\sum_{i=1}^N\Psi_i$, it suffices to show that each $\Psi_i$ is lower semicontinuous to conclude that their negative sum is upper semicontinuous, i.e., for every $x^n \to x$, 
\begin{equation}
\label{eq:psi-i-lower-semicont}
\Psi_i(x) \leq \liminf_n \Psi_i(x^n)\,.
\end{equation}

First, by definition, $\Psi_i(y) \geq 0$ for every $y \in \mathcal{X}$. If $g_i(x) = 0$, then $\Psi_i(x) = 0$ by definition (see \eqref{eq:psi_i}). Therefore for any such $x$ for which $g_i(x) = 0$, \eqref{eq:psi-i-lower-semicont} holds since $0 = \Psi_i(x) \leq \liminf_n \Psi_i(x^n)\,.$ 

It remains to show lower semicontinuity of $\Psi_i$ whenever $g_i(x)\neq0$. Positive homogeneity of the projection onto a cone yields: 
\begin{equation}
\label{eq:alternative-expression-psi-i}
\Psi_i(x)
= \|g_i(x)\| \cdot 
\left\|
\Pi_{T_{\mathcal X_i}(x_i)}
\left(
-\frac{g_i(x)}{\|g_i(x)\|}
\right)
\right\|^2
= \frac{\left\|
\Pi_{T_{\mathcal X_i}(x_i)}
\left(
-g_i(x)
\right)
\right\|^2}{\|g_i(x)\|}\,.
\end{equation}
Using Moreau's decomposition gives: 
\begin{equation}
\label{eq:dist-to-normal-cone}
\operatorname{dist}
\left(
-g_i(x),N_{\mathcal X_i}(x_i)
\right) = \left\| - g_i(x) - \Pi_{\mathcal{N}_{\mathcal X_i}(x_i)}(- g_i(x)) \right\|
= \left\|\Pi_{T_{\mathcal X_i}(x_i)}
\left(
-g_i(x)
\right)\right\|\,.
\end{equation}
Combining \eqref{eq:alternative-expression-psi-i} and \eqref{eq:dist-to-normal-cone} yields: 
\begin{equation}
\label{eq:Psi-normal-distance}
\Psi_i(x)
=
\frac{
\operatorname{dist}^2
\left(
-g_i(x),N_{\mathcal X_i}(x_i)
\right)
}{
\|g_i(x)\|
}.
\end{equation}
We now use this expression to show that $\Psi_i$ is lower semicontinuous in the case where $g_i(x) \neq 0$. Noting that $g_i$ is continuous, our main task is to show that the function $d_i: \mathcal{X} \to \mathbb{R}_+$ defined for any $x \in \mathcal{X}$ by:  
\[
d_i(x)
:=
\operatorname{dist}
\left(
-g_i(x),N_{\mathcal X_i}(x_i)
\right)
\]
is lower semicontinuous. Note here, that this is not a simple distance function to a fixed convex set which is known to be continuous. Here the set $\mathcal{N}_{\mathcal{X}_i}(x_i)$ also depends on the variable $x_i.$ We prove the lower semicontinuity claim below. 

Let $x^n\to x$ and, after passing to
a subsequence realizing $\liminf_n d_i(x^n)$, let 
\[
v_i^n
:=
\Pi_{N_{\mathcal X_i}(x_i^n)}
\bigl(-g_i(x^n)\bigr).
\]
Since $0\in N_{\mathcal X_i}(x_i^n)$,
$\|v_i^n\|\le\|g_i(x^n)\|$ by nonexpansiveness of the projection operator. As a consequence, since $(g(x^n))$ converges by continuity of $g$, it follows that it is bounded and hence that $(v_i^n)$ is also bounded. 
We can then pass to a further subsequence and let $v_i^n\to v_i$. 
The normal-cone mapping of a closed convex set has closed graph, and since $v_i^n \in N_{\mathcal X_i}(x_i^n)$, we obtain  
$v_i\in N_{\mathcal X_i}(x_i)$. Therefore,
\[
d_i(x)
= \inf_{y_i \in N_{\mathcal X_i}(x_i)} \|- g_i(x) - y_i\|
\le
\|-g_i(x)-v_i\|
= \lim_{n \to \infty} \|-g_i(x^n) - v_i^n\| 
=
\liminf_{n\to\infty} d_i(x^n)\,,
\]
where the penultimate identity follows from continuity of $g_i$ and $x^n \to x, v_i^n \to v_i$, and the last identity follows from recalling that $x^n$ is a subsequence realizing $\liminf_n d_i(x^n)$. 
We have proved that $d_i$ is lower semicontinuous. 

We now conclude to show that $\Psi_i$ as expressed in \eqref{eq:Psi-normal-distance} is lower semicontinuous. 
We use the following elementary fact: 

\begin{lemma}
\label{lem:liminf}
If $a_n\geq 0$ and $b_n\to b>0$, then
$\liminf_{n\to\infty}\frac{a_n}{b_n}
=\frac{\liminf_{n\to\infty}a_n}{b}.$ 
\end{lemma}

\begin{proof}
If $\liminf_n a_n=+\infty$, the claim follows immediately from
$b_n\to b>0$. Hence assume that $\liminf_n a_n<+\infty$. 
We prove the two inequalities separately.

For the first inequality, let $(a_{n_k})$ be a subsequence such that
$a_{n_k}\to \liminf_{n\to\infty}a_n.$ 
Since $b_{n_k}\to b>0$, we have
$\frac{a_{n_k}}{b_{n_k}}
\longrightarrow
\frac{\liminf_{n\to\infty}a_n}{b}.$ 
Since the liminf of a sequence is no larger than any of its
subsequential limits,
\[
\liminf_{n\to\infty}\frac{a_n}{b_n}
\le
\frac{\liminf_{n\to\infty}a_n}{b}.
\]

For the reverse inequality, let
$\bigl(a_{n_k}/b_{n_k}\bigr)$ be a subsequence such that
$\frac{a_{n_k}}{b_{n_k}}
\longrightarrow
\liminf_{n\to\infty}\frac{a_n}{b_n}.$
Since $b_{n_k}\to b>0$,
\[
a_{n_k}
=
b_{n_k}\frac{a_{n_k}}{b_{n_k}}
\longrightarrow
b\liminf_{n\to\infty}\frac{a_n}{b_n}.
\]
Again using that the liminf is no larger than any subsequential limit,
\[
\liminf_{n\to\infty}a_n
\le
b\liminf_{n\to\infty}\frac{a_n}{b_n}.
\]
Since $b>0$,
$\frac{\liminf_{n\to\infty}a_n}{b}
\le
\liminf_{n\to\infty}\frac{a_n}{b_n}.$ 
Combining the two inequalities proves the claim.
\end{proof}

Take now $a_n = d_i(x^n)^2, b_n = \|g_i(x^n)\|$. Since $g_i(x) \neq 0$ and $g_i$ is continuous, it follows that $b_n \to \|g_i(x)\| > 0$. By lower semicontinuity of $d_i$, $d_i(x)^2 \leq (\liminf_{n \to \infty} d_i(x^n))^2 = \liminf_{n \to \infty} d_i(x^n)^2\,.$ Therefore, we can conclude using Lemma~\ref{lem:liminf}: 
\begin{equation}
\Psi_i(x) = \frac{d_i(x)^2}{\|g_i(x)\|} \leq \frac{\liminf_{n \to \infty} d_i(x^n)^2}{\|g_i(x)\|} 
= \liminf_{n \to \infty} \frac{d_i(x^n)^2}{\|g_i(x^n)\|} =  \liminf_{n \to \infty} \Psi_i(x^n)\,,
\end{equation}
which shows that $\Psi_i$ is lower semicontinuous.
Consequently, $U=-\sum_{i=1}^N\Psi_i$ is upper semicontinuous.\\ 

\noindent\textbf{Proof of \eqref{eq:identification-zero-set}: zeros of the Lyapunov upper bound.} Since each $\Psi_i(x)\ge0$,
$U(x)=0$ if and only if $\Psi_i(x)=0$ for all $i \in [N].$

If $g_i(x)=0$, then $-g_i(x)=0\in N_{\mathcal X_i}(x_i)$ 
and $\Psi_i(x)=0$ by definition.

Suppose now that $g_i(x)\neq0$. Then
$\Psi_i(x)=0$ 
if and only if
$
\Pi_{T_{\mathcal X_i}(x_i)}
\left(
-\frac{g_i(x)}{\|g_i(x)\|}
\right)
=0.$

For a closed convex cone $C$,
$\Pi_C(v)=0$ if and only if $v\in C^\circ$.
Since $T_{\mathcal X_i}(x_i)^\circ
= N_{\mathcal X_i}(x_i)$, 
we obtain
$\Psi_i(x)=0$ if and only if $-\frac{g_i(x)}{\|g_i(x)\|}
\in N_{\mathcal X_i}(x_i).$ 
Since the normal cone is a cone, this is equivalent to
$-g_i(x)\in N_{\mathcal X_i}(x_i).$ 
This proves \eqref{eq:identification-zero-set}. 
The product-space characterization follows from noting that $N_{\mathcal X}(x)
=\prod_{i=1}^N N_{\mathcal X_i}(x_i).$ 
\end{proof}

For compatibility with Theorem~\ref{thm:majewski-et-al4.2}, we
extend $U$ from $\mathcal X$ to the ambient open set $\mathcal U$.
Define
\begin{equation}
\label{eq:extension-U}
\widetilde U(x)
:=
U\bigl(\Pi_{\mathcal X}(x)\bigr),
\qquad x\in\mathcal U.
\end{equation}
Since $\mathcal X$ is nonempty, closed, and convex,
$\Pi_{\mathcal X}$ is continuous. Hence $\widetilde U$ is upper
semicontinuous on $\mathcal U$ as the composition of the upper
semicontinuous function $U$ with the continuous map
$\Pi_{\mathcal X}$. Moreover,
$\widetilde U(x)=U(x), 
\forall x\in\mathcal X,$ 
and $\widetilde U\le0$ on $\mathcal U$.
With a slight abuse of notation, we henceforth denote this extension
again by $U$.

\begin{tcolorbox}[colframe=white!, top=2pt,left=2pt,right=2pt,bottom=2pt] 
\begin{lemma}[Potential stationarity is equivalent to Nash equilibrium]
\label{lem:potential-stationarity-ne}
Under Assumption~\ref{ass:block-pseudoconvex} and
Assumption~\ref{as:non-degeneracy-ord-pot},
$\mathcal S = \operatorname{NE}(\mathcal G),$ 
where $\mathcal S=\{x\in\mathcal X:U(x)=0\}$.
\end{lemma}
\end{tcolorbox}

\begin{proof}
By \eqref{eq:identification-zero-set} in Lemma~\ref{lem:lyapunov-ineq} ,
$x\in\mathcal S$ if and only if 
$-\nabla_i\Phi(x) \in N_{\mathcal X_i}(x_i), 
\forall i\in[N].$

We first prove that $\mathcal S\subseteq\operatorname{NE}(\mathcal G)$. Let $x\in\mathcal S$, $i \in [N]$. 

If $\nabla_i\Phi(x)=0,$
Assumption~\ref{as:non-degeneracy-ord-pot} gives
$\nabla_iJ_i(x)=0.$ 
Suppose for contradiction that $x_i$ is not a best response to
$x_{-i}$. Then there exists $y_i\in\mathcal X_i$ such that
$J_i(y_i,x_{-i}) < J_i(x_i,x_{-i}).$ 
By blockwise pseudoconvexity,
$\left\langle
\nabla_iJ_i(x),
y_i-x_i
\right\rangle
<0,$ contradicting $\nabla_iJ_i(x)=0$.
Thus $x_i$ is a best response.

Suppose now that
$\nabla_i\Phi(x)\neq0.$ 
By Lemma~\ref{lem:ordinal-potential-gradient-alignment},
$\nabla_iJ_i(x) =
\lambda_i(x)\nabla_i\Phi(x)$ 
for some $\lambda_i(x)>0$.
Since
$-\nabla_i\Phi(x) \in N_{\mathcal X_i}(x_i)$ 
and the normal cone is a cone,
\[
-\nabla_iJ_i(x)
=
-\lambda_i(x)\nabla_i\Phi(x)
\in
N_{\mathcal X_i}(x_i).
\]
Equivalently,
$\left\langle
\nabla_iJ_i(x),
y_i-x_i
\right\rangle
\ge0, 
\forall y_i\in\mathcal X_i.$ 
If there were $y_i\in\mathcal X_i$ satisfying
\[
J_i(y_i,x_{-i})
<
J_i(x_i,x_{-i}),
\]
blockwise pseudoconvexity would imply
$\left\langle
\nabla_iJ_i(x),
y_i-x_i
\right\rangle
<0,$ 
a contradiction. Hence $x_i$ is a best response.
Since this holds for every $i \in [N]$,
$x\in\operatorname{NE}(\mathcal G).$

Conversely, suppose
$x\in\operatorname{NE}(\mathcal G).$ 
For every $i \in [N]$, $x_i$ minimizes the differentiable
function
$y_i\mapsto J_i(y_i,x_{-i})$ 
over the convex set $\mathcal X_i$. The first-order necessary
optimality condition therefore gives
$\left\langle
\nabla_iJ_i(x),
y_i-x_i
\right\rangle
\ge0, 
\forall y_i\in\mathcal X_i,$
or equivalently, 
$-\nabla_iJ_i(x)
\in
N_{\mathcal X_i}(x_i).$ 

If $\nabla_iJ_i(x)=0$, nondegeneracy gives
$\nabla_i\Phi(x)=0$, and consequently
$-\nabla_i\Phi(x)=0 \in N_{\mathcal X_i}(x_i).$ 

If $\nabla_iJ_i(x)\neq0$, Lemma
~\ref{lem:ordinal-potential-gradient-alignment} gives
$\nabla_iJ_i(x) = \lambda_i(x)\nabla_i\Phi(x),
\lambda_i(x)>0.$ 
Hence
\[
-\nabla_i\Phi(x)
=
\frac1{\lambda_i(x)}
\bigl(-\nabla_iJ_i(x)\bigr)
\in
N_{\mathcal X_i}(x_i).
\]
Thus the potential stationarity condition holds for every player,
so
$x\in\mathcal S.$
This proves 
$\mathcal S
=
\operatorname{NE}(\mathcal G).$
\end{proof}

\noindent\textbf{End of the proof of Theorem~\ref{thm:as-comparison-ne}.} 
Applying Theorem~\ref{thm:majewski-et-al4.2} with
$\mathcal K=\mathcal X,
F=-\mathcal D,
V=\Phi$, 
and with the upper semicontinuous extension of $U$ defined in
\eqref{eq:extension-U}, the Lyapunov condition
\eqref{eq:majewski-specialized-Lyapunov} follows from
Lemma~\ref{lem:lyapunov-ineq}, since this extension agrees with the
original $U$ on $\mathcal X$.
Using $\mathcal S=\operatorname{NE}(\mathcal G)$ (see Lemma~\ref{lem:potential-stationarity-ne}) 
and $\operatorname{int}_{\mathbb R}\Phi(\mathcal S)
=\operatorname{int}_{\mathbb R}\Phi(\NE(\mathcal G))=\varnothing$ by assumption, 
we obtain, almost surely,
\begin{equation}
\label{eq:as-cv-NE}
\operatorname{dist}
\bigl(
z_t,\operatorname{NE}(\mathcal G)
\bigr)
\longrightarrow0\,, \quad 
\Phi(z_t)\longrightarrow\Phi_\infty\,, 
\end{equation}
for some finite random variable $\Phi_\infty$.

We now transfer this conclusion to the sequence $(x_t)$. By construction,
$x_{i,t}=
\mathsf Q_{i,\delta_t}(z_{i,t}),$ 
and hence
$\|x_{i,t}-z_{i,t}\|
\le
\frac{R_i}{\rho_i}\delta_t.$ 
Therefore, with
$C_{\mathcal X}
:=
\left(
\sum_{i=1}^N
\frac{R_i^2}{\rho_i^2}
\right)^{1/2}$, 
we have
\begin{equation}
\label{eq:x-z-asymptotic}
\|x_t-z_t\|
\le
C_{\mathcal X}\delta_t
\longrightarrow0.
\end{equation}
Combining
\eqref{eq:as-cv-NE} and
\eqref{eq:x-z-asymptotic},
\begin{align*}
\operatorname{dist}
\bigl(
x_t,
\operatorname{NE}(\mathcal G)
\bigr)
&\le
\operatorname{dist}
\bigl(
z_t,
\operatorname{NE}(\mathcal G)
\bigr)
+
\|x_t-z_t\| \underset{t \to \infty}{\longrightarrow} 0\,. 
\end{align*}
Moreover, uniform continuity of $\Phi$ on the compact set
$\mathcal X$ yields
$|\Phi(x_t)-\Phi(z_t)|\longrightarrow0,$ 
and hence 
$\Phi(x_t)\longrightarrow\Phi_\infty.$

Suppose
$\operatorname{NE}(\mathcal G) =\{x_\star^1,\ldots,x_\star^m\}$ 
is finite. The rest of the proof is standard. For $m\ge2$ let
$\Delta
:=
\min_{j\neq k}\|x_\star^j-x_\star^k\|>0.$ 
Since
$\operatorname{dist}(z_t,\operatorname{NE}(\mathcal G))\to0,$
eventually
\[
z_t\in
\bigcup_{j=1}^m
B(x_\star^j,\Delta/4).
\]
On the other hand, nonexpansiveness of projection and the uniform bound
$\|\widehat D_t\|\le G$ give
\[
\|z_{t+1}-z_t\|
\le
G\eta_t
\longrightarrow0.
\]
Hence, eventually,
$\|z_{t+1}-z_t\|<\Delta/2.$ 
Since distinct balls
$B(x_\star^j,\Delta/4)$ are separated by distance at least
$\Delta/2$, while eventually
$\|z_{t+1}-z_t\|<\Delta/2$, the sequence cannot move between
two distinct such balls after some finite time, and remains in only one such
ball.  Together with
$\operatorname{dist}(z_t,\operatorname{NE}(\mathcal G))\to0$, this
implies
\[
z_t\to X_\infty
\]
for some random variable $X_\infty\in\operatorname{NE}(\mathcal G)$. Finally,
$\|x_t-z_t\|\to0$ gives $x_t\to X_\infty.$ 

If $\operatorname{NE}(\mathcal G)$ is a singleton, the same
conclusion follows immediately from
$\operatorname{dist}
\bigl(
x_t,\operatorname{NE}(\mathcal G)
\bigr)
\longrightarrow0.$ 
This completes the proof.

\subsubsection{Almost sure convergence to Nash in ordinal potential games in the unconstrained case}

We slightly strengthen Assumption~\ref{as:rm-stepsizes} with the condition $\sum_{t=1}^\infty\eta_t\delta_t<\infty$ for our alternative simpler proof based on the Robbins-Monro theorem in the unconstrained case ($\mathcal{X}_i = \mathbb{R}^{d_i}$). The stepsize and radii choices made in the main part for Theorem~\ref{thm:as-comparison-ne}, $\eta_t=\frac{\eta}{t^p}, 
\delta_t=\frac{\delta}{t^q},$
with $1/2<p\leq1$, $q>1-p$ and $\delta \in (0,\rho)$, remains valid and satisfies the following assumption. 

\begin{tcolorbox}[colframe=white!, top=2pt,left=2pt,right=2pt,bottom=2pt] 
\begin{assumption}[Robbins-Monro stepsizes and vanishing resolution]
\label{as:rm-stepsizes2} The stepsizes $(\eta_t)$ and the resolution $(\delta_t)$ satisfy: 
$\sum_{t=1}^\infty\eta_t=\infty, \sum_{t=1}^\infty\eta_t^2<\infty$ 
and  $\sum_{t=1}^\infty\eta_t\delta_t<\infty\,,$ 
with $\eta_t>0, \delta_t \in (0,\rho), \forall t \in \mathbb{N}$ and $\delta_t\to0.$
\end{assumption}
\end{tcolorbox}

In the unconstrained setting $X_i=\mathbb R^{d_i}$, set \(x_t=z_t\), query \(x_{i,t}\pm\delta_tU_{i,t}\), and update \(x_{t+1}=x_t-\eta_t\widehat D_t\) instead of \eqref{eq:cngp}. 

\begin{tcolorbox}[colframe=white!, top=2pt,left=2pt,right=2pt,bottom=2pt] 
\begin{theorem}[Almost-sure convergence to Nash equilibrium]
\label{thm:as-comparison-ne-apdx}
Let $\mathcal G$ be an ordinal potential game with potential $\Phi$. Let Assumptions~\ref{ass:block-pseudoconvex} and~\ref{as:non-degeneracy-ord-pot} hold. 
Suppose that $\Phi$ is $L$-smooth and coercive\footnote{i.e., $\Phi(x)\to+\infty$ as $\|x\|\to\infty\,.$}. 
Let $(x_t)$ be generated by \eqref{eq:cngp} with the projection inactive, where $\widehat D_{i,t}$ is defined by
\eqref{eq:cngp} under the gap-independent noisy
comparison model with $\nu_i\in(0,1/2]$.
Suppose that the stepsizes and comparison radii satisfy
Assumption~\ref{as:rm-stepsizes2}.
Then, with probability one,
\[
\Phi(x_t)\ \text{converges},\qquad 
\|\nabla\Phi(x_t)\|_{2,1}\to0\,, \quad \operatorname{dist}(x_t,\operatorname{NE}(\mathcal G))\to0\,.
\]
If $\operatorname{NE}(\mathcal G)$ is finite, then there exists an
$\operatorname{NE}(\mathcal G)$-valued random variable $X_\infty$ such that
$x_t\to X_\infty$ almost surely. In particular, if $\operatorname{NE}(\mathcal G)$ is a singleton, the iterates converge to it almost surely.
\end{theorem}
\end{tcolorbox}

A convenient admissible choice is
$\eta_t=\frac{\eta}{t^p}, 
\delta_t=\frac{\delta}{t^q},$
with $1/2<p\leq1$ and $q>1-p$. 

The next estimate is useful because it controls the quantity that enters
potential descent directly. In contrast, a uniform bound on
$\|D_i^\delta-D_i\|$ necessarily becomes ill-conditioned as the gradient
vanishes.

\begin{tcolorbox}[colframe=white!, top=2pt,left=2pt,right=2pt,bottom=2pt] 
\begin{lemma}[Finite-resolution potential alignment]
\label{lem:uniform-potential-alignment-as}
Under the assumptions of Theorem~\ref{thm:as-comparison-ne}, for every player $i \in [N]$,
every $\delta>0$ and every $x \in \mathcal{X}^{\delta}$,
\begin{equation}
\label{eq:first-ineq-to-prove}
\left\langle
\nabla_i\Phi(x),
D_i^\delta(x)
\right\rangle
\ge
\|\nabla_i\Phi(x)\|
-
\frac{L\delta}{\kappa_{d_i}}.
\end{equation}
If in addition $\nabla_i\Phi(x)\neq0$, then $D_i(x)
=
\frac{\nabla_i\Phi(x)}
{\|\nabla_i\Phi(x)\|}$ and 
\begin{equation}
\label{eq:second-ineq-to-prove}
\left\langle
\nabla_i\Phi(x),
D_i^\delta(x)-D_i(x)
\right\rangle
\ge
-\frac{L\delta}{\kappa_{d_i}}.
\end{equation}
\end{lemma}
\end{tcolorbox}

\begin{proof}
Let $x \in \mathcal{X}, i \in [N].$ Recalling the definition of the finite resolution vector field~$D_i^{\delta}$, we have 
\begin{equation}
\label{eq:inner-prod-rewriting}
\left\langle
\nabla_i\Phi(x),
D_i^\delta(x)
\right\rangle 
= \frac{1}{\kappa_{d_i}} \langle \nabla_i \Phi(x), \mathbb{E}_{U_i}[S_i^{\delta}(x,U_i) U_i] \rangle 
= \frac{1}{\kappa_{d_i}} 
\mathbb{E}_{U_i}[S_i^{\delta}(x,U_i) \langle \nabla_i \Phi(x), U_i \rangle]\,. 
\end{equation}
We will prove first that for any $u \in \mathbb{R}^{d_i}$ such that $\|u\|_2 = 1$, 
\begin{equation}
\label{eq:to-prove-bias-control}
S_i^{\delta}(x,u) \langle \nabla_i \Phi(x), u \rangle \geq |\langle \nabla_i \Phi(x), u \rangle| - L \delta\,.
\end{equation}
Define for any $\delta > 0$ and any $u$ s.t. $x_i \pm \delta u \in \mathcal{X}_i$ and $\|u\|_2 = 1$ the central difference quotient: 
\begin{equation}
\Delta_{\Phi, i}^\delta(u) := 
\frac{
\Phi(x_i+\delta u,x_{-i})
-
\Phi(x_i-\delta u,x_{-i})
}{2\delta}\,.
\end{equation}
By $L$-smoothness of $\Phi$, we have: 
\begin{align}
|\Phi(x_i+\delta u,x_{-i}) - 
(\Phi(x) + \delta \langle \nabla_i \Phi(x), u \rangle)| &\leq \frac{L \delta^2}{2}\,,\\
|\Phi(x_i -\delta u,x_{-i}) - 
(\Phi(x) - \delta \langle \nabla_i \Phi(x), u \rangle)| &\leq \frac{L \delta^2}{2}\,.
\end{align}
Using both the above inequalities, it follows that: 
\begin{align} 
\label{eq:diff-delta-inner-prod} 
&|\Delta_{\Phi,i}^{\delta}(u) - \langle \nabla_i \Phi(x), u \rangle| \nonumber\\
&= \frac{1}{2 \delta}\left| \Phi(x_i+\delta u,x_{-i})-\Phi(x_i-\delta u,x_{-i}) - 2\delta \langle \nabla_i \Phi(x), u \rangle \right| \nonumber\\
&= \frac{1}{2 \delta} \left|\Phi(x_i+\delta u,x_{-i}) - (\Phi(x) + \delta \langle \nabla_i \Phi(x), u \rangle ) - (\Phi(x_i -\delta u,x_{-i}) - 
(\Phi(x) - \delta \langle \nabla_i \Phi(x), u \rangle )) \right|\nonumber\\
&\leq \frac{1}{2 \delta} \left(\frac{L \delta^2}{2} +  \frac{L \delta^2}{2} \right) = \frac{\delta L}{2}\,.
\end{align}
By the ordinal-potential property,
$S_i^\delta(x,u) = \operatorname{sgn}(\Delta_{\Phi, i}^{\delta}(u)).$
If $\operatorname{sgn}(\Delta_{\Phi,i}^{\delta}(u)) =\operatorname{sgn}(\langle \nabla_i \Phi(x), u \rangle)$,
then
$S_i^\delta(x,u)\langle \nabla_i \Phi(x), u \rangle=|\langle \nabla_i \Phi(x), u \rangle|.$
Otherwise, the signs are opposite and \eqref{eq:diff-delta-inner-prod} implies that 
$|\langle \nabla_i \Phi(x), u \rangle|
\le
\frac{\delta L}{2},$
and hence
\begin{equation}
S_i^\delta(x,u) \langle \nabla_i \Phi(x), u \rangle
=
-|\langle \nabla_i \Phi(x), u \rangle|
\ge
|\langle \nabla_i \Phi(x), u \rangle|-L\delta.
\end{equation}
Therefore, the above inequality holds for any $u$, in both cases and even if one of the two quantities is zero, which proves \eqref{eq:to-prove-bias-control}. 
Using now \eqref{eq:inner-prod-rewriting} and taking expectation in \eqref{eq:to-prove-bias-control} yields: 
\begin{equation}
\left\langle
\nabla_i\Phi(x),
D_i^\delta(x)
\right\rangle  
= \frac{1}{\kappa_{d_i}} 
\mathbb{E}_{U_i}[S_i^{\delta}(x,U_i) \langle \nabla_i \Phi(x), U_i \rangle] 
\geq \frac1{\kappa_{d_i}}
\mathbb E
[|\langle \nabla_i \Phi(x) ,U_i\rangle|]
-
\frac{L\delta}{\kappa_{d_i}}\,. 
\end{equation}
Rotational invariance of the uniformly distributed random variable $U_i$ on the sphere gives
$\mathbb E[|\langle \nabla_i \Phi(x),U_i\rangle|]
= \kappa_{d_i} \cdot\|\nabla_i \Phi(x)\|$ (see Lemma 3.1 in~\citet{fu-wang-yan23minimax-spherical})
and this concludes the proof of \eqref{eq:first-ineq-to-prove}. Equation \eqref{eq:second-ineq-to-prove} follows from recalling that $D_i(x)=
\frac{\nabla_i\Phi(x)}
{\|\nabla_i\Phi(x)\|}$ when $\nabla_i\Phi(x)\neq0$ by Lemma~\ref{lem:ordinal-potential-gradient-alignment}, which implies that $\langle \nabla_i \Phi(x),D_i(x)\rangle=\|\nabla_i \Phi(x)\|\,.$
\end{proof}

We are now ready to prove Theorem~\ref{thm:as-comparison-ne-apdx}. 
\begin{proof}
We divide the proof into four steps.\\

\noindent\textbf{Step 1: Conditional Lyapunov drift.}
Since $\nabla\Phi$ is $L$-Lipschitz, we have 
\[
\Phi(x_{t+1})
\le
\Phi(x_t)
+
\left\langle
\nabla\Phi(x_t),x_{t+1} - x_t
\right\rangle
+
\frac{L}{2}
\|x_{t+1}-x_t\|^2.
\]
Using the update rule of the algorithm, it follows that: 
\begin{equation}
\Phi(x_{t+1})
\le
\Phi(x_{t})
-
\eta_t
\left\langle
\nabla\Phi(x_{t}),\widehat D_t
\right\rangle
+
\frac{L\eta_t^2}{2}
\|\widehat D_t\|^2.
\end{equation}
Taking conditional expectation with respect to $\mathcal F_t$ and using Lemma~\ref{lem:one-bit-estimator} and
the definition of $D_i^{\delta}$,
\begin{equation}
\label{eq:master}
\mathbb E[
\Phi(x_{t+1})\mid\mathcal F_t
]
\leq 
\Phi(x_{t})
-
\eta_t
\sum_{i=1}^N
\left\langle
\nabla_i\Phi(x_t),
D_i^{\delta_t}(x_t)
\right\rangle 
+ \frac{LC_{2,\nu}}{2}\eta_t^2\,,
\end{equation}
where $C_{2,\nu} = \sum_{i=1}^N\frac1{(2\nu_i\kappa_{d_i})^2}.$

Lemma~\ref{lem:uniform-potential-alignment-as} yields
\begin{equation}
\label{eq:bias-control}
\sum_{i=1}^N
\left\langle
\nabla_i\Phi(x_t),
D_i^{\delta_t}(x_t)
\right\rangle
\ge
\|\nabla\Phi(x_t)\|_{2,1}
-
LC_1\delta_t\,, 
\end{equation}
where $C_1 = \sum_{i=1}^N\frac1{\kappa_{d_i}}.$
Hence, combining \eqref{eq:master} with \eqref{eq:bias-control}, we obtain: 
\begin{equation}
\label{eq:master-lyap-drift}
\mathbb E[
\Phi(x_{t+1})\mid\mathcal F_t
]
\leq 
\Phi(x_{t})
-
\eta_t
\|\nabla\Phi(x_{t})\|_{2,1}
+
LC_1\eta_t\delta_t
+
\frac{LC_{2,\nu}}{2}\eta_t^2.
\end{equation}

\noindent\textbf{Step 2: Application of Robbins-Siegmund's theorem.}
Since $\Phi$ is continuous and coercive, it is bounded below by $\Phi_\star=\inf_{x\in\mathcal X}\Phi(x)>-\infty.$
Set
\begin{equation}
V_t :=\Phi(x_t)-\Phi_\star\ge0\,, \quad a_t :=LC_1\eta_t\delta_t
+\frac{LC_{2,\nu}}{2}\eta_t^2\,,\quad  \forall t \geq 0\,. 
\end{equation}
We can then rewrite \eqref{eq:master-lyap-drift} as follows: 
\begin{equation}
\mathbb E[V_{t+1}\mid\mathcal F_t] \le V_t -
\eta_t \|\nabla\Phi(x_t)\|_{2,1} + a_t.
\end{equation}

By Assumption~\ref{as:rm-stepsizes2}, $\sum_{t=1}^\infty a_t<\infty$
and we can apply the Robbins-Siegmund almost supermartingale theorem \citep{robbins-siegmund71} to obtain that
$V_t$ converges almost surely to a finite random variable and that
\begin{equation}
\label{eq:summability-eta-grad}
\sum_{t=1}^\infty
\eta_t
\|\nabla\Phi(x_t)\|_{2,1}
<
\infty
\qquad\text{a.s.}
\end{equation}

Since $\Phi(x_t)$ converges (as $V_t$ does) and $\Phi$ is coercive, the sequence
of play $(x_t)$ is almost surely bounded.\\

\noindent\textbf{Step 3: First-order stationarity guarantee.}
Define for any $x \in \mathcal{X}, t \geq 1$, 
\[
r(x)
:=
\|\nabla\Phi(x)\|_{2,1}
=
\sum_{i=1}^N
\|\nabla_i\Phi(x)\|,
\qquad
r_t :=r(x_t).
\]
Since 
$\sum_t\eta_t=\infty$ by Assumption~\ref{as:rm-stepsizes2} and $\sum_{t=1}^\infty \eta_t\|\nabla\Phi(x_t)\|_{2,1}<\infty$ as shown in \eqref{eq:summability-eta-grad}, it follows that 
$\liminf_{t\to\infty}r_t=0.$

We show that in fact $r_t\to0$. The proof is similar to the proof of $\lim_{t\to\infty}\|\nabla_i\Phi(x(t))\|=0$ in continuous time in the proof of Theorem~\ref{thm:cv-BNPG-ordinal-pot}. However, here, we deal with the discrete time setting and we provide a complete proof using again standard arguments. 

First, by the reverse triangle
inequality and the Lipschitz continuity of $\nabla\Phi$, for any $x, y \in \mathcal{X},$
\begin{equation}
\label{eq:lip}
|r(x)-r(y)|
\le \sum_{i=1}^N \|\nabla_i\Phi(x)-\nabla_i\Phi(y)\|
\le \sqrt N\,
\|\nabla\Phi(x)-\nabla\Phi(y)\|
\le \sqrt N L\|x-y\|.
\end{equation}
Furthermore, using again the update rule of the algorithm, we have: 
\begin{equation}
\label{eq:increments-xt-bound}
\|x_{t+1}-x_t\| = \eta_t \|\hat{D}_t\| = \eta_t \left\|\left(\frac{Y_{i,t}}{2 \nu_i \kappa_{d_i}} U_{i,t}\right)_{i \in [N]} \right\| \leq 
\sqrt{C_{2,\nu}}\,\eta_t.
\end{equation}
Combining the above inequality with \eqref{eq:lip} yields for any $t \geq 1$,  
\begin{equation}
\label{eq:stability}
|r_{t+1}-r_t|
\le
K\eta_t,
\qquad
K:=L\sqrt{NC_{2,\nu}}.
\end{equation}

Suppose, towards a contradiction, that
$\limsup_{n\to\infty}r_n>0.$
Together with $\liminf_{t\to\infty}r_t=0$, this implies that there exists
$\varepsilon>0$ and infinitely many disjoint upcrossings from below
$\varepsilon$ to above $2\varepsilon$. 
Consider one such sufficiently late upcrossing. Let $p<q$ be chosen so
that $r_p<\varepsilon, r_q\ge2\varepsilon,$
and $r_n\ge\varepsilon$ for $p<n\le q$.
Using \eqref{eq:stability}, we have 
\begin{equation*}
\varepsilon
< r_q-r_p
\le K\sum_{t=p}^{q-1}\eta_t.
\end{equation*}
Since $\eta_t\to0$, for all sufficiently large $p$, $K \eta_p\le\frac{\varepsilon}{2}$ which implies, using the above inequality, that 
\begin{equation}
\sum_{t=p+1}^{q-1}\eta_t
\ge
\frac{\varepsilon}{2K}\,,
\end{equation}
up to adjustments of the endpoints in the sum. On this interval,
$r_t\ge\varepsilon$, and hence
\[
\sum_{t=p+1}^{q-1}\eta_t r_t
\ge
\frac{\varepsilon^2}{2K}.
\]
Infinitely many disjoint upcrossings would therefore force
$\sum_t\eta_t r_t=\infty,$ contradicting \eqref{eq:summability-eta-grad}. Thus
$r_t\to0$ almost surely.\\ 

\noindent\textbf{Step 4: Identification of the limit set with Nash equilibrium.} The proof follows similar lines to the proof of Lemma~\ref{lem:potential-stationarity-ne} and the end of the proof of Theorem~\ref{thm:as-comparison-ne}. 

Since the sequence $(\Phi(x_t))$ converges and $\Phi$ is coercive, $(x_t)$ is bounded. If $\operatorname{dist}\bigl(x_t,\operatorname{NE}(\mathcal G)
\bigr) \not\to 0$, there exists a converging subsequence $x_{t_k} \to \bar{x}$ remaining a positive distance from $\operatorname{NE}(\mathcal G)$. Since $\|\nabla\Phi(x_{t_k})\|_{2,1} \to 0$, $\nabla_i \Phi(\bar{x}) = 0$ for every $i \in [N]$ by continuity of $\nabla_i \Phi$. Assumption~\ref{as:non-degeneracy-ord-pot} then gives $\nabla_i J_i(\bar{x}) = 0$ for all $i \in [N]$, and Assumption~\ref{ass:block-pseudoconvex} implies that $\bar{x} \in \NE(\mathcal{G}),$ a contradiction. Hence  $\operatorname{dist}\bigl(x_t,\operatorname{NE}(\mathcal G)
\bigr) \rightarrow 0$. 

Moreover, $\|x_{t+1} - x_t\| \leq 
\sqrt{C_{2,\nu}}\,\eta_t$ by \eqref{eq:increments-xt-bound}. 
If $\NE(\mathcal{G})$ is finite, the same separation argument used at the end of the proof of Theorem~\ref{thm:as-comparison-ne} prevents the iterates from switching infinitely often between neighborhoods of distinct equilibria. Hence $x_t \rightarrow X_{\infty} \in \NE(\mathcal{G})$ almost surely. 
\end{proof}

\subsection{Finite-time first-order stationarity}
\label{apdx:finite-time-fos}

We conclude the unconstrained analysis with a finite-time first-order stationarity guarantee to complement the asymptotic convergence result of Theorem~\ref{thm:as-comparison-ne-apdx} above. 

\begin{tcolorbox}[colframe=white!, top=2pt,left=2pt,right=2pt,bottom=2pt] 
\begin{theorem}[Finite-time stationarity]
\label{thm:comparison-potential-descent}
Suppose $\Phi$ is $L$-smooth and bounded below by $\Phi_{\inf}$. Then for any time horizon $T > 0$, 
with constant stepsize $\eta_t$ and comparison radius $\delta_t$ of order $T^{-1/2}$, 
a uniformly sampled iterate $x_{\tau}$ from $\{x_1, \ldots, x_T\}$ satisfies
$\mathbb E\|\nabla\Phi(x_\tau)\|_{2,1}
=\mathcal O(T^{-1/2}).$
\end{theorem}
\end{tcolorbox}
We provide a more detailed statement of this result in Theorem~\ref{thm:comparison-potential-descent-apdx} with the dependence on all constants of the problem (potential gap $\Phi(x_1)-\Phi_{\inf}$, dimension $d_i$, number $N$ of players).

We now give a more detailed version of Theorem~\ref{thm:comparison-potential-descent} with explicit dependence on the problem constants.  

Let
$C_1 := \sum_{i=1}^N \frac1{\kappa_{d_i}} = \Theta\left(\sum_{i=1}^N\sqrt{d_i}\right),
C_{2,\nu} := \sum_{i=1}^N \frac1{(2 \nu_i\kappa_{d_i})^2} =
\Theta\left(\sum_{i=1}^N \frac{d_i}{\nu_i^2} \right).$

\begin{tcolorbox}[colframe=white!, top=2pt,left=2pt,right=2pt,bottom=2pt] 
\begin{theorem}[Finite-time stationarity]
\label{thm:comparison-potential-descent-apdx}
Suppose $\Phi$ is $L$-smooth and bounded below by $\Phi_{\inf}$ along
the iterates of the algorithm. Then for any time horizon $T > 0$, 
\begin{equation}
\mathbb E
\|\nabla\Phi(x_\tau)\|_{2,1}
\le
\frac{
\Delta_\Phi
+
LC_1\sum_{t=1}^T\eta_t\delta_t
+
\frac{LC_{2,\nu}}{2}\sum_{t=1}^T\eta_t^2
}{
\sum_{t=1}^T\eta_t
}\,,
\end{equation}
where $\Delta_\Phi =\Phi(x_1)-\Phi_{\inf}.$
and $\tau$ is a random variable sampled according to
$\mathbb P(\tau=t) =
\frac{\eta_t}{\sum_{s=1}^T\eta_s}\,.$
Setting $
\eta_t=\frac{a}{\sqrt T},
\delta_t=\frac{b}{\sqrt T}$
for fixed constants $a,b>0$ yields 
$\mathbb E\|\nabla\Phi(x_\tau)\|_{2,1} = \mathcal{O}(T^{-1/2})$
where $\tau$ is uniform on
$\{1,\ldots,T\}$. 
\end{theorem}
\end{tcolorbox}

\begin{proof}
Recall from \eqref{eq:master-lyap-drift} that 
\begin{equation}
\mathbb E[
\Phi(x_{t+1})\mid\mathcal F_t
]
\leq 
\Phi(x_{t})
-
\eta_t
\|\nabla\Phi(x_{t})\|_{2,1}
+
LC_1\eta_t\delta_t
+
\frac{LC_{2,\nu}}{2}\eta_t^2.
\end{equation}
Taking total expectation in the above inequality, summing over $t = 1, \dots, T$, using the lower bound
$\Phi_{\inf}$, rearranging and dividing by $\sum_{t=1}^T\eta_t$ yields: 
\begin{equation}
\frac{\sum_{t=1}^T
\eta_t
\mathbb E
\left[
\|\nabla\Phi(x_t)\|_{2,1}
\right]}{\sum_{t=1}^T \eta_t}
\le
\frac{
\Delta_\Phi
+
LC_1
\sum_{t=1}^T
\eta_t\delta_t
+
\frac{LC_{2,\nu}}{2}
\sum_{t=1}^T
\eta_t^2}{\sum_{t=1}^T \eta_t},
\end{equation}
where
$\Delta_\Phi =\Phi(x_1)-\Phi_{\inf}.$ 
The above can be rewritten equivalently using 
a random variable $\tau$ sampled according to
$\mathbb P(\tau=t) =
\frac{\eta_t}{\sum_{s=1}^T\eta_s},$ as follows: 
\begin{equation}
\mathbb E
\|\nabla\Phi(x_\tau)\|_{2,1}
\le
\frac{
\Delta_\Phi
+
LC_1\sum_{t=1}^T\eta_t\delta_t
+
\frac{LC_{2,\nu}}{2}\sum_{t=1}^T\eta_t^2
}{
\sum_{t=1}^T\eta_t
}.
\end{equation}

\end{proof}

\end{document}